\documentclass[12pt, english]{article}

\usepackage{amsmath}
\usepackage{amsthm}
\usepackage{thm-restate}
\usepackage{amssymb}
\usepackage{lmodern}
\usepackage[LGR,T1]{fontenc}
\usepackage{textcomp}
\usepackage{xcolor}
\usepackage{babel}
\usepackage{stackrel}
\usepackage[a4paper]{geometry}
\usepackage[authoryear]{natbib}
\PassOptionsToPackage{normalem}{ulem}
\usepackage{ulem}
\usepackage[]
 {hyperref}
\usepackage{tikz-cd}

\usepackage{tikz}
\usepackage{pgfplots}
\pgfplotsset{compat=1.18}
\usepackage{bm}

\theoremstyle{plain}
\newtheorem{prop}{Proposition}
\newtheorem{thm}{\protect\theoremname}
\theoremstyle{definition}
\newtheorem{example}[thm]{\protect\examplename}

\providecommand{\examplename}{Example}
\providecommand{\theoremname}{Theorem}
\newtheorem{defn}{Definition}

\def\pp{p}				
\def\bP{\mathbf{P}}		
\def\mP{\mathbb{P}}		
\def\mm{\mu}			
\def\bM{\mathbf{M}}		
\def\mM{\mathbb{M}}	    
\def\SA{A}              

\begin{document}

\title{Trajectory of Probabilities, Probability on Trajectories, and the Stochastic--Quantum Correspondence}
\author{\textit{Gy{\H{o}}z{\H{o}} Egri\textsuperscript{\,1}, M\'arton G\"om\"ori\textsuperscript{\,2,3,*}, Bal\'azs Gyenis\textsuperscript{\,3}, G\'abor Hofer-Szab\'o\textsuperscript{\,3}}\\[6pt]
\small \textit{\textsuperscript{1}Faulhorn Labs, Budapest}\\
\small \textit{\textsuperscript{2}ELTE Faculty of Humanities, Institute of Philosophy, Budapest}\\
\small \textit{\textsuperscript{3}ELTE Research Centre for the Humanities, Institute of Philosophy, Budapest}\\
\small \textit{\textsuperscript{*}Corresponding author: \href{mailto:gomori.marton@btk.elte.hu}{gomori.marton@btk.elte.hu}}}

\maketitle

\begin{abstract}
The probabilistic description of the time evolution of a physical system can take two conceptually distinct forms: a trajectory of probabilities, which specifies how probabilities evolve over time, and a probability on trajectories, which assigns probabilities to possible histories. A lack of a clear distinction between these two probabilistic descriptions has given rise to a number of conceptual difficulties, particularly in recent analyses of stochastic--quantum correspondence. This paper provides a systematic account of their relationship. We define probability dynamics and stochastic process families together with a precise notion of implementation that connects the two descriptions. We show that implementations are generically non-unique, that every probability dynamics admits a Markovian implementation, and characterize when non-Markovian implementations are possible. We expose fallacies in common arguments for the linearity of probability dynamics based on the law of total probability and clarify the proper interpretation of ``transition matrices'' by distinguishing dynamics-level maps from the conditional probability matrices of implementing processes. We further introduce decomposability as the appropriate general notion of stepwise evolution for (possibly nonlinear) probability dynamics, relate it to divisibility in the linear case---showing that the two can come apart---and disentangle both notions from Markovianity and time-homogeneity. Finally, we connect these results to what we call statistical dynamics, in which linearity is indeed physically motivated, and contrast the framework with quantum mechanics.

\end{abstract}

\section{Introduction}

The time evolution of a particular physical system can be probabilistically described in two conceptually distinct ways: either by a {\em trajectory of probabilities}, or by a {\em probability on trajectories}. This paper analyzes the relationship between the two probabilistic descriptions in general terms. Although the mathematical distinction between these two probabilistic descriptions is straightforward, their insufficient separation has led to a number of conceptual difficulties in characterizing quantum mechanical phenomena, including the approach taken by \cite{barandes2025,gillespie1994,gillespie2000,gillespie2001,gullo2014,vacchini2011,vacchini2012,smirne2013,rivas2014,breuer2016,li2018,li2019,spekkens2019,canturk2024}. 

We begin by illustrating the distinction using a simple example and highlighting ten observations of this paper concerning conceptual ambiguities found in the literature.
\par\vspace{0.5\baselineskip}

\noindent 1. Suppose that we toss a coin three times. The coin has a hidden built-in weight which can change position between the different tosses and consequently change the coin's bias; we want to describe the time evolution of the probabilities with which the coin may land on heads or tails at the three moments of time. 

In the first probabilistic description, we fix a trajectory of probabilities: a {\em probability vector trajectory} for the two configurations $C \doteq \{ H, T \}$ is defined as
\begin{equation}
\vec{\pp}(t) \doteq (\pp_H(t), \pp_T(t))
\end{equation}
for $t=1,2,3$, where the three numbers $0 \leq \pp_H(t) \leq 1$ are interpreted as the probability that the outcome of the $t$th coin toss is heads, and where $\pp_T(t) \doteq 1 - \pp_H(t)$. 

In the second probabilistic description, we fix a probability on trajectories. First, we define the event $H(t)$, with the interpretation that the outcome of the $t$th coin toss is heads, as the set of all trajectories in the configuration trajectory sample space $$C^3 = \{ HHH, HHT, HTH, HTT, THH, THT, TTH, TTT \}$$ whose $t$th outcome is $H$. Thus, 
\begin{eqnarray}\label{eq:headsdef1}
	H(1)						& \doteq &  \{ HHH, HHT, HTH, HTT  \} \\
	H(2) 						& \doteq & \{ HHH, HHT, THH, THT \} \\
	H(3) 						& \doteq & \{ HHH, HTH, THH, TTH \} \label{eq:headsdef3}
\end{eqnarray}
(and {\em mutatis mutandis} for the events $T(t)$ stating that the outcome of the $t$th coin toss is tails). It follows directly that
\begin{eqnarray}
	H(1) \wedge H(2)			& = & \{ HHH, HHT \} \label{eq:headsdef4} \\
	H(1) \wedge H(3)			& = & \{ HHH, HTH \} \\
	H(2) \wedge H(3)			& = & \{ HHH, THH \} \\
	H(1) \wedge H(2) \wedge H(3)	& = & \{ HHH \} \label{eq:headsdef7}
\end{eqnarray}
Second, we define a probability measure $\mm$ on $C^3$ as sample space. This can be achieved by fixing seven parameters; for instance, by assigning values to the probabilities of the seven events in (\ref{eq:headsdef1})-(\ref{eq:headsdef7}). A so-defined $\mm$ probability on trajectories naturally corresponds to a so-called {\em canonical stochastic process} whose random variable for time $t$ assigns to a trajectory in $C^3$ the $t$th outcome of the trajectory. Since different canonical stochastic processes on $C^3$ are only distinguished by their probability measures (see Section \ref{sec:twodescriptions} for details), when this does not lead to ambiguity, we will use $\mm$ to denote both the probability measure on the configuration trajectories and the corresponding canonical stochastic process.

The two probabilistic descriptions are closely related. For $t=1,2,3$, the vectors 
\begin{equation}
	\vec{\mm}(t) \doteq \left( \mm(H(t)), \mm(T(t)) \right)
\end{equation}
which are naturally defined from $\mm$, form a probability vector trajectory in the sense described above. In turn, we say that the stochastic process $\mm$ {\em implements} the probability vector trajectory $\vec{\pp}(t)$ if
\begin{equation}\label{eq:impl}
	\vec{\mm}(t) = \vec{\pp}(t)
\end{equation}
holds for all $t=1,2,3$. 

Section \ref{sec:twodescriptions} develops the distinction between these two probabilistic descriptions and the notion of implementation in more general terms. In our coin toss example, the probability vector trajectory $\vec{\pp}(t)$ and the implementation condition (\ref{eq:impl}) fix the $\mm$-probabilities of the three events (\ref{eq:headsdef1})-(\ref{eq:headsdef3}), but these are only three out of the seven parameters needed to define $\mm$. Without fixing four additional parameters, such as the $\mm$-probabilities of the four conjunctions (\ref{eq:headsdef4})-(\ref{eq:headsdef7}), an implementation of $\vec{\pp}(t)$ has four free parameters left. This leads to our

\par\vspace{0.5\baselineskip}
\noindent\textit{First observation: the implementation of a probability vector trajectory by a stochastic process is not unique.}
\par\vspace{0.5\baselineskip}

\noindent 2. A canonical stochastic process has more mathematical structure than a probability vector trajectory. The additional structure is superfluous if one only wants to describe probabilities of events at given moments of time. However, the additional structure is necessary if one also wants to describe the joint probabilities of events at {\em different} moments of time. If we merely fix a probability vector trajectory, it is {\em meaningless} to ask what is the probability that the third and the second coin tosses are both heads, or to ask what is the conditional probability that the third coin toss is heads given that the second coin toss is heads, etc. This is so since there is no unique extension of the three probability measures generated by $\vec{\pp}(1)$, $\vec{\pp}(2)$, and $\vec{\pp}(3)$ on three copies of the configuration space $C$ to a probability measure on the trajectory sample space $C^3$. Thus, merely fixing a probability vector trajectory $\vec{\pp}(t)$ leaves expressions of the form ``$\pp(H(3) \wedge H(2))$'' or ``$\pp(H(3) | H(2))$'' ill-defined. To make claims about joint probabilities of events at different moments of time we need to define a canonical stochastic process, which both makes it meaningful to talk about the conjunction of events at different times and fixes the probabilities of these conjunctions. We may sum up these remarks as our

\par\vspace{0.5\baselineskip}
\noindent\textit{Second observation: a probability vector trajectory says nothing about the probabilities of conjunctions of events at different moments of time. A canonical stochastic process does.}
\par\vspace{0.5\baselineskip}

\noindent 3. As a further consequence, the question whether a probability vector trajectory $\vec{\pp}(t)$ satisfies the Markov condition is ill-defined, since the Markov condition is expressed by connecting events at different moments of time. However, it is meaningful to ask, for instance, whether for the stochastic process $\mm$ the Markov condition
\begin{equation}\label{eq:mark}
	\mm\left( H(3) \mid H(2) \wedge H(1) \right) = \mm\left( H(3) \mid H(2) \right)
\end{equation}
holds. Whether equation (\ref{eq:mark}) holds is not determined, in general, by the implementation conditions (\ref{eq:impl}) alone, and hence the question of Markovianity is in general independent of how the probability vector trajectory $\vec{\pp}(t)$ evolves. We give a more precise characterization of this claim in Section \ref{sec:relating-Markovianity}; for now, it can be succinctly stated as our

\par\vspace{0.5\baselineskip}
\noindent\textit{Third observation: a probability vector trajectory says nothing about Markovianity. In general, a probability vector trajectory can be implemented both by a stochastic process that is Markovian and by a stochastic process that is non-Markovian.}
\par\vspace{0.5\baselineskip}

\noindent 4. It is natural to generalize from a single probability vector trajectory $\vec{\pp}(t)$ to a set of probability vector trajectories which satisfy a law. When these probability vector trajectories are in one-to-one correspondence with initial conditions we obtain what may be called {\em probability dynamics}. To illustrate the idea using our coin toss example, suppose that the built-in weight in the coin moves in such a way that for an initial coin bias $0 \leq r \leq 1$ for heads the bias changes in time as
\begin{eqnarray}
	\pp_H(1) & = & r					\\
	\pp_H(2) & = &f(r)  \label{eq:ph2f}	\\
	\pp_H(3) & = & r
\end{eqnarray}
with a function $f: [0,1]\rightarrow[0,1]$. In other words, we defined a probability dynamics $\bP(t, \vec{\pp}_1) = \bP_t(\vec{\pp}_1)$ with initial value $\vec{\pp}_1 = (r, 1-r)$ such that 
\begin{eqnarray}\label{eq_probdyn1}
	\bP(1, \vec{\pp}_1) & \doteq & \vec{\pp}(1) = (r, 1-r)					\\
	\bP(2, \vec{\pp}_1) & \doteq & \vec{\pp}(2) = (f(r), 1-f(r))	\\
	\bP(3, \vec{\pp}_1) & \doteq & \vec{\pp}(3) = (r, 1-r)  \label{eq_probdyn3}
\end{eqnarray}

Since a probability dynamics is a set of probability vector trajectories, and since a single probability vector trajectory is implemented by a single stochastic process, a probability dynamics is implemented by a {\em set} of canonical stochastic processes. Continuing our coin-toss example, for every $0 \leq r \leq 1$ define the canonical stochastic processes $\mm_{r}$ by fixing
\begin{eqnarray}
			\mm_{r}(HHH) & = & r \, f(r) \, r	\label{eq:mur1}	\\
			\mm_{r}(HHT)  & = & r \, f(r) \, (1-r)		\\
			\mm_{r}(HTH) & =  & r \, (1-f(r)) \, r			\\
			\mm_{r}(HTT)  & =  & r \, (1-f(r)) \, (1-r)				\\
			\mm_{r}(THH)  & = & (1-r) \, f(r) \, r			\\
			\mm_{r}(THT)  & = & (1-r) \, f(r) \, (1-r)		\\
			\mm_{r}(TTH)  & = & (1-r) \, (1-f(r)) \, r			\\
			\mm_{r}(TTT)  & = & (1-r) \, (1-f(r)) \, (1-r)	\label{eq:mur8}	
\end{eqnarray}
It is easy to verify that for every $0 \leq r^* \leq 1$ the so-defined canonical stochastic process $\mm_{r^*}$ implements the probability vector trajectory $\bP(t, (r^*, 1-r^*))$. In other words, we may say that the {\em stochastic process family} $\bM$, understood as a set of canonical stochastic processes $\{ \mm_{r} \}_{r \in [0,1]}$ parametrized by $r$, implements probability dynamics $\bP$, understood as a set of probability vector trajectories parametrized by $r$. With Section \ref{sec:twodescriptions} elaborating on details and important qualifications, we arrive at our

\par\vspace{0.5\baselineskip}
\noindent\textit{Fourth observation: while a single probability vector trajectory is implemented by a single stochastic process, one has to be aware that a probability dynamics, which consists of a set of probability vector trajectories, is implemented by a set of canonical stochastic processes, each of which assigns different probabilities to the possible configuration trajectories.}
\par\vspace{0.5\baselineskip}

\noindent 5. Overlooking the fourth observation has contributed to a number of problematic statements in the literature. To see this, define the stochastic matrices $\mM_{r} \left( t \leftarrow t' \right)$ using the two-time conditional probabilities as
\begin{align}
	\mM_{r} \left( t \leftarrow t' \right) & \doteq \left(
			\begin{array}{cc}
				\mm_{r} \left( H(t) | H(t') \right)	&	\mm_{r} \left( H(t) | T(t') \right)		\\
				\mm_{r} \left( T(t) | H(t') \right)		&	\mm_{r} \left( T(t) | T(t') \right)		
			\end{array}
		\right) \label{eq_delta_re}
\end{align}
for $t,t' = 1,2, 3$. Then, by the law of total probability, one immediately obtains
\begin{equation}
	\vec{\mm}_{r}(t) = \mM_{r} \left(t \leftarrow t' \right) \cdot \vec{\mm}_{r}(t')
\end{equation}
Since the stochastic process family $\bM$ implements probability dynamics $\bP$, it then follows that
\begin{eqnarray}\label{eq:delta_p2p1}
	\vec{\pp}(2) & = & \mM_{r}\left(2 \leftarrow 1 \right) \cdot \vec{\pp}(1)	\\
	\vec{\pp}(3) & = & \mM_{r}\left(3 \leftarrow 2 \right) \cdot \vec{\pp}(2) \label{eq:delta_p3p2}	\\
	\vec{\pp}(3) & = & \mM_{r}\left(3 \leftarrow 1 \right) \cdot \vec{\pp}(1) \label{eq:delta_p3p1}
\end{eqnarray}

Equations (\ref{eq:delta_p2p1})-(\ref{eq:delta_p3p1}) may give the impression that the probability dynamics $\bP$ must be linear in the sense of preserving convex combination of probability vectors. However, this impression is misleading. In our example, whether probability dynamics $\bP$ is linear depends on how the function $f$ in (\ref{eq:ph2f}) is defined. For example, if $f(r)=r$, then $\bP$ is linear (the identity); however, if $f(r) = r^2$, then $\bP$ is not linear, since
\begin{align}
\bP_2(r,1-r) = (r^2, 1-r^2) \neq (r,1-r) = r\,\bP_2(1,0) + (1-r)\,\bP_2(0,1)
\end{align}
when $r \neq 0,1$. 
\noindent In fact, equations (\ref{eq:delta_p2p1})-(\ref{eq:delta_p3p1}) hold regardless of whether the probability dynamics is linear or not linear. The wrong impression arises from ignoring the dependence of $\mM_{r} \left( t \leftarrow t' \right)$ on the initial value $r$, that is, from ignoring our fourth observation above. In general, different $r$ initial values correspond to different probability measures $\mm_r$ and henceforth to different matrices $\mM_{r} \left( t \leftarrow t' \right)$ (see equation (\ref{eq_delta_re})), and thus a single matrix $\mM_{r^*} \left( t \leftarrow t' \right)$, for a single fixed value $r = r^*$, cannot be viewed as a generator of probability dynamics $\bP$, which is a function of non-fixed variable $r$ (see Section \ref{sec:Linearity} for more details). In sum,

\par\vspace{0.5\baselineskip}
\noindent\textit{Fifth observation: the fact that each stochastic process of a stochastic process family that implements a probability dynamics trivially satisfies its own temporal law of total probability has nothing to do with the question whether the probability dynamics is linear. In particular, this fact cannot be relied upon to motivate the linearity of probability dynamics.}
\par\vspace{0.5\baselineskip}

\noindent
6. Recall that $\bP_2$ tells how an initial probability vector $\vec{\pp}(1) = (r, 1-r)$ evolves from time 1 to time 2, and $\bP_3$ tells how $\vec{\pp}(1)$ evolves from time 1 to time 3. It is natural to ask whether this time evolution is {\em decomposable} in the sense that there exists, on the range of $\bP_2$, a function $\bP_{3 \leftarrow 2}$ such that 
\begin{equation}
    \bP_3 = \bP_{3 \leftarrow 2} \circ \bP_2
\end{equation}
\noindent When $\bP$ is linear---say, when $f(r) = r$ in (\ref{eq:ph2f}) in our example---then there exist $r$-independent stochastic matrices $\mP(2)$ and $\mP(3)$ such that
\begin{align}
	\vec{\pp}(2) & = \mP(2) \cdot \vec{\pp}(1)  \label{eq:linearpintro1} \\
	\vec{\pp}(3) & = \mP(3) \cdot \vec{\pp}(1)   \label{eq:linearpintro2}
\end{align}
for all $0 \leq r \leq 1$. By simple application of the definition, a linear $\bP$ is decomposable if there exists a function $\bP_{3 \leftarrow 2}$ such that
\begin{equation}
    \vec{\pp}(3) = \mP(3) \cdot \vec{\pp}(1) = \bP_{3 \leftarrow 2} \left( \mP(2) \cdot \vec{\pp}(1) \right)
\end{equation}
for all $0 \leq r \leq 1$. We obtain a {\em special case} of a decomposable linear $\bP$ when there exists a {\em stochastic} matrix $\mP(3 \leftarrow 2)$ such that 
\begin{equation}\label{eq:divisibilityeq1}
    \vec{\pp}(3) = \mP(3) \cdot \vec{\pp}(1) = \mP(3 \leftarrow 2) \mP(2) \cdot \vec{\pp}(1)
\end{equation}
for all $0 \leq r \leq 1$. This special case of decomposability is often called {\em divisibility} (see Section \ref{sec:Decomposability} for qualifications).

Now, substituting (\ref{eq:delta_p2p1}) into (\ref{eq:delta_p3p2}) results in
\begin{equation}\label{eq:delta_p3p2p1}
	\vec{\pp}(3) = \mM_{r}\left(3 \leftarrow 2 \right) \mM_{r}\left(2 \leftarrow 1 \right) \cdot \vec{\pp}(1)
\end{equation}
and a superficial comparison with \eqref{eq:divisibilityeq1} might give the impression that divisibility always holds. However, this impression would also be wrong for the same reason as before: it ignores the dependence of $\mM_{r} \left( t \leftarrow t' \right)$ on the initial value $\vec{\pp}(1) = (r, 1-r)$.

In Section \ref{sec:Markovianity} we address the substantial ambiguity in the literature regarding the different concepts of stepwise evolution of probability dynamics, clarify the notions of divisibility and decomposability, and analyze their relationships with time-homogeneity and Markovianity. One of our conceptual findings can be summarized as follows:

\par\vspace{0.5\baselineskip}
\noindent\textit{Sixth observation: the appropriate mathematical notion that captures the intuitive idea of a stepwise evolution of probability dynamics is decomposability, not divisibility. Since every divisible probability dynamics is decomposable, but the converse does not hold even for linear probability dynamics, in cases where the probability dynamics is decomposable but not divisible, it is not the intuitive notion of stepwise evolution that fails.}
\par\vspace{0.5\baselineskip}

\noindent 7. Our third observation entails that Markovianity is a property that may or may not characterize each stochastic process of a stochastic process family which implements a probability dynamics. Since linearity, decomposability, and divisibility are---in contrast to Markovianity---properties of probability dynamics, it is easy to see (see Section \ref{sec:Markovianity} for details) that 

\par\vspace{0.5\baselineskip}
\noindent\textit{Seventh observation: Markovianity or non-Markovianity of stochastic processes of a stochastic process family that implements a probability dynamics has nothing to do with the question whether the probability dynamics is linear, decomposable, or divisible.}
\par\vspace{0.5\baselineskip}

\noindent 8. Suppose again that probability dynamics $\bP$ is linear, hence there exist $r$-independent stochastic matrices $\mP(t)$ such that
\begin{equation}
	\vec{\pp}(t) = \bP \left( t, \vec{\pp}(1) \right) = \mP(t) \cdot \vec{\pp}(1)
\end{equation}
hold for all $0 \leq r \leq 1$, $\vec{\pp}(1) = (r, 1-r)$, and for all $t =1,2,3$. Every such linear probability dynamics can be implemented by a stochastic process family $\bM$ so that
\begin{equation}\label{eq:linearmatrixeq}
	\mP(t) = \mM_r (t \leftarrow 1)
\end{equation}
for all $0 < r < 1$ and for all $t =1,2,3$ (see Section \ref{sec:Transition}). Since the matrices $\mP(t)$ are independent of $r$, equation (\ref{eq:linearmatrixeq}) can give the impression that at least a linear probability dynamics can be meaningfully generated from a single canonical stochastic process. However, this impression is again wrong: although a linear probability dynamics may indeed be implemented by a stochastic process family such that the conditional probability matrices $\mM_r (t \leftarrow 1)$ are independent of $r$, in the nontrivial case (when the $\bP(t, (r, 1-r))$ probability vector trajectories differ for different values of $r$) different canonical stochastic processes are still needed to implement the different probability vector trajectories. This is so even though all of the canonical stochastic processes in the family are such that they entail the same conditional probability matrices, that is, the stochastic process family may be called {\em transition-constant} (for a precise definition, see Section \ref{sec:Transition})! Thus, we still need a stochastic process \emph{family} to implement linear probability dynamics, not just a \emph{single} canonical stochastic process.

To further drive home the point, we note that the interpretations of the matrices $\mP(t)$ and $\mM_{r^*} (t \leftarrow 1)$ are quite different. $\mP(2)$ answers the question what is the probability that the second coin toss is heads if the probability of the first coin toss being heads is $r$ (for all $0 \leq r \leq 1$), but $\mP(2)$ is silent on the question what is the probability that the second coin toss is heads if the first coin toss is heads. In contrast, $\mM_{r^*} (2 \leftarrow 1)$ answers the question what is the probability that the second coin toss is heads if the first coin toss is heads, but $\mM_{r^*} (2 \leftarrow 1)$ is silent on the question what is the probability that the second coin toss is heads if the probability of the first coin toss being heads is $r$ (except when $r = r^*$). For more details, see Section \ref{sec:Transition}. In summary,

\par\vspace{0.5\baselineskip}
\noindent\textit{Eighth observation: constructing a probability dynamics from a single canonical stochastic process is a category mistake even if linearity of probability dynamics could be motivated on appropriate physical grounds.}
\par\vspace{0.5\baselineskip}

\noindent 9. Appropriate physical--interpretational assumptions can justify linearity of probability dynamics. Section \ref{sec:statistical-dynamics} introduces and analyzes the relationship of three varieties of what we call statistical dynamics, and shows that deterministic, stochastic, and initially independent system--ancilla statistical dynamics each entail a probability dynamics that is linear. Beyond ensuring linearity, the three varieties possess additional appealing features; for instance, in the simplest case of deterministic statistical dynamics, decomposability, divisibility, and Markovianity are all mathematically equivalent properties. However, in all three varieties, linearity of the ensuing probability dynamics depends crucially on the physical interpretation that instantaneous probability vectors represent frequency distributions of configurations in an ensemble of systems, each system evolving independently according to a deterministic or stochastic law of evolution. Thus,

\par\vspace{0.5\baselineskip}
\noindent\textit{Ninth observation: although linearity of probability dynamics cannot be assumed purely on mathematical grounds, it can be derived from certain physical--interpretational assumptions, notably from statistical mixing. Consequently, to justify assuming that a certain probability dynamics is linear, one must always clearly and explicitly state and defend the physical--interpretational assumption from which linearity is supposed to ensue.}
\par\vspace{0.5\baselineskip}

\noindent 10. Our ninth observation becomes especially pertinent in the case of quantum mechanics, the probabilities of which, according to the prevailing view, do not admit a statistical interpretation. In Section \ref{sec:quantum-dynamics} we emphasize the difficulty of obtaining a well-defined probability dynamics from Born probabilities alone, the generic nonlinearity of quantum probability evolution, and the limitations this imposes on proposed stochastic--quantum correspondences. One of the upshots of our analysis can be put as follows:

\par\vspace{0.5\baselineskip}
\noindent\textit{Tenth observation: linearity of quantum probability evolution cannot be inferred from statistical mixing. Moreover, linearity of quantum probability evolution (in contrast with linearity of density matrix evolution) would be in tension with quantum mechanics itself.}
\par\vspace{0.5\baselineskip}

\noindent In sum, the paper is structured as follows. In Section \ref{sec:twodescriptions} we introduce the conceptual distinction between trajectories of probabilities and probabilities on trajectories by providing precise definitions of probability vector trajectories and probability dynamics, canonical stochastic processes, stochastic process families, and the notion of implementation that relates them. Section \ref{sec:Linearity} analyzes linearity at the level of probability dynamics and shows why the temporal law of total probability, when applied within individual implementing stochastic processes, does not by itself justify linear evolution of probabilities; crucially, the dependence of multi-time conditional probabilities on initial conditions cannot in general be ignored. Section \ref{sec:Transition} then distinguishes the matrices $\mP(t)$ representing a linear probability dynamics from the conditional-probability matrices $\mM_{\vec{\pp}_{0}}(t)$ associated with particular implementing processes, introduces the notion of a transition-constant stochastic process family, and proves that transition-constancy suffices to recover linearity together with the usual identification of ``transition matrices'' with conditional probabilities (while also showing that this implication is not reversible in general). In Section \ref{sec:Markovianity} we disentangle the literature's notion of ``divisibility'' by separating it into the distinct concepts of decomposability and divisibility for probability dynamics, and we relate these notions to Markovianity as a property of stochastic processes; among other results, we show that every probability dynamics admits a Markovian implementation, that implementations are highly non-unique, and that non-Markovian implementations are also generically available. We also provide explicit examples establishing that decomposability does not imply divisibility, even in the linear case. Section \ref{sec:statistical-dynamics} examines a physically motivated special case---statistical dynamics---in which probability vectors represent ensemble frequencies, thereby providing clear interpretational conditions under which linearity is warranted and under which decomposability, divisibility, and Markovianity coincide or come apart. Finally, Section \ref{sec:quantum-dynamics} contrasts the framework with quantum mechanics, emphasizing the difficulty of defining a probability dynamics from Born probabilities alone, the generic nonlinearity of quantum probability evolution, and the resulting constraints on proposed stochastic--quantum correspondences. The Appendices collect proofs of key propositions (Appendix \ref{appendix:proofs}), discuss an alternative notion of implementation (Appendix \ref{appendix:innerimplementation}), and provide a detailed, critical analysis of the interference argument of \cite{barandes2025} (Appendix \ref{appendix:Barandes-interference}).

\section{Trajectories of Probabilities vs. Probabilities on Trajectories}
\label{sec:twodescriptions}

We consider a physical system with a finite set of possible configurations
$C=\left\{ 1,2,\ldots,N\right\} $, labelled by positive integers.
A vector $\vec{\pp}=\left(p_{1},p_{2}\ldots,p_{N}\right)\in\left[0,1\right]^{N}$
with
\begin{equation}
\sum_{i=1}^{N}p_{i}=1
\end{equation}
is called an $N$-dimensional \emph{probability vector}. We denote
the set of all such vectors by $\mathcal{S}_{N}$. Obviously, a probability
distribution on $C$ is uniquely characterized by the $N$-dimensional
probability vector whose $i$th entry specifies the probability of
the $i$th configuration. Hence, the set of all probability distributions
on $C$ can be identified with $\mathcal{S}_{N}$.

Let $T\subseteq\mathbb{R}$, with $0\in T$, be a set of time indices.
We now formulate two distinct ways to describe the temporal evolution
of the system in probabilistic terms. In comparing these two descriptions,
the sets $C$ and $T$ are regarded as fixed.

\begin{defn}
\label{def:prob-dyn} A\emph{ probability vector trajectory} is a
map $T\rightarrow\mathcal{S}_{N}$, and a\emph{ probability dynamics}
is a map 
\begin{equation}
\bP:T\times\mathcal{S}_{N}\rightarrow\mathcal{S}_{N}
\end{equation}
such that $\bP\left(0,\cdot\right):\mathcal{S}_{N}\rightarrow\mathcal{S}_{N}$
is the identity function. The \emph{solution} of probability dynamics
$\bP$ with \emph{initial condition} $\vec{\pp}_{0}\in\mathcal{S}_{N}$
is the probability vector trajectory 
\begin{equation}
T\rightarrow\mathcal{S}_{N},\,\,\,\,\,t\mapsto\bP\left(t,\vec{\pp}_{0}\right)\label{eq:solutionofprobdyn}
\end{equation}
We will use the notation $\vec{\pp}\left(t\right)$ for a generic
probability vector trajectory, and $\bP_{t}$ for the map $\bP\left(t,\cdot\right):\mathcal{S}_{N}\rightarrow\mathcal{S}_{N}$. 
\end{defn}

To introduce the second conception, recall that
the standard notion of a stochastic process consists of a probability
space $\left(\Omega,\mathcal{F},\mu\right)$ together with a time-parametrized
family of random variables $\left\{ X_{t}:\Omega\rightarrow C\right\} _{t\in T}$.
In the special case where the sample space $\Omega$ is taken to be
the set of all temporal trajectories through configuration space,
the stochastic process is called \emph{canonical} (\cite{Rogers_Williams2000},
p.~122; \cite{medvegyev2007}, pp.~2--3). More precisely, for a canonical
stochastic process
\begin{itemize}
\item $\Omega$ is the set of all functions $T\rightarrow C$,
\item $\mathcal{F}$ is the cylinder $\sigma$-algebra of $\Omega$,\footnote{That is, $\mathcal{F}$ is the natural $\sigma$-algebra on $\Omega$,
generated by the so-called cylinder sets---subsets of $\Omega$ of
the form 
\begin{equation}
\left\{ \omega\in\Omega~|~\omega(t_{1})\in C_{1},\omega(t_{2})\in C_{2},\dots,\omega(t_{n})\in C_{n}\right\} \label{eq:cylinder-algebra}
\end{equation}
where $t_{1},t_{2},\dots,t_{n}\in T$, $C_{1},C_{2},...,C_{n}\subseteq C$,
and $n\in\mathbb{N}$.} 
\item $ X_{t}:\Omega\rightarrow C$  is defined as $X_{t}\left(\omega\right)\doteq\omega\left(t\right)$.
\end{itemize}
In what follows, $\Omega$, $\mathcal{F}$ and $X_{t}$ will always
denote these canonical objects. Since for a given choice of sets $C$
and $T$ these three components are fixed, we can identify a canonical
stochastic process as follows.

\begin{defn}
\label{def:canonical-stoch-process}A \emph{canonical stochastic
process} is a probability measure $\mu$ on the measurable space $\left(\Omega,\mathcal{F}\right)$.
\end{defn}

\noindent For our present purposes it will be sufficient to consider only \emph{canonical}
stochastic processes.\footnote{Every stochastic process naturally gives rise to a canonical
one, but not vice versa. A generic stochastic process therefore contains additional structure that the canonical one does not.  To see this, consider a generic stochastic process consisting
of probability space $\left(\tilde{\Omega},\tilde{\mathcal{F}},\tilde{\mu}\right)$
together with a family of random variables $\left\{ \tilde{X}_{t}:\tilde{\Omega}\rightarrow C\right\} _{t\in T}$.
This family defines a single (measurable) function
\begin{equation}
\tilde{\mathcal{X}}:\tilde{\Omega}\rightarrow\Omega,\,\,\,\,\,\left(\tilde{\mathcal{X}}\left(\tilde{\omega}\right)\right)\left(t\right)\doteq\tilde{X}_{t}\left(\tilde{\omega}\right)
\end{equation}
which simply repackages the family $\left\{ \tilde{X}_{t}\right\} $.
What $\tilde{\mathcal{X}}$ does is specifying how each random event $\tilde{\omega}\in\tilde{\Omega}$
determines a full temporal trajectory of the system through its configuration
space. This map then induces a probability distribution over the set
$\Omega$ of possible trajectories---thus defining a canonical stochastic
process---by setting $\mu\doteq\tilde{\mu}\circ\tilde{\mathcal{X}}^{-1}$,
where $\tilde{\mathcal{X}}^{-1}$ denotes the preimage map associated with $\tilde{\mathcal{X}}$. From the canonical process $\mu$ alone, however, one cannot recover the original stochastic process from which it was constructed.
For details, see e.g. \cite{Rogers_Williams2000}, p.~122; \cite{medvegyev2007},
pp.~2--3. \\ Appendix \ref{appendix:innerimplementation} briefly addresses the representation of information encoded by a probability dynamics with a non-canonical stochastic process.} Accordingly, we will omit the adjective ``canonical'' and simply
refer to stochastic processes in the sense of Definition~\ref{def:canonical-stoch-process}.

While stochastic processes are well known and extensively studied
in the mathematical literature, concepts akin to what we have defined
as probability dynamics have only recently appeared in the physics
literature. In quantum foundations, one often encounters similar notions
under labels such as ``dynamical map'' (\cite{breuer2016}), ``dynamical
semigroup'' (\cite{breuerpetruccione}), ``stochastic map'' (\cite{spekkens2019}),
``transition map'' (\cite{barandes2025}), or ``propagator'' (\cite{vacchini2012}).
However, the standard characterizations of these notions frequently
conflate features that belong separately to probability dynamics and
to stochastic processes. To clarify this, in the remainder of this
section we begin a broad comparison of these two conceptions.

One reason it may be tempting to associate a probability dynamics
with a stochastic process is that both express law-like relationships.
A stochastic process is a specification of a probability distribution
on the set of possible temporal trajectories through configuration
space. In this sense, a stochastic process provides what can be interpreted
as a \emph{probabilistic law about the temporal evolution of configurations}.
By contrast, a probability dynamics can be thought to specify admissible
probability vector trajectories corresponding to every possible initial
condition. In this sense, a probability dynamics can be interpreted
as a \emph{deterministic law about the temporal evolution of probabilities
of configurations}. Put differently: what a stochastic process supplies
is a single time-independent distribution over time-dependent
configurations (probability on trajectories); whereas what a probability
dynamics supplies are multiple time-dependent distributions over time-independent configurations (trajectories
of probabilities).

Crucially, the two structures contain different types of probabilistic
information. On the one hand, in the context of a stochastic process one can speak
about the probabilities of conjunctions of events across different
moments of time. Indeed, define 
\begin{equation}
E_{i}\left(t\right)\doteq\left\{ \omega\in\Omega~|~X_{t}\left(\omega\right)=i\right\} \label{eq:event}
\end{equation}
as the event that at time $t$ the system occupies the $i$th configuration.
Then, one can express the probabilities of conjunctions of such events
as 
\begin{gather}
\mm\left(E_{i_{1}}\left(t_{1}\right)\wedge E_{i_{2}}\left(t_{2}\right)\wedge...\wedge E_{i_{n}}\left(t_{n}\right)\right)\\
=\mm\left(\left\{ \omega\in\Omega~|~X_{t_{1}}\left(\omega\right)=i_{1},X_{t_{2}}\left(\omega\right)=i_{2},\dots,X_{t_{n}}\left(\omega\right)=i_{n}\right\} \right)\label{eq:prob-conjunction-2}
\end{gather}
By contrast, in the context of a probability dynamics no probability
is assigned to such conjunctions. This is because the sample space
$C$, on which the probability distributions $\vec{\pp}\left(t\right)$
are defined, only includes the timeless configurations $1,2,\dots,N$;
it includes no events of the form ``a configuration occurs at a particular
time.'' In other words, in the context of a probability dynamics,
no probability is assigned to temporal trajectories in $C$. Yet it is precisely the assignment of probabilities to temporal trajectories that allows one to determine the probabilities of conjunctions of events at different times, in the manner of \eqref{eq:prob-conjunction-2}.

An important consequence of this is the following. In the language
of a stochastic process we can define conditional probabilities across
different moments of time using Bayes' rule.\footnote{We assume, throughout the paper, that conditional probabilities are defined by Bayes' rule. In principle, one could introduce both conditional and unconditional probabilities as independent, primitive notions, but without connecting conditional probabilities with unconditional probabilities via Bayes' rule it is not possible to derive either the law of total probability or Bayes' theorem (however, the argument for linearity that this paper challenges (see Section \ref{sec:Linearity}) invokes the law of total probability as a mathematical theorem.)} For example: 
\begin{equation}
\mm\left(E_{i}\left(t\right)|E_{j}\left(0\right)\right)\doteq\frac{\mm\left(E_{i}\left(t\right)\wedge E_{j}\left(0\right)\right)}{\mm\left(E_{j}\left(0\right)\right)}\label{eq:bayes-rule}
\end{equation}
However, such multi-time conditional probabilities are not well-defined
in the context of a probability dynamics. This is because the definition
of multi-time conditional probability involves the probability of
a conjunction of events at different times, such as $\mm\left(E_{i}\left(t\right)\wedge E_{j}\left(0\right)\right)$
in \eqref{eq:bayes-rule}. But, as we have seen, this kind of joint
probability is not defined in a probability dynamics.

On the other hand, in the context of a probability dynamics one can
speak about a functional dependence of the probability of a configuration
at a certain moment of time on the probabilities of configurations
at an initial moment of time. Indeed, the $i$th component of $\bP_{t}\left(\vec{\pp}_{0}\right)$
expresses the probability of the $i$th configuration at time $t$
as a function of the initial $\vec{\pp}_{0}$ probabilities of configurations.
By contrast, even though the probability vector $\vec{\mm}(t)$ defined
as 
\begin{align}
\vec{\mm}\left(t\right) & \doteq\left(\mm\left(E_{1}\left(t\right)\right),\mm\left(E_{2}\left(t\right)\right),...,\mm\left(E_{N}\left(t\right)\right)\right)\label{eq:mu-vector}
\end{align}
expresses the probabilities of configurations at time $t$, including
$t=0$, a single stochastic process does not contain information
about the functional dependency of how probabilities $\vec{\mm}(t)$
would differ in case initial probabilities $\vec{\mm}(0)$ were different.

What this means is that if one simultaneously wants to speak about
the functional relationship between probabilities at different times
and multi-time conditional probabilities, one has to relate and integrate
the two probabilistic descriptions. There is a natural way to do this.\footnote{Although the notion of implementation defined here is quite natural,
one may introduce other notions which capture the idea that a stochastic
process can reproduce a probability dynamics; see Appendix \ref{appendix:innerimplementation}
for details.}

\begin{defn}
\label{def:implementation}(i) Let ${\cal M}$ denote the set of
all probability measures on $\left(\Omega,\mathcal{F}\right)$. A
\emph{stochastic process family} is a map
\[
\bM:\mathcal{S}_{N}\rightarrow{\cal M},\,\,\,\,\,\vec{\pp}_{0}\mapsto\mm_{\vec{\pp}_{0}}
\]
such that $\vec{\mm}_{\vec{\pp}_{0}}\left(0\right)=\vec{\pp}_{0}$
for all $\vec{\pp}_{0}\in\mathcal{S}_{N}$.

(ii) We say that a \emph{stochastic process $\mm$ implements a probability
vector trajectory} $\vec{\pp}\left(t\right)$ iff
\begin{equation}
\vec{\mm}\left(t\right)=\vec{\pp}\left(t\right)
\end{equation}
for all $t\in T$. A \emph{stochastic process family $\bM$ implements
probability dynamics $\bP$} iff for each $\vec{\pp}_{0}\in\mathcal{S}_{N}$,
the stochastic process $\mm_{\vec{\pp}_{0}}$ implements the solution
of $\bP$ with initial condition $\vec{\pp}_{0}\in\mathcal{S}_{N}$.
That is,
\begin{equation}
\vec{\mm}_{\vec{\pp}_{0}}\left(t\right)=\bP\left(t,\vec{\pp}_{0}\right)
\end{equation}
for all $t\in T$ and $\vec{\pp}_{0}\in\mathcal{S}_{N}$.
\end{defn}

A stochastic process family is a mathematical structure that extends
the expressive power of both a single stochastic process and the probability
dynamics the family implements. It is richer than a single stochastic process
because it specifies one for every initial probability distribution
$\vec{\pp}_{0}\in\mathcal{S}_{N}$. This allows expressing how probabilities
of configurations at a given time depend on probabilities of configurations
at the initial time. A stochastic process family also enriches the
structure of the probability dynamics it implements, since each solution
of the dynamics---a probability vector trajectory that can only encode
the one-time probabilities $\vec{\pp}\left(t\right)$---is extended
to a full stochastic process, which can specify all multi-time joint
probabilities. This, in turn, allows one to express multi-time conditional
probabilities. Note that the different stochastic processes $\mm_{\vec{\pp}_{0}}$
implementing different solutions of a probability dynamics \emph{need not
assign the same values to multi-time conditional probabilities}. This
observation will play a crucial role in our analysis.

To ensure that the unification of the two probabilistic frameworks
achieved by Definition~\ref{def:implementation} is fully general,
we first need to establish that every probability dynamics can, in
principle, be extended to a stochastic process family. This is guaranteed
by the following corollary of Proposition~\ref{prop:Markovian-implementation}
of Section~\ref{sec:relating-Markovianity}:
\begin{prop}
\label{prop:implementsalways} Every probability dynamics is implemented
by a stochastic process family.
\end{prop}

\section{Understanding Linearity}\label{sec:Linearity}

$\mathcal{S}_{N}$ is the probability simplex in $\mathbb{R}^{N}$,
and the function $\bP_{t}$, for any $t\in T$, maps $\mathcal{S}_{N}$
into itself. Any $\mathcal{S}_{N}\rightarrow\mathcal{S}_{N}$ map
that preserves convex combinations extends uniquely to an $\mathbb{R}^{N}\rightarrow\mathbb{R}^{N}$
linear map.\footnote{Let $f:\mathcal{S}_{N}\rightarrow\mathcal{S}_{N}$ be a convex combination
preserving function. Define the following map: 
\begin{equation}
F:\mathbb{R}^{N}\rightarrow\mathbb{R}^{N},\,\,\,\,\left(x_{1},x_{2},...,x_{N}\right)\mapsto\sum_{i=1}^{N}x_{i}f\left(\vec{e}_{i}\right)
\end{equation}
where $\vec{e}_{i}$ is the $i$th standard basis vector in $\mathbb{R}^{N}$
(and also one of the $N$ vertices of $\mathcal{S}_{N}$). Clearly,
$F$ is linear and extends $f$. Moreover, any linear map that extends
$f$ must be equal to $F$, since a linear map is uniquely determined
by its values on a basis---here fixed by $f\left(\vec{e}_{1}\right),f\left(\vec{e}_{2}\right),...,f\left(\vec{e}_{N}\right)$.\label{fn:extension}} Hence, if $\bP_{t}$, for a given $t$, preserves convex combinations,
then it can be represented as multiplication by a suitable $N\times N$
matrix. Since, conversely, if $\bP_{t}$ can be represented as multiplication
by an $N\times N$ matrix then it preserves convex combinations, we
may define linearity of a probability dynamics as follows:
\begin{defn}
\label{def:linear-prob-dyn} We say that a probability dynamics $\bP$ is \emph{linear} iff $\bP_{t}$
preserves convex combinations for all $t\in T$. 
\end{defn}

Definition~\ref{def:prob-dyn} does not require a probability dynamics
to be linear by default, and we see no mathematical justification to impose such a requirement in general. Indeed, we showed an example of a physically reasonable
scenario in the Introduction for which the probability dynamics is nonlinear.

This, of course, stands in tension with the widely held assumption in the literature that the time evolution of probabilities is linear. To illustrate how this view is typically presented, we provide a sample of quotations from the literature (with minor adjustments to notation):\footnote{The earliest occurrence of this assumption that we could identify in the research literature—though not as explicitly formulated as in the following quotations—appears in \cite{hanggi1977}.}
\begin{quote}
``The one-time probability density $p\left(\boldsymbol{x}_{t}\right)$
of a classical statistical process is related to the initial probability
density via $p\left(\boldsymbol{x}_{t}\right)=\int d\mu\left(\boldsymbol{x}_{t_{0}}\right)p\left(\boldsymbol{x}_{t}|\boldsymbol{x}_{t_{0}}\right)p\left(\boldsymbol{x}_{t_{0}}\right)$.
{[}...{]} {[}T{]}he classical conditional probability density $p\left(\boldsymbol{x}_{t}|\boldsymbol{x}_{t_{0}}\right)$,
also called the transition matrix, {[}...{]} is always independent
of the initial probability density $p\left(\boldsymbol{x}_{t_{0}}\right)$
{[}...{]}. To emphasize this, we will write it in the form
\begin{equation}
T\left(\boldsymbol{x}_{t},\boldsymbol{x}_{t_{0}}\right)\doteq p\left(\boldsymbol{x}_{t}|\boldsymbol{x}_{t_{0}}\right)
\end{equation}

for any $t\geq t_{0}$. {[}...{]} For a discrete-valued stochastic
process, $T$ is indeed a matrix, depending on $t$, with elements
$T\left(\boldsymbol{x}_{t},\boldsymbol{x}_{t_{0}}\right)$.'' (\cite{li2018},
p.~41)

\medskip{}

``In this section, we shall focus on the evolution of one-point probabilities
$p\left(x,t\right)$ during a stochastic process. Thus, consider a linear map $T$ that connects the probability of a random variable $X$, at different times $t_{0}$ and $t_{1}$:
\begin{equation}
p\left(x_{1},t_{1}\right)=\underset{x_{0}}{\sum}T\left(x_{1},t_{1}|x_{0},t_{0}\right)p\left(x_{0},t_{0}\right)
\end{equation}
{[}...{]} Consider $t=t_{0}$ to be the initial time of some {[}...{]}
stochastic process $\left\{ X\left(t\right),\,t\in I\right\} $. From
the definition of conditional probability,
\begin{equation}
\begin{split}p\left(x_{2},t_{2};x_{0},t_{0}\right) & =p\left(x_{2},t_{2}|x_{0},t_{0}\right)p\left(x_{0},t_{0}\right)\\
\Rightarrow p\left(x_{2},t_{2}\right) & =\underset{x_{0}}{\sum}p\left(x_{2},t_{2}|x_{0},t_{0}\right)p\left(x_{0},t_{0}\right)
\end{split}
\end{equation}
and therefore $T\left(x_{2},t_{2}|x_{0},t_{0}\right)=p\left(x_{2},t_{2}|x_{0},t_{0}\right)$
for every $t_{2}$.'' (\cite{rivas2014}, p.~3)

\medskip{}

``An \emph{indivisible stochastic process} will be defined as a model
consisting of two basic ingredients: a \emph{configuration space}
$C$; and a dynamical law in the form of a family of \emph{transition
maps} $\Gamma_{t\leftarrow t_{0}}$ that acts linearly on probability
distributions over $C$ at times $t_{0}$ from some index set, called
\emph{conditioning times}, to yield corresponding probability distributions
over $C$ at times $t$ from some possibly distinct index set, called
\emph{target times}.

{[}...{]} the system's \emph{standalone probabilities} at a conditioning
time $t_{0}$ can be denoted by $p_{j}(t_{0})$, the standalone probabilities
at a target time $t$ can be denoted by $p_{i}(t)$, and the transition
maps $\Gamma_{t\leftarrow t_{0}}$ consist of \emph{conditional probabilities}
\begin{equation}
\Gamma_{ij}\left(t\leftarrow t_{0}\right)\equiv\pp\left(i,t|j,t_{0}\right)
\end{equation}
each of which is the conditional probability for the system to be
in its $i$th configuration at the target time $t$, given that the
system is in its $j$th configuration at the conditioning time $t_{0}$.
{[}...{]} Then, from the law of total probability, or marginalization,
$p_{i}\left(t\right)=\sum_{j=1}^{N}p\left(i,t|j,t_{0}\right)p_{j}\left(t_{0}\right)$,
one has the \emph{linear} relationship 
\begin{equation}
\pp_{i}\left(t\right)=\sum_{j=1}^{N}\Gamma_{ij}\left(t\leftarrow t_{0}\right)\pp_{j}\left(t_{0}\right)\label{eq:fake-marginalization}
\end{equation}
where the standalone probabilities $\pp_{j}\left(t_{0}\right)$ at
the conditioning time $t_{0}$ are assumed to be arbitrary and contingent,
and can therefore be freely adjusted without altering the conditional
probabilities $\Gamma_{ij}\left(t\leftarrow t_{0}\right)$, which
are regarded as fixed features of the model.'' (\cite{barandes2025},
pp.~3--4, emphasis in original) 
\end{quote}
In the account of \cite{li2018}, starting from the law of total probability
and the assumption that two-time conditional probabilities are independent
of the initial one-time probabilities, it is concluded that one-time
probabilities have to be related to the initial probabilities by matrix
multiplication---that is, linearly. In the exposition of \cite{rivas2014},
by contrast, linearity is taken as an assumption, and from this together
with the law of total probability (and the tacit assumption that two-time
conditionals are independent of the initial one-time probabilities)
it is derived that the matrix representing the linear evolution map
must consist of two-time conditional probabilities. The presentation
of \cite{barandes2025} combines elements of the earlier two. However,
all of these arguments depend on premises that do not hold in general---as becomes
clear once they are formulated in the precise probabilistic framework
introduced in the preceding section, which allows one to speak simultaneously
about the functional relationship between probabilities at different
times and multi-time conditional probabilities.

To see this, suppose that the stochastic process family $\bM: \vec{\pp}_0 \mapsto \vec{\mm}_{\vec{\pp}_0}$ implements
probability dynamics $\bP$, that is, for all $t\in T$ and $\vec{\pp}_{0}\in\mathcal{S}_{N}$,
\begin{equation}
\bP\left(t,\vec{\pp}_{0}\right)=\vec{\mm}_{\vec{\pp}_{0}}\left(t\right)\label{eq:implement-p0}
\end{equation}
Recall the definition of events $E_{i}\left(t\right)$ in \eqref{eq:event}.
For all fixed $t\in T$, $\left\{ E_{i}\left(t\right)\right\} _{i=1}^{N}$
forms a partition of $\Omega$. Hence, in any given probability space
$\left(\Omega,\mathcal{F},\mm_{\vec{\pp}_{0}}\right)$, one can apply
the law of total probability: 
\begin{align}
\mm_{\vec{\pp}_{0}}\left(E_{i}\left(t\right)\right) & =\sum_{j=1}^{N}\mm_{\vec{\pp}_{0}}\left(E_{i}\left(t\right)|E_{j}\left(0\right)\right)\mm_{\vec{\pp}_{0}}\left(E_{j}\left(0\right)\right)\label{eq:law-of-total-prob}
\end{align}
if $\pp_{01},\pp_{02},...,\pp_{0N}>0$. This can be expressed in matrix
form as 
\begin{align}
\vec{\mm}_{\vec{\pp}_{0}}\left(t\right) & =\mM_{\vec{\pp}_{0}}\left(t\right)\vec{\mm}_{\vec{\pp}_{0}}\left(0\right)\label{eq:law-of-total-prob-matrix-p0}
\end{align}
with 
\begin{align}
\mM_{\vec{\pp}_{0}}\left(t\right) & \doteq\left(\begin{array}{cccc}
\mm_{\vec{\pp}_{0}}\left(E_{1}\left(t\right)|E_{1}\left(0\right)\right) & \mm_{\vec{\pp}_{0}}\left(E_{1}\left(t\right)|E_{2}\left(0\right)\right) & \ldots & \mm_{\vec{\pp}_{0}}\left(E_{1}\left(t\right)|E_{N}\left(0\right)\right)\\
\mm_{\vec{\pp}_{0}}\left(E_{2}\left(t\right)|E_{1}\left(0\right)\right) & \mm_{\vec{\pp}_{0}}\left(E_{2}\left(t\right)|E_{2}\left(0\right)\right) & \ldots & \mm_{\vec{\pp}_{0}}\left(E_{2}\left(t\right)|E_{N}\left(0\right)\right)\\
\vdots & \vdots &  & \vdots\\
\mm_{\vec{\pp}_{0}}\left(E_{N}\left(t\right)|E_{1}\left(0\right)\right) & \mm_{\vec{\pp}_{0}}\left(E_{N}\left(t\right)|E_{2}\left(0\right)\right) & \ldots & \mm_{\vec{\pp}_{0}}\left(E_{N}\left(t\right)|E_{N}\left(0\right)\right)
\end{array}\right)\label{eq:delta-matrix-p0}
\end{align}
\eqref{eq:implement-p0}, \eqref{eq:law-of-total-prob-matrix-p0}
and $\bP\left(0,\vec{\pp}_{0}\right)=\vec{\pp}_{0}$ together imply 

\begin{equation}
\bP\left(t,\vec{\pp}_{0}\right)=\mM_{\vec{\pp}_{0}}\left(t\right)\vec{\pp}_{0}\label{eq:fake-lin-0}
\end{equation}
for all $t\in T$ and for all $\vec{\pp}_{0}\in\mathcal{S}_{N}$,
$\pp_{01},\pp_{02},...,\pp_{0N}>0$. Adopting the more expressive
notation $\vec{\pp}\left(t\right)\doteq\bP\left(t,\vec{\pp}_{0}\right)$,
we can write this as
\begin{equation}
\vec{\pp}\left(t\right)=\mM_{\vec{\pp}_{0}}\left(t\right)\vec{\pp}_{0}\label{eq:fake-lin}
\end{equation}
This equation is the precise formulation of \eqref{eq:fake-marginalization}
as quoted earlier from \cite{barandes2025}. Without the subscript
on $\mM_{\vec{\pp}_{0}}\left(t\right)$, it might give the impression
that the relationship between $\vec{\pp}\left(t\right)$ and $\vec{\pp}_{0}$
must necessarily be linear, since it is represented as matrix multiplication. But this is not the case. Contrary to the assertion in \cite{barandes2025}, the standalone probabilities $\vec{\pp}_{0}$ cannot be assumed to be arbitrary and cannot be adjusted freely without altering the conditional probabilities.
As indicated by the subscript $\vec{\pp}_{0}$, the conditional probabilities
$\mm_{\vec{\pp}_{0}}\left(E_{i}\left(t\right)|E_{j}\left(0\right)\right)$,
and hence the matrix $\mM_{\vec{\pp}_{0}}\left(t\right)$, are dependent
on $\vec{\pp}_{0}$. To see this more explicitly, recall that 
\begin{equation}
\mm_{\vec{\pp}_{0}}\left(E_{i}\left(t\right)|E_{j}\left(0\right)\right)=\frac{\mm_{\vec{\pp}_{0}}\left(E_{i}\left(t\right)\wedge E_{j}\left(0\right)\right)}{\mm_{\vec{\pp}_{0}}\left(E_{j}\left(0\right)\right)}\label{eq:dependence-on-p0}
\end{equation}
Here the denominator is simply the $j$th entry of $\vec{\pp}_{0}$;
and the numerator also depends on $\vec{\pp}_{0}$, since the measure
$\mm_{\vec{\pp}_{0}}$ that implements the particular solution of
$\bP$ itself depends on the corresponding initial condition $\vec{\pp}_{0}$.
Therefore, adjusting the initial probabilities $\vec{\pp}_{0}$ generally
\emph{does} alter the conditional probabilities that make up the entries
of $\mM_{\vec{\pp}_{0}}\left(t\right)$, and hence the map $\vec{\pp}_{0}\mapsto\mM_{\vec{\pp}_{0}}\left(t\right)\vec{\pp}_{0}$
is \emph{not} linear.

More fundamentally, note the following. Given a single stochastic
process $\mm$, one can write down the law of total probability, analogously
to \eqref{eq:law-of-total-prob-matrix-p0}, as 
\begin{align}
\vec{\mm}\left(t\right) & =\mM\left(t\right)\vec{\mm}\left(0\right)\label{eq:law-of-total-prob-matrix}
\end{align}
with 
\begin{equation}
\mM\left(t\right)\doteq\left(\begin{array}{cccc}
\mm\left(E_{1}\left(t\right)|E_{1}\left(0\right)\right) & \mm\left(E_{1}\left(t\right)|E_{2}\left(0\right)\right) & \ldots & \mm\left(E_{1}\left(t\right)|E_{N}\left(0\right)\right)\\
\mm\left(E_{2}\left(t\right)|E_{1}\left(0\right)\right) & \mm\left(E_{2}\left(t\right)|E_{2}\left(0\right)\right) & \ldots & \mm\left(E_{2}\left(t\right)|E_{N}\left(0\right)\right)\\
\vdots & \vdots &  & \vdots\\
\mm\left(E_{N}\left(t\right)|E_{1}\left(0\right)\right) & \mm\left(E_{N}\left(t\right)|E_{2}\left(0\right)\right) & \ldots & \mm\left(E_{N}\left(t\right)|E_{N}\left(0\right)\right)
\end{array}\right)\label{eq:delta-matrix}
\end{equation}
It must be emphasized that equation \eqref{eq:law-of-total-prob-matrix},
by itself, does \emph{not} encode any functional/law-like relationship
(linear or otherwise) between $\vec{\mm}\left(t\right)$ and $\vec{\mm}\left(0\right)$.
The matrix $\mM\left(t\right)$ is merely a characterization of the
particular measure $\mm$ on $\left(\Omega,\mathcal{F}\right)$ and
says nothing about ``what happens to $\vec{\mm}\left(t\right)$ if
$\vec{\mm}\left(0\right)$ is modified.'' In fact, this latter question
is meaningless in the context of a \emph{single} stochastic process.
If one aims to describe a functional/law-like relationship between probabilities at different times—assuming such a relationship exists at all for the system in question—$\mM\left(t\right)$, a collection of conditional
probabilities, is not the right sort of object. The right sort
of object for this role is instead what we have defined as probability
dynamics $\bP$, or a stochastic process family $\bM$ that implements
$\bP$. There is no probability-theoretic reason to expect that such
an object yields a linear relationship between probabilities at different
times.

\section{Transition Matrices and Transition Probabilities}
\label{sec:Transition}

One can, of course, consider the special case when a probability dynamics
\emph{is} linear. In this case, there exists a unique family $\left\{ \mP\left(t\right)\right\} _{t\in T}$
of $N\times N$ matrices such that 
\begin{equation}
\bP_{t}\left(\vec{\pp}_{0}\right)=\mP\left(t\right)\vec{\pp}_{0}\label{eq:lin-prob-dyn}
\end{equation}
for all $\vec{\pp}_{0}\in\mathcal{S}_{N}$ and $t\in T$. Suppose
further that $\bP$ is implemented by the stochastic process family
$\bM$. Then the system is described by two kinds of matrices: $\mP\left(t\right)$,
which encodes the probability dynamics, and $\mM_{\vec{\pp}_{0}}\left(t\right)$,
defined in \eqref{eq:delta-matrix-p0}, which characterizes the stochastic
process implementing a given solution of the probability dynamics.
As witnessed by the quotations in the preceding section, these two types of matrices are often treated interchangeably in the literature under the label ``transition matrix'' (see also \cite{gillespie1994,gillespie2000,gillespie2001,spekkens2019})
But one has to ask: is there any relation between $\mP\left(t\right)$
and $\mM_{\vec{\pp}_{0}}\left(t\right)$?

Formally, both of them are \emph{stochastic matrices}, meaning that their
columns are probability vectors. However, the reason for this is different
in the two cases. For $\mM_{\vec{\pp}_{0}}\left(t\right)$, the $j$th column 
consists of entries 
\begin{equation}
\mm_{\vec{\pp}_{0}}\left(E_{1}\left(t\right)|E_{j}\left(0\right)\right),\mm_{\vec{\pp}_{0}}\left(E_{2}\left(t\right)|E_{j}\left(0\right)\right),...,\mm_{\vec{\pp}_{0}}\left(E_{N}\left(t\right)|E_{j}\left(0\right)\right)
\end{equation}
which, for $\pp_{0j}\neq0$, are numbers between 0 and 1, and sum
to $1$. This follows from the fact $\mm_{\vec{\pp}_{0}}\left(\cdot|E_{j}\left(0\right)\right)$
is a probability measure and $\left\{ E_{i}\left(t\right)\right\} _{i=1}^{N}$
is a partition. For $\mP\left(t\right)$, the $j$th column 
is $\bP_{t}\left(\vec{e}_{j}\right)$, where $\vec{e}_{j}$ is the
$j$th standard basis vector of $\mathbb{R}^{N}$. Since $\vec{e}_{j}$
is a probability vector and $\bP_{t}$ maps probability vectors to
probability vectors, $\bP_{t}\left(\vec{e}_{j}\right)$ is also a
probability vector.

What is the connection between these two types of stochastic matrices?
From the fact that the stochastic process family $\bM$ implements $\bP$,
equations \eqref{eq:fake-lin} and \eqref{eq:lin-prob-dyn} together yield
\begin{align}
\mP\left(t\right)\vec{\pp}_{0} & =\mM_{\vec{\pp}_{0}}\left(t\right)\vec{\pp}_{0}\label{eq:implement-lin}
\end{align}
for all $\vec{\pp}_{0}\in\mathcal{S}_{N}$, $\pp_{01},\pp_{02},...,\pp_{0N}>0$.
Crucially, however, this does \emph{not} entail that $\mM_{\vec{\pp}_{0}}\left(t\right)=\mP\left(t\right)$
for any $\vec{\pp}_{0}$. Since $\mM_{\vec{\pp}_{0}}\left(t\right)$
depends on $\vec{\pp}_{0}$,  equation \eqref{eq:implement-lin} merely  states that the two
matrices act identically on that specific vector $\vec{\pp}_{0}$. It does not assert that they coincide on all input vectors, which would be required for the matrices themselves to be equal. For a simple illustration of a situation in which  $\mP\left(t\right)$
and $\mM_{\vec{\pp}_{0}}\left(t\right)$ differ for every $\vec{\pp}_{0}$, despite the fact that equation~\eqref{eq:implement-lin} holds for all  $\vec{\pp}_{0}$,
consider the following example.
\begin{example}
\label{exa:different-Ms}Let $N=2$, $T=\left\{ 0,t\right\} $, and
consider a probability dynamics $\bP$ such that $\bP_{t}$ can be
represented by the matrix 
\begin{equation}
\mP\left(t\right)=\left(\begin{array}{cc}
0 & 1\\
1 & 0
\end{array}\right)
\end{equation}
flipping the probability vector $\vec{\pp}_{0}=\left(\pp_{0},1-\pp_{0}\right)$
into $\vec{\pp}\left(t\right)=\left(1-\pp_{0},\pp_{0}\right)$.

On the other hand, consider the following stochastic process family
$\bM$ implementing $\bP$:

\begin{equation}
\begin{aligned}
\mu_{\vec{\pp}_{0}}\!\left(E_{1}(0)\right) &= p_{0} 
&\qquad\qquad \mu_{\vec{\pp}_{0}}\!\left(E_{1}(t)\mid E_{1}(0)\right) &= 1-p_{0} \\
\mu_{\vec{\pp}_{0}}\!\left(E_{2}(0)\right) &= 1-p_{0} 
&\qquad\qquad \mu_{\vec{\pp}_{0}}\!\left(E_{1}(t)\mid E_{2}(0)\right) &= 1-p_{0} \\
\mu_{\vec{\pp}_{0}}\!\left(E_{1}(t)\right) &= 1-p_{0} 
&\qquad\qquad \mu_{\vec{\pp}_{0}}\!\left(E_{2}(t)\mid E_{1}(0)\right) &= p_{0} \\
\mu_{\vec{\pp}_{0}}\!\left(E_{2}(t)\right) &= p_{0} 
&\qquad\qquad \mu_{\vec{\pp}_{0}}\!\left(E_{2}(t)\mid E_{2}(0)\right) &= p_{0}
\end{aligned}
\end{equation}where $\vec{\pp}_{0}=\left(\pp_{0},1-\pp_{0}\right)$ with $\pp_{0}\in\left[0,1\right]$,
except that some of the conditional probabilities are undefined for
 $\pp_{0}\in\left\{ 0,1\right\}$. Accordingly, for $\pp_{0}\in\left(0,1\right)$ we
have
\begin{equation}
\mM_{\vec{\pp}_{0}}\left(t\right)=\left(\begin{array}{cc}
1-\pp_{0} & 1-\pp_{0}\\
\pp_{0} & \pp_{0}
\end{array}\right)
\end{equation}

Clearly, (i) both $\mP\left(t\right)$ and all the $\mM_{\vec{\pp}_{0}}\left(t\right)$-s
are stochastic matrices; (ii) $\mP\left(t\right)\vec{\pp}_{0}=\mM_{\vec{\pp}_{0}}\left(t\right)\vec{\pp}_{0}$
for all $\vec{\pp}_{0}=\left(\pp_{0},1-\pp_{0}\right)$ with $\pp_{0}\in\left(0,1\right)$;
but (iii) $\mM_{\vec{\pp}_{0}}\left(t\right)\neq\mP\left(t\right)$
for any $\vec{\pp}_{0}\in\mathcal{S}_{2}$.
\end{example}

Consequently, although $\mP\left(t\right)$ and $\mM_{\vec{\pp}_{0}}\left(t\right)$
are both stochastic matrices, they are generally not identical, and
so the entries of $\mP\left(t\right)$ cannot be identified with two-time
conditional probabilities (relative to the given stochastic process
family that implements $\bP$). Yet this identification is often made
in the literature, as is made particularly transparent in the quote
from \cite{rivas2014} in the preceding section, with further examples
found in \cite{barandes2025,gillespie1994,gillespie2000,gillespie2001,vacchini2012,spekkens2019,canturk2024}.

One way to locate the source of the problem is by using the following
notion.
\begin{defn}
\label{def:transition-const}We call a stochastic process family $\bM$
\emph{transition-constant} iff 
\begin{equation}
\mm_{\vec{\pp}_{0}}\left(E_{i}\left(t\right)|E_{j}\left(0\right)\right)=\mm_{\vec{q}_{0}}\left(E_{i}\left(t\right)|E_{j}\left(0\right)\right)\label{eq:transition-constant}
\end{equation}
holds for all $i,j\in C$, $t\in T$, and for all $\vec{\pp}_{0},\vec{q}_{0}\in\mathcal{S}_{N}$,
$\pp_{0j},q_{0j}\neq0$.\footnote{The notion of a transition-constant stochastic process {\em family} should
not be confused with that of a {\em single} stochastic process whose ``transition
probabilities'' or ``rates'' remain constant \emph{over time}. A concept
related to the latter is that of a \emph{time-homogeneous Markov  process},
introduced in Definition~\ref{def:time-homogeneous-Markov-process}
in Section~\ref{sec:Markovianity-sub}.}
\end{defn}

\noindent The matrices $\mM_{\vec{\pp}_{0}}\left(t\right)$ of a transition-constant
stochastic process family do not depend on $\vec{\pp}_{0}$, except
for the undefined entries $\left(\mM_{\vec{\pp}_{0}}\left(t\right)\right)_{ij}$
where $\pp_{0j}=0$. To avoid this minor technical complication, it
will be useful to define a common transition matrix as 
\begin{equation}
\mM\left(t\right)\doteq\mM_{\vec{\pp}_{0}}\left(t\right),\,\,\,\,\,\pp_{01},\pp_{02},...,\pp_{0N}>0\label{eq:common-M}
\end{equation}
for all $t\in T$. 

Notice that when a ``stochastic process'' is described ``as a model {[}...{]}
where the standalone probabilities $\vec{\pp}_{0}$ are assumed to
be arbitrary and contingent, and can therefore be freely adjusted
without altering the conditional probabilities $\mM_{ij}\left(t\right)$,
which are regarded as fixed features of the model,'' (\cite{barandes2025},
p.~3) the structure being defined corresponds not to a single stochastic process but rather to a transition-constant stochastic process family. In such a model, the standard assumptions regarding linearity and the equality of the two types of transition matrices—assumptions that are unwarranted in general—do in fact hold:

\begin{restatable}{prop}{transitionconstimplieslinearity}
\label{prop:transition-const-implies-linearity}Suppose that a probability
dynamics $\bP$ is implemented by a transition-constant stochastic
process family $\bM$. Then $\bP$ is linear and $\mP\left(t\right)=\mM\left(t\right)$
for all $t\in T$.
\end{restatable}

\noindent\emph{Sketch of proof.} By transition-constancy, the two-time conditionals $\mm_{\vec{\pp}_{0}}(E_i(t)\mid E_j(0))$ are independent of $\vec{\pp}_{0}$, so the law of total probability yields $\mm_{\vec{\pp}_{0}}(E_i(t))=\sum_{j=1}^{N}\mM_{ij}(t)\pp_{0j}$ for every initial $\vec{\pp}_{0}$. Hence $\bP(t,\vec{\pp}_{0})=\mM(t)\vec{\pp}_{0}$, so $\bP$ is linear and $\mP(t)=\mM(t)$. (Formal proof: Appendix \ref{proof:transition-const-implies-linearity}.)

Crucially, however, the converse is not true: a stochastic process
family implementing a linear probability dynamics need \emph{not}
be transition-constant. This is illustrated by Example~\ref{exa:different-Ms}.
What is true is the following existence claim:

\begin{restatable}{prop}{linearimplement}
\label{prop:linearimplement} Let $\bP$ be a linear probability dynamics
with $T\subseteq\mathbb{R}_{\geq0}$. There exists a transition-constant
stochastic process family $\bM$ that implements $\bP$. 
\end{restatable}

\noindent\emph{Sketch of proof.} For each initial distribution $\vec{\pp}_{0}$, construct an implementing stochastic process by first drawing the initial configuration $j$ with probability $\pp_{0j}$, and then (conditional on $j$) letting the random variables at all later times be independent with $\mm_{\vec{\pp}_{0}}(E_i(t)\mid E_j(0))=\mP_{ij}(t)$ for every $t>0$. These finite-dimensional probabilities are consistent, so Kolmogorov's extension theorem yields a measure on full trajectories whose one-time marginals reproduce $\vec{\mm}_{\vec{\pp}_{0}}(t)=\mP(t)\vec{\pp}_{0}$, and since the conditionals are exactly $\mP_{ij}(t)$ (independent of $\vec{\pp}_{0}$), the resulting implementing family is transition-constant. (Formal proof: Appendix \ref{proof:linearimplement}.)

What this proposition shows is that, given only the one-time
probabilities of a linear probability dynamics $\bP\left(t,\vec{\pp}_{0}\right)$, it is always possible,
at least in principle, to interpret the entries of the corresponding matrix $\mP\left(t\right)$
as two-time conditional probabilities---relative to a suitably
selected transition-constant stochastic process family that implements
$\bP$ and hence for which $\mP\left(t\right)=\mM\left(t\right)$.
Of course, the mere \emph{possibility} of such an interpretation does
not imply that $\mP\left(t\right)$ has any \emph{necessary/conceptual}
connection to multi-time conditional probabilities.

More broadly, the distinction between the two types of ``transition
matrices'' is tightly related to the distinction between two types
of ``transition probabilities'':
\begin{itemize}
\item[1)] ``the probability of configuration $i$ at time $t_{}$, given the
system is in configuration $j$ at time $0$''
\item[2)] ``the probability of configuration $i$ at time $t$, given the \emph{probabilities}
of the system over all of its $j$ configurations at time $0$''  
\end{itemize}
As discussed in Section~\ref{sec:twodescriptions}, the first quantity
describes the (probabilistic) transition \emph{of configurations}
and corresponds to multi-time conditional probability, whereas the
second describes the (deterministic) transition \emph{of the probabilities
themselves}, reflecting their functional relationship across different
times. Accordingly, the first probability is well-defined in the context
of a stochastic process and is given by $\mm\left(E_{i}\left(t\right)|E_{j}\left(0\right)\right)$;
but it is not well-defined in a probability dynamics. Conversely,
the second probability is well-defined in a probability dynamics (provided
that the antecedent probability is specified for all $j=1,2,...,N$),
and it is given by the $i$th entry of $\bP\left(t,\vec{\pp}_{0}\right)$;
but it is not well-defined for a (single) stochastic process. To make
sense of both notions of ``transition probability,'' one has to
consider a stochastic process family implementing a probability
dynamics.

In that context, a connection between the two kinds of transition probabilities arises only in the special case in which the antecedent probability associated with the type-2 transition probability---the probability that the system
is in configuration $j$ at time $0$ for some $j$---equals 1. In
this case, the two probabilities coincide. More precisely, if $\mm_{\vec{\pp}_{0}}\left(E_{j}\left(0\right)\right)=1$,
then 
\begin{equation}
\mm_{\vec{\pp}_{0}}\left(E_{i}\left(t\right)|E_{j}\left(0\right)\right)=\mm_{\vec{\pp}_{0}}\left(E_{i}\left(t\right)\right)=\pp_{i}\left(t\right)
\end{equation}
where $\pp_{i}\left(t\right)$ denotes $i$th entry of $\bP\left(t,\vec{\pp}_{0}\right)$.

When $\bP$ is linear, this implies a link between the two types of
transition matrices. $\mm_{\vec{\pp}_{0}}\left(E_{j}\left(0\right)\right)=1$
means that $\vec{\pp}_{0}=\vec{e}_{j}$, one of the standard basis
vectors. In that case, the right hand side of \eqref{eq:implement-lin}
is the $j$th column of $\mM_{\vec{\pp}_{0}}\left(t\right)$, and
the left hand side of \eqref{eq:implement-lin} is the $j$th column
of $\mP\left(t\right)$. Hence, for all $j=1,2,...,N$, the $j$th
columns of $\mM_{\vec{e}_{j}}\left(t\right)$ and $\mP\left(t\right)$
must coincide. All other entries $\mm_{\vec{e}_{j}}\left(E_{i}\left(t\right)|E_{k}\left(0\right)\right)$
of $\mM_{\vec{e}_{j}}\left(t\right)$ are undefined, since their denominators
vanish: $\mm_{\vec{e}_{j}}\left(E_{k}\left(0\right)\right)=0$ if
$k\neq j$.

Note that even in the special case in which a probability dynamics
is implemented by a transition-constant stochastic-process family,
so that $\mP\left(t\right)=\mM\left(t\right)$, the two types of transition
\emph{probabilities} coincide only when the initial probability that the system
is in configuration $j$ equals 1 for some $j$. This
is because only in this case do the entries of the matrix $\mP\left(t\right)$
coincide with the corresponding values of the second type of transition probability. Even in the transition-constant case, for a general initial
probability distribution $\vec{\pp}_{0}$, the two kinds of transition
probabilities are related by
\begin{equation}
\pp_{i}\left(t\right)=\left(\mM\left(t\right)\vec{\pp}_{0}\right)_{i}=\sum_{j=1}^{N}\mm_{\vec{\pp}_{0}}\left(E_{i}\left(t\right)|E_{j}\left(0\right)\right)\pp_{0j}
\end{equation}
So $\pp_{i}\left(t\right)= \mm_{\vec{\pp}_{0}}\left(E_{i}\left(t\right)|E_{j}\left(0\right)\right)$ for all $i$ if and only if $\pp_{0j}=1$.

\section{Divisibility, Decomposability, and Markovianity}
\label{sec:Markovianity}

One may impose various additional conditions that the probabilistic
description of a system's evolution must satisfy. One such condition
is known in the literature as \emph{divisibility}. Here is an introduction
to the idea: 
\begin{quote}
Crucially, the transition matrix $\Gamma(t\leftarrow t_{0})$ will
\emph{not} be assumed to be ``divisible'' {[}...{]} That is, $\Gamma(t\leftarrow t_{0})$
will generically be \emph{indivisible} {[}...{]} meaning that for
intermediate times $t'$ satisfying $t>t'>t_{0}$, there will not
generally exist a genuinely stochastic matrix $\tilde{\Gamma}(t\leftarrow t')$
satisfying the composition law or \emph{divisibility condition} 
\begin{equation}
\Gamma\left(t\leftarrow t_{0}\right)=\tilde{\Gamma}\left(t\leftarrow t'\right)\Gamma\left(t'\leftarrow t_{0}\right)\label{eq:divis-Barandes}
\end{equation}
(\citet{barandes2025}, p.~4.) 
\end{quote}
Given that the notation $\Gamma(t\leftarrow t_{0})$ in \cite{barandes2025} treats interchangeably properties of what we have denoted as $\mM(t)$ and $\bP_{t}$ (see equation~\eqref{eq:delta-matrix}
and Definition~\ref{def:prob-dyn}), we have to ask: is ``divisibility''
meant to be a condition about the former or the latter? We consider each possibility in turn. We will use \emph{``divisibility''} (with quotation marks) to refer
to the notion as it appears in the literature, while \emph{divisibility}
(without quotation marks) will denote the concept defined in the context
of a probability dynamics in Definition~\ref{def:divisib}.

\subsection{Markovianity}\label{sec:Markovianity-sub}

Assuming first that ``divisibility'' concerns $\mM(t)$, consider a stochastic process $\mm$. Define,
for any $t,t'\in T$, the transition matrix 
\begin{align}
\mM\left(t\leftarrow t'\right) & \doteq\left(\begin{array}{cccc}
\mm\left(E_{1}\left(t\right)|E_{1}\left(t'\right)\right) & \mm\left(E_{1}\left(t\right)|E_{2}\left(t'\right)\right) & \ldots & \mm\left(E_{1}\left(t\right)|E_{N}\left(t'\right)\right)\\
\mm\left(E_{2}\left(t\right)|E_{1}\left(t'\right)\right) & \mm\left(E_{2}\left(t\right)|E_{2}\left(t'\right)\right) & \ldots & \mm\left(E_{2}\left(t\right)|E_{N}\left(t'\right)\right)\\
\vdots\\
\mm\left(E_{N}\left(t\right)|E_{1}\left(t'\right)\right) & \mm\left(E_{N}\left(t\right)|E_{2}\left(t'\right)\right) & \ldots & \mm\left(E_{N}\left(t\right)|E_{N}\left(t'\right)\right)
\end{array}\right)\label{eq:delta-t-t'}
\end{align}
Similarly to $\mM\left(t\right) = \mM\left(t\leftarrow 0\right)$, $\mM\left(t\leftarrow t'\right)$ is a stochastic matrix. By applying the law of total probability first
to the partition $\left\{ E_{j}\left(t'\right)\right\} _{j=1}^{N}$
and then to $\left\{ E_{k}\left(0\right)\right\} _{k=1}^{N}$, we
obtain 
\begin{equation}
\mm\left(E_{i}\left(t\right)\right)=\sum_{j,k=1}^{N}\mm\left(E_{i}\left(t\right)|E_{j}\left(t'\right)\right)\mm\left(E_{j}\left(t'\right)|E_{k}\left(0\right)\right)\mm\left(E_{k}\left(0\right)\right)
\end{equation}
Using the definition of $\mM\left(t\leftarrow t'\right)$, this equation
can be recast in matrix form as 
\begin{align}
\vec{\mm}\left(t\right) & =\mM\left(t\leftarrow t'\right)\mM\left(t'\right)\vec{\mm}\left(0\right)\label{eq:transition-2}\\
 & =\mM\left(t\right)\vec{\mm}\left(0\right)\label{eq:transition-1-again}
\end{align}
where the second equality comes from \eqref{eq:law-of-total-prob-matrix-p0}.
 Although \eqref{eq:transition-2} and \eqref{eq:transition-1-again} both hold by the definition of a stochastic process,
in general we do not have 
\begin{equation}
\mM\left(t\right)=\mM\left(t\leftarrow t'\right)\mM\left(t'\right)\label{eq:transition-3}
\end{equation}
This is because two matrices that yield the same result when applied
to a \emph{specific} vector---here, $\vec{\mm}\left(0\right)$---are
not necessarily equal. 

Equation~\eqref{eq:transition-3} is the matrix
form of the so-called Chapman-Kolmogorov equation. A sufficient condition
for this to hold is the Markov property. 
\begin{defn}
A stochastic process $\mm$ is called \emph{Markovian} iff 
\begin{equation}
\mm\left(E_{i_{m+1}}\left(t_{m+1}\right)|E_{i_{1}}\left(t_{1}\right)\wedge...\wedge E_{i_{m}}\left(t_{m}\right)\right)=\mm\left(E_{i_{m+1}}\left(t_{m+1}\right)|E_{i_{m}}\left(t_{m}\right)\right)\label{eq:markov-1}
\end{equation}
for all $m\in\mathbb{N}$, $i_{1},...,i_{m},i_{m+1}\in C$, and $t_{1}\leq t_{2}\leq...\leq t_{m}\leq t_{m+1}\in T$ (when both sides are well-defined). 
\end{defn}

\noindent It is easy to see that:

\begin{restatable}{prop}{markoviansatisfieschapmankolmogorov}
\label{prop:markovian-satisfies-chapman-kolmogorov}
If a stochastic process is Markovian, then the Chapman-Kolmogorov equation \eqref{eq:transition-3}
holds for all $t,t'\in T,\,0\leq t'\leq t$. 
\end{restatable}

\noindent\emph{Sketch of proof.} Expand $\mm(E_i(t)|E_k(0))$ by conditioning on the intermediate-time $E_j(t')$ events, and then use the Markov property to replace $\mm(E_i(t)\mid E_k(0)\wedge E_j(t'))$ by $\mm(E_i(t)\mid E_j(t'))$, which yields the Chapman--Kolmogorov equation \eqref{eq:transition-3} entrywise. (Formal proof: Appendix \ref{proof:markovian-satisfies-chapman-kolmogorov}.)

The Markov property \eqref{eq:markov-1} is sometimes described as
expressing that the system has ``no memory.'' This holds in the following
sense. In any stochastic process $\mm$, given a configuration $j\in C$
and a time $t'\in T$ at which the system occupies $j$, the probabilities
of subsequent configurations for $t\geq t'$ are determined as 
conditional probabilities $\mM_{ij}\left(t\leftarrow t'\right)=\mm\left(E_{i}\left(t\right)|E_{j}\left(t'\right)\right)$,
$i=1,2,...,N$. In a Markovian stochastic process, these conditional
probabilities are unaffected by information about the earlier trajectory
of configurations ($t\leq t'$) from which $j$ evolved. In this sense,
in a Markov process, a configuration can be regarded as the instantaneous
``stochastic state'' of the system---one that does not uniquely
determine future configurations, as in a deterministic system, but
nonetheless specifies past-independent probabilities for them. Importantly,
Markovianity still allows this specification to depend explicitly
on $t'$, the time at which the system is in configuration $j$, as
indicated by the argument of $\mM_{ij}\left(t\leftarrow t'\right)$.
One might instead require that the system's
configuration $j$ at time $t'$ determine the probabilities of later
configurations independently of $t'$ itself---that is, independently of \emph{when} the system occupies $j$. This may be regarded as a stronger sense of “no memory”: the system need not keep track of its past trajectory, nor even of the passage of time. This stronger condition is
often referred to as \emph{time-homogeneity}, and its precise definition
is as follows.
\begin{defn}
\label{def:time-homogeneous-Markov-process}A Markovian stochastic
process $\mm$ is called \emph{time-homogeneous} iff 
\begin{equation}
\mm\left(E_{i}\left(t\right)|E_{j}\left(t'\right)\right)=\mm\left(E_{i}\left(t-t'\right)|E_{j}\left(0\right)\right)\label{eq:time-homogeneous-Markov}
\end{equation}
for all $t,t'\in T,\,0\leq t'\leq t$ and $i,j\in C$, $\mm\left(E_{j}\left(0\right)\right),\mm\left(E_{j}\left(t'\right)\right)\neq0$.
\end{defn}

\noindent Thus, in a time-homogeneous Markov process, the transition
matrices satisfy
\begin{equation}
\mM\left(t\leftarrow t'\right)=\mM\left(t-t'\right)\label{eq:transition-matrix-stationary}
\end{equation}
for $0\leq t'\leq t$, and the Chapman-Kolmogorov equation becomes
a condition on the matrix family $\left\{ \mM\left(t\right)\right\} _{t\in T}$
itself:
\begin{equation}
\mM\left(t\right)=\mM\left(t-t'\right)\mM\left(t'\right)\label{eq:transition-3-stationary}
\end{equation}
for all $t,t'\in T,\,0\leq t'\leq t$.

It must be emphasized that neither equation \eqref{eq:transition-2}
nor the Chapman-Kolmogorov equation \eqref{eq:transition-3} or \eqref{eq:transition-3-stationary}
can be interpreted as a condition about the ``divisibility'' of
the evolution of probabilities; that is, they do \emph{not} express
that the transition from $\vec{\mm}\left(0\right)$ to $\vec{\mm}\left(t\right)$
can be decomposed into two steps---from $\vec{\mm}\left(0\right)$
to $\vec{\mm}\left(t'\right)$ and then from $\vec{\mm}\left(t'\right)$
to $\vec{\mm}\left(t\right)$. As discussed in Section~\ref{sec:Transition},
this is because a single stochastic process, characterized by the
matrices $\mM$, does \emph{not} describe the transition of \emph{probability
vectors};\emph{ }rather, it describes the (probabilistic) transition
of \emph{configurations}. Note that it is only under this latter construal,
as conditions concerning the evolution of \emph{configurations}, that
Markovianity \eqref{eq:markov-1} and time-homogeneity \eqref{eq:time-homogeneous-Markov}
can be regarded as expressing the ``memorylessness'' of the time
evolution.

Motivated by the close analogy with the ``divisibility condition'' \eqref{eq:divis-Barandes} in \cite{barandes2025},\footnote{This question is discussed and explicitly related to the notion of
``divisibility'' in the literature. See e.g. \citet{vacchini2012,canturk2024}.} one may ask, for a stochastic process that does \emph{not} satisfy
the Chapman-Kolmogorov equation, whether there exists a stochastic
matrix $\tilde{\mM}\left(t\leftarrow t'\right)$, distinct from $\mM\left(t\leftarrow t'\right)$,
such that 
\begin{equation}
\mM\left(t\right)=\tilde{\mM}\left(t\leftarrow t'\right)\mM\left(t'\right)\label{eq:transition-4}
\end{equation}
Although mathematically well-defined, this question again has nothing
to do with the ``divisibility'' of the evolution of probabilities.
The reason is the same as before: a single stochastic process---and
the matrices $\mM\left(t\right)$ and $\tilde{\mM}\left(t\leftarrow t'\right)$
used to describe it---does \emph{not} determine how probability
vectors evolve. 

\subsection{Decomposability}\label{sec:Decomposability}

To express the ``divisibility'' of the evolution of probabilities,
one must instead apply equation \eqref{eq:divis-Barandes} in the
context of a probability dynamics, as a condition about $\bP_{t}$. First, note that in this context, equation \eqref{eq:divis-Barandes}
only makes sense if $\bP$ is linear. Otherwise $\bP_{t}$ cannot
generally be represented by a matrix $\mP\left(t\right)$, so there
is no reason to expect the existence of a \emph{matrix} $\mP\left(t\leftarrow t'\right)$
satisfying the analogue of \eqref{eq:divis-Barandes}, and to impose
the stochasticity condition on it. However, before discussing the
linear case, let us first formulate a simple, general definition that
we believe captures the sought-after idea behind ``divisibility,'' and that is
applicable in the nonlinear case as well. 
\begin{defn}
\label{def:decomposability}Let $\mathrm{Ran}\bP_{t}\subseteq\mathcal{S}_{N}$
denote the range of map $\bP_{t}$. We say that the family of maps
$\left\{ \bP_{t\leftarrow t'}:\mathrm{Ran}\bP_{t'}\rightarrow\mathcal{S}_{N}\right\} _{t,t'\in T,\,t'\leq t}$
\emph{decomposes} probability dynamics $\bP$ iff 
\begin{align}
\bP_{t} & =\bP_{t\leftarrow t'}\circ\bP_{t'}\label{eq:decomp}
\end{align}
for all $t,t'\in T,\,t'\leq t$.\footnote{We note that equation \eqref{eq:decomp}, together with some of the
notions introduced below, would in fact make sense for all $t,t'\in T$,
even when the condition $t'\leq t$ is not satisfied. Since this added
generality plays no role in the present discussion, we shall restrict
attention to the standard case $t'\leq t$.} We call $\bP$ \emph{decomposable} iff there exists a family of maps
that decomposes it.
\end{defn}
\noindent Decomposability provides a natural formalization of the
idea that the evolution of probabilities, whether linear or nonlinear,
can be broken down into intermediate steps. The following simple result,
which is an immediate consequence of Definition~\ref{def:decomposability},
will allow us to refine this interpretation.
\begin{prop}
\label{prop:decomp-characterization}(i) Probability dynamics $\bP$
is decomposable if and only if, for all $t,t'\in T,\,t\geq t'$, $\bP_{t'}\left(\vec{\pp}_{0}\right)=\bP_{t'}\left(\vec{q}_{0}\right)$
entails $\bP_{t}\left(\vec{\pp}_{0}\right)=\bP_{t}\left(\vec{q}_{0}\right)$
for all $\vec{\pp}_{0},\vec{q}_{0}\in\mathcal{S}_{N}$.

(ii) If $\bP$ is decomposable, then the family of maps $\left\{ \bP_{t\leftarrow t'}\right\} $
decomposing it is unique.
\end{prop}

\noindent In line with this characterization, what decomposability
says is that there are no initial probability vectors $\vec{\pp}_{0}$
and $\vec{q}_{0}$ such that $\bP_{t'}\left(\vec{\pp}_{0}\right)=\bP_{t'}\left(\vec{q}_{0}\right)$
while $\bP_{t}\left(\vec{\pp}_{0}\right)\neq\bP_{t}\left(\vec{q}_{0}\right)$,
for $t\geq t'$. That is, for $t\geq t'$, $\vec{\pp}\left(t'\right)$ uniquely determines
$\vec{\pp}\left(t\right)$. In that sense, just like $\vec{\pp}_{0}$ at the initial
moment of time, $\vec{\pp}\left(t'\right)$ serves as the description
of the ``dynamical state'' of the system at time $t'$. This
is the core idea behind decomposability. 

Note that this idea asserts only that the instantaneous probability
vector of the system at any given time uniquely determines its later values.
It does \emph{not} imply the claim that at any given time (other than
$t=0$) the system could be initialized from an \emph{arbitrary} probability
vector. Indeed, decomposability is a notion relative to a particular
probability dynamics $\bP$, which may exclude some probability vectors
(those not in the range of $\bP_{t'}$) from the set of possible states
at time $t'$---unless $\bP_{t'}$ is surjective. Put differently, the dynamics can be initialized arbitrarily only \emph{once} (at the moment we label $0$), after which the dynamical law itself determines which states are possible at any later time $t'$. This is the reason
why the decomposing maps $\bP_{t\leftarrow t'}$ are only required
to be defined on the range of $\bP_{t'}$ and not on the entire simplex
$\mathcal{S}_{N}$. As we will see, in the linear case this subtlety
plays an important role.

The fact that the instantaneous probability vector provides a notion
of state, in the sense just discussed, is sometimes characterized
as the system having ``no memory.'' For a decomposable dynamics, this
is true in the following sense: given a probability vector $\vec{\pp}\in\mathcal{S}_{N}$,
and a time $t'\in T$ at which the system occupies $\vec{\pp}$, the
subsequent trajectory of probability vectors for $t\geq t'$ is fully
determined, without any need for information about the earlier trajectory
($t\leq t'$) from which $\vec{\pp}$ evolved. Importantly, however,
decomposability allows this determination to depend explicitly on
$t'$, the time at which the system is in $\vec{\pp}$, as indicated
by the subscript of the decomposing maps $\bP_{t\leftarrow t'}$.
One might adopt a stronger interpretation of “no memory,” according to which the system's state
$\vec{\pp}$ at time $t'$ determines the future evolution independently
of $t'$ itself---that is, regardless of \emph{when} the system occupies
$\vec{\pp}$. As before, this stronger condition is often referred to as \emph{time-homogeneity},
and its precise definition is as follows (note the distinction between the notations  $\bP_{t- t'}$ and  $\bP_{t\leftarrow t'}$). 
\begin{defn}
\label{def:time-homogeneity}Probability dynamics $\bP$ is called
\emph{time-homogeneous} iff 
\begin{align}
\bP_{t} & =\bP_{t-t'}\circ\bP_{t'}\label{eq:time-homogeneous}
\end{align}
for all $t,t'\in T,\,t'\leq t$.
\end{defn}

\noindent Every time-homogeneous probability dynamics is
decomposable, with 
\begin{equation}
\bP_{t\leftarrow t'}=\left.\bP_{t-t'}\right|{}_{\mathrm{Ran}\bP_{t'}}\label{eq:time-homogeneity}
\end{equation}
for all $t,t'\in T,\,t'\leq t$. What equation \eqref{eq:time-homogeneity}
expresses is that in the time-homogeneous case there not only exists
a map specifying how instantaneous probability vectors at time $t'$
evolve, but moreover this map is the \emph{same} one applied at time
$0$ (restricted to those vectors that are possible at $t'$).

Let us now consider a linear probability dynamics that is decomposable.
Does it then follow that the family of maps decomposing it consists
of maps that themselves preserve convex combinations? The answer is
affirmative: 

\begin{restatable}{prop}{decompmapslinear}
\label{prop:decomp-maps-linear}Suppose that $\bP$ is a linear and
decomposable probability dynamics. Then each map $\bP_{t\leftarrow t'}$
in the family decomposing $\bP$ is convex-combination preserving.
\end{restatable}

\noindent\emph{Sketch of proof.} Take $\vec{\pp},\vec{q}\in\mathrm{Ran}\bP_{t'}$ so that $\vec{\pp}=\bP_{t'}(\vec{\pp}_0)$ and $\vec{q}=\bP_{t'}(\vec{q}_0)$, and use decomposability $\bP_t=\bP_{t\leftarrow t'}\circ\bP_{t'}$ together with convex-combination preservation of $\bP_{t'}$ and $\bP_t$ to compute
$\bP_{t\leftarrow t'}(\lambda\vec{\pp}+(1-\lambda)\vec{q})=\bP_t(\lambda\vec{\pp}_0+(1-\lambda)\vec{q}_0)=\lambda\bP_{t\leftarrow t'}(\vec{\pp})+(1-\lambda)\bP_{t\leftarrow t'}(\vec{q})$ for all $\lambda\in[0,1]$. (Formal proof: Appendix \ref{proof:decomp-maps-linear}.)

If $\bP$ is a linear probability dynamics, there exists a unique
family $\left\{ \mP\left(t\right)\right\} _{t\in T}$ of stochastic
matrices that represents $\bP$ in the sense of \eqref{eq:lin-prob-dyn}.
In terms of these matrices, one can now introduce the notion of \emph{divisibility}.
\begin{defn}
\label{def:divisib}A linear probability dynamics $\bP$ is called
\emph{divisible} iff there exists a family $\left\{ \mP\left(t\leftarrow t'\right)\right\} _{t,t'\in T,\,t'\leq t}$
of stochastic matrices such that 
\begin{equation}
\mP\left(t\right)=\mP\left(t\leftarrow t'\right)\mP\left(t'\right)\label{eq:divisbility-matrix-definition}
\end{equation}
for all $t,t'\in T,\,t'\leq t$. 
\end{defn}

\noindent We now ask: what is the relationship between decomposability
and divisibility for a linear probability dynamics $\bP$?

It is clear that if $\bP$ is divisible, then it is also decomposable.
Indeed, the action of each stochastic matrix $\mP\left(t\leftarrow t'\right)$
defines an $\mathcal{S}_{N}\rightarrow\mathcal{S}_{N}$ convex combination
preserving map, which, when restricted to the range of $\bP_{t'}$,
yields $\bP_{t\leftarrow t'}$. This map satisfies the defining equation
\eqref{eq:decomp} of decomposability as a consequence of \eqref{eq:divisbility-matrix-definition}.
(For a more detailed argument, see proof of Proposition~\ref{prop:linear-extension}
below.)

However, the converse does \emph{not} hold, as shown by the following
counterexample.
\begin{example}
\label{exa:counterexample}Let $N=2,\ T=\left\{ 0,1,2\right\} $, and
consider the linear probability dynamics $\bP$ defined by the following
stochastic matrices: 
\begin{gather}
\mP\left(1\right)=\left(\begin{array}{cc}
1 & \frac{1}{2}\\
0 & \frac{1}{2}
\end{array}\right),\,\,\,\,\,\mP\left(2\right)=\left(\begin{array}{cc}
\frac{1}{2} & 1\\
\frac{1}{2} & 0
\end{array}\right)
\end{gather}
If $\bP$ were divisible, there would exist a stochastic matrix $\mP\left(2\leftarrow1\right)$
satisfying 
\begin{equation}
\mP\left(2\right)=\mP\left(2\leftarrow1\right)\mP\left(1\right)\label{eq:decomposibility-matrix}
\end{equation}
Since $\mP\left(1\right)$ is invertible, with inverse 
\begin{equation}
\mP\left(1\right)^{-1}=\left(\begin{array}{cc}
1 & -1\\
0 & 2
\end{array}\right)
\end{equation}
the only solution of \ref{eq:decomposibility-matrix} for $\mP\left(2\leftarrow1\right)$
is 
\begin{equation}
\mP\left(2\right)\mP\left(1\right)^{-1}=\left(\begin{array}{cc}
\frac{1}{2} & 1\\
\frac{1}{2} & 0
\end{array}\right)\left(\begin{array}{cc}
1 & -1\\
0 & 2
\end{array}\right)=\left(\begin{array}{cc}
\frac{1}{2} & \frac{3}{2}\\
\frac{1}{2} & -\frac{1}{2}
\end{array}\right)\label{eq:decomposing-matrix}
\end{equation}
However, this matrix is not stochastic: some of its entries lie outside
the $\left[0,1\right]$ interval. Therefore, no stochastic matrix
$\mP\left(2\leftarrow1\right)$ can satisfy \eqref{eq:decomposibility-matrix},
and hence $\bP$ is not divisible.

On the other hand, $\bP$ \emph{is} decomposable. Indeed, since $\mP\left(1\right)$
is invertible, the map $\bP_{1}$ is injective. Consequently, $\bP_{1}\left(\vec{\pp}_{0}\right)=\bP_{1}\left(\vec{q}_{0}\right)$
trivially entails $\bP_{2}\left(\vec{\pp}_{0}\right)=\bP_{2}\left(\vec{q}_{0}\right)$
for all $\vec{\pp}_{0},\vec{q}_{0}\in\mathcal{S}_{2}$, and therefore,
by Proposition~\ref{prop:decomp-characterization}, $\bP$ is decomposable.
\end{example}

The example highlights an important point. Observe that the range
of $\bP_{1}$ is restricted to one ``half'' of the simplex $\mathcal{S}_{2}$:
\begin{equation}
\mathrm{Ran}\bP_{1}=\left\{ \left(\pp,1-\pp\right)|\frac{1}{2}\leq\pp\leq1\right\} 
\end{equation}
On this range, the map $\bP_{2\leftarrow1}$ is given by
\begin{equation}
\bP_{2\leftarrow1}\left(\pp,1-\pp\right)=\left(\frac{3}{2}-\pp,\pp-\frac{1}{2}\right)\label{eq:L21}
\end{equation}
Importantly, although $\bP_{2\leftarrow1}$ preserves convex combinations
and maps probability vectors to probability vectors, it cannot be
represented as the action of a stochastic matrix.\footnote{It can, however, be expressed as the action of the non-stochastic
matrix $\mP\left(2\right)\mP\left(1\right)^{-1}$ given in \eqref{eq:decomposing-matrix}.} The reason is that its domain is not the entire simplex $\mathcal{S}_{2}$,
just a subset of it. But only those convex-combination preserving
maps can be represented by stochastic matrices which take \emph{all}
probability vectors to probability vectors, or can be extended so
as to do so. Crucially, however, \emph{$\bP_{2\leftarrow1}$ cannot
be extended to the entire $\mathcal{S}_{2}$ in a way that both preserves
convex combinations and keeps the image within $\mathcal{S}_{2}$}
(see Fig.~\ref{fig:example-nondivisib}). In such a case, even though decomposability holds, divisibility
fails. The following proposition states this observation in general
terms.

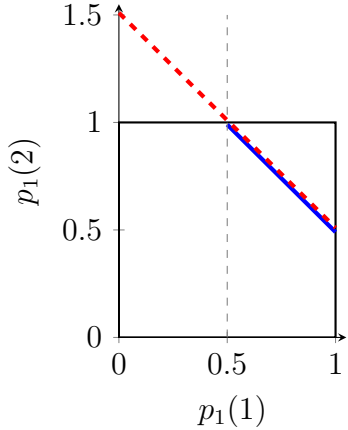
\begin{figure}
\begin{centering}
\begin{tikzpicture}
  \begin{axis}[
    width=6cm,
    height=6cm,
    xmin=0, xmax=1.05,  
    ymin=0, ymax=1.55,  
    axis lines=left,     
    xlabel={$p_{1}(1)$},
    ylabel={$p_{1}(2)$},
    xtick={0,0.5,1},
    ytick={0,0.5,1,1.5},
    axis equal image,    
    clip marker paths=true,
    ticks=both,
    enlargelimits=false
  ]

  \draw[black, thick] (axis cs:0,0) rectangle (axis cs:1,1);

  \addplot+[blue, ultra thick, mark=none, domain=0.5:1, samples=2] {1.49 - x};

  \addplot+[red, dashed, ultra thick, mark=none, domain=0:1, samples=2] {1.51 - x};

  \addplot+[gray, dashed, mark=none] coordinates {(0.5,0) (0.5,1.5)};

\end{axis}
\end{tikzpicture} 
\par\end{centering}
\caption{\emph{Example~\ref{exa:counterexample}: the  first component of the probability vector at time 2, $p_{1}(2)$, is plotted as a function of the first component of the probability vector at time 1, $p_{1}(1)$. The blue continuous graph corresponds to the map
$\bP_{2\leftarrow1}$ as given by \eqref{eq:L21}; the red dashed
graph corresponds to the action of the matrix in \eqref{eq:decomposing-matrix}.
There is no extension of the blue line segment to a graph over the
entire interval }$\left[0,1\right]$\emph{ that is both a straight
line and remains within the unit square.}}
\label{fig:example-nondivisib} 
\end{figure}

\begin{restatable}{prop}{linearextension}
\label{prop:linear-extension}Suppose $\bP$ is a decomposable linear
probability dynamics. Then $\bP$ is divisible if and only if each
map $\bP_{t\leftarrow t'}$ in the family decomposing $\bP$ can be
extended to a convex-combination preserving map \textup{$\bar{\bP}_{t\leftarrow t'}:\mathcal{S}_{N}\rightarrow\mathcal{S}_{N}$}. 
\end{restatable}

\noindent\emph{Sketch of proof.} If $\bP_{t\leftarrow t'}$ extends to a convex-combination preserving map on $\mathcal{S}_N$, then it has a (unique) linear extension to $\mathbb{R}^N$, hence $\bP_{t\leftarrow t'}(\vec{\pp})=\mP(t\leftarrow t')\vec{\pp}$ for a stochastic matrix $\mP(t\leftarrow t')$, and using $\bP_t=\bP_{t\leftarrow t'}\circ\bP_{t'}$ yields $\mP(t)=\mP(t\leftarrow t')\mP(t')$, i.e.\ divisibility. Conversely, if $\bP$ is divisible with stochastic matrix $\mP(t\leftarrow t')$, then matrix multiplication defines the required convex-combination preserving extension of $\bP_{t\leftarrow t'}$. (Formal proof: Appendix \ref{proof:linear-extension}.)

According to Proposition~\ref{prop:linear-extension}, divisibility
strengthens decomposability by the additional requirement that the decomposing maps
$\bP_{t\leftarrow t'}$ admit extensions beyond the range of the maps
$\bP_{t'}$ in a specific way. From a mathematical point of view,
this is certainly a sensible additional requirement. However, its physical motivation
is much less clear. Why should the ``divisibility'' of a dynamics
$\bP$ depend on how the decomposing maps $\bP_{t\leftarrow t'}$,
once their existence is ensured, act on probability vectors \emph{that
are excluded by the dynamics itself}---that is, on vectors lying
outside the range of $\bP_{t'}$? 

In one special case, the required extension of $\bP_{t\leftarrow t'}$
is guaranteed to exist: when the dynamics is not merely decomposable
but also time-homogeneous, in the sense defined in Definition~\ref{def:time-homogeneity}.
In that case, $\bP_{t\leftarrow t'}$ is given by \eqref{eq:time-homogeneity}
as the restriction of $\bP_{t-t'}$ to the range of $\bP_{t'}$. Since
$\bP_{t-t'}$ is convex-combination preserving and defined on the
entire simplex $\mathcal{S}_{N}$, this restriction can obviously
be extended back to all of $\mathcal{S}_{N}$. Hence, by Proposition~\ref{prop:linear-extension},
we obtain:
\begin{prop}
Any time-homogeneous linear probability dynamics is divisible. 
\end{prop}

\noindent However, without the additional assumption of time-homogeneity,
decomposability alone does not guarantee the required extension of
$\bP_{t\leftarrow t'}$ and, in turn, does not imply divisibility.

In light of this remark, it is worth taking a look at how the intuitive
content of ``divisibility'' is characterized in the literature (the quotations
below include minor adjustments to notation and equation references):
\begin{quote}
``{[}...{]} the stochastic process based on the transition matrix
$\Gamma(t\leftarrow t_{0})$ will generically fail to be Markovian,
so {[}...{]} the model will {[}...{]} lack specific dynamical laws
describing transitions between arbitrarily chosen intermediate times.''
(\citet{barandes2025}, p.~4) 

``In summary, for a Markov process the transition matrices {[}...{]}
satisfy the composition law of equation \eqref{eq:divis-Barandes}.
Essentially, equation \eqref{eq:divis-Barandes} states that the evolution
from $t_{0}$ to $t$ can be written as the composition of the evolution
from $t$ to some intermediate time $t'$, and from this $t'$ to
the final time $t$.'' (\citet{rivas2014}, p.~3)

``{[}...{]} it is natural to define classical divisibility as corresponding
to being able to describe the evolution of the probability density
$p\left(x_{t}\right)$, between two arbitrary times, by a classical
information channel. Any such classical channel, mapping $p\left(x_{t_{1}}\right)$
to $p\left(x_{t_{2}}\right)$, is defined by a stochastic matrix $S\left(x_{t_{1}},x_{t_{2}}\right)$
{[}...{]}.'' (\citet{li2018}, p.~41)
\end{quote}
As the first two quotations illustrate, the literature tends to blur the distinction between divisibility, a property of probability dynamics, and
Markovianity, a property of a stochastic process. However, the idea meant to be conveyed by equation
\eqref{eq:divis-Barandes} is clear: it corresponds precisely to what
we have introduced as decomposability. It is the notion that the probabilities
at any given time---not only at $t=0$---determine the probabilities
at later times, so that the transition of probabilities from $t=0$
to $t>0$ can be divided into intermediate steps. (Notice that neither
the above quotations nor the passage from \citet{barandes2025} at
the beginning of this section suggests that these intermediate steps
should satisfy time-homogeneity.)
However, what happens in the usual accounts is that this intuitive
conception is expressed through a formalism that does not quite match the intended idea, namely by the technical notion of divisibility, which, as Proposition~\ref{prop:linear-extension}
shows, is burdened with an additional requirement whose physical motivation is unclear.

The source of the conceptual problem, we believe, is well illustrated by the
third quote above from \citet{li2018}. By a classical information channel,
the authors mean any rule governing the temporal evolution of (classical)
probabilities. If one assumes, as discussed in Section~\ref{sec:Linearity},
that such a rule must by default be represented by a linear map, then---given
the initial probability vector can be chosen arbitrarily---it follows
that a classical channel is associated with a stochastic matrix. This
is precisely the assumption made by \citet{li2018}. The following
\emph{further} assumption is then implicitly made: when the initial
time is changed (from $0$ to $t'$), the corresponding classical channel should still
be represented by a stochastic matrix. However, as we have argued,
even granting the linearity assumption, this last step is problematic.
The problem becomes apparent once we recognize that linear evolution
is only a special case, and that therefore the matrix-based paradigm of probability
dynamics must be replaced by a more general perspective based on functions and function composition. In
the next section, we add further support to this view by presenting
specific examples of probability dynamics where the linearity assumption
is justified, yet decomposability and divisibility nonetheless diverge.

\subsection{Relating Markovianity and Decomposability}\label{sec:relating-Markovianity}

In the first part of the section, we saw that the formal analogue
of ``divisibility'' for a stochastic process---the satisfaction
of the Chapman-Kolmogorov equation---has to do with the Markov property.
In the second part, we argued that the stepwise composition of a probability
dynamics is most appropriately captured by what we have called decomposability.
This raises a natural question: is there a connection between Markovianity
and decomposability? The short answer is no: Markovianity is a property
of a stochastic process, while decomposability is a property of a
probability dynamics. However, something much stronger can in fact
be said. To formulate this precisely, we first introduce some definitions. 
\begin{defn}
We say that \emph{the implementation $\mm$ of a probability vector
trajectory $\vec{\pp}(t)$ is Markovian} iff $\mm$ is a Markovian
stochastic process; otherwise, we say that \emph{the implementation
is non-Markovian}. We say that a\emph{ probability dynamics $\bP$
has a Markovian implementation} iff there exists a stochastic process
family implementing $\bP$ which is Markovian for every solution.
We say that \emph{$\bP$ has a non-Markovian implementation} iff there
exists a stochastic process family implementing $\bP$ which is non-Markovian
for at least one solution. 
\end{defn}

\smallskip{}

\begin{defn}
We say that \emph{a probability vector trajectory $\vec{\pp}(t)$
is non-degenerate} iff there exist $t_{1}<t_{2}<t_{3}\in T$ and $i_{1},i_{3}\in C$
such that $0<\pp_{i_{1}}(t_{1})<1$ and $0<\pp_{i_{3}}(t_{3})<1$,
where $\pp_{i}(t)$ is the $i$th entry of $\vec{\pp}(t)$. We say
that \emph{probability dynamics $\bP$ is non-degenerate} iff at least
one solution of $\bP$ is non-degenerate. 
\end{defn}

\noindent One can show the following:

\begin{restatable}{prop}{Markovianimplementation}
\label{prop:Markovian-implementation}(a) Every probability vector
trajectory has a Markovian implementation. Hence, every probability
dynamics has a Markovian implementation.

(b) A probability vector trajectory has a non-Markovian implementation
if and only if it is non-degenerate. Hence, a probability dynamics
has a non-Markovian implementation if and only if it is non-degenerate. 
\end{restatable}

\noindent \emph{Sketch of proof.} Existence is shown by defining consistent finite-dimensional distributions with no inter-time correlations (products of the one-time marginals), and invoking Kolmogorov's extension theorem; the resulting process is Markov by construction. If the trajectory is non-degenerate, one can perturb a three-time joint law to violate the Markov conditional independence, whereas in the degenerate case the Markov condition becomes vacuous or undefined so no genuinely non-Markovian implementation can be found. (Formal proof: Appendix \ref{proof:Markovian-implementation}.)

There has been a discussion in the literature as to whether it is possible to find a realistic interpretation of the probability evolution arising from quantum mechanics in terms of an embedding into a classical Markov process  (\citet{gillespie1994, garbaczewski1996, hardy1997, gillespie2001}).
On one natural reading, our result answers this question in the affirmative. As we will see in the last section, the evolution of quantum probabilities is best characterized as a collection of probability vector trajectories, parametrized by the initial quantum state, and each such trajectory admits a Markovian implementation, as established by Proposition~\ref{prop:Markovian-implementation}.

\section{Statistical Dynamics}
\label{sec:statistical-dynamics} 

In the previous sections we argued that, contrary to common sentiment in the literature, there is no \emph{probability-theoretic} reason to assume that 
\begin{itemize}
\item[(a)] probability dynamics is linear, 
\item[(b)] the two types of transition matrices (when both exist) are
identical, 
\item[(c)] decomposability is equivalent to divisibility, 
\item[(d)] decomposability is equivalent to Markovianity. 
\end{itemize}
However, under a specific interpretation of the general formalism
we have developed, some of these assumptions can be given a \emph{physical}
motivation. In this section we look at this special case, which we
refer to as \emph{statistical dynamics}. We consider three types of
it in turn.

\subsection{Deterministic Statistical Dynamics}

Consider a deterministic dynamical system on configuration space $C$,
given by a map 
\begin{equation}
D:T\times C\rightarrow C\label{eq:deterministic-dynamics}
\end{equation}
such that $D_{0}$ is the identity function on $C$, where $D_{t}\doteq D\left(t,\cdot\right)$.
Imagine a large ensemble of copies of the system initiated from different
configurations and each evolving in time in accord with $D$. Let
$\pp_{0i}$ represent the relative frequency of copies in the ensemble
that start in configuration $i$. Then, since the initial configuration
of each copy uniquely determines its trajectory through configuration
space, the initial distribution $\vec{\pp}_{0}=\left(\pp_{01},...,\pp_{0N}\right)\in\mathcal{S}_{N}$
uniquely determines the distribution of trajectories over time in
the ensemble. To specify this distribution, first introduce the Dirac
measure $\mm_{i}^{D}$ concentrated at the trajectory with a given
initial configuration $i\in C$: 
\begin{equation}
\mm_{i}^{D}(E_{i_{1}}(t_{1})\wedge...\wedge E_{i_{n}}(t_{n}))=\begin{cases}
1 & \textrm{if }D(t_{1},i)=i_{1},...,D(t_{n},i)=i_{n}\\
0 & \textrm{otherwise}
\end{cases}\label{eq:det-joint-prob}
\end{equation}
for all $n\in\mathbb{N}$, $i_{1},...,i_{n}\in C$, and $t_{1},...,t_{n}\in T$.
The frequencies of the system's trajectories in the imagined ensemble
can now be written as a statistical mixture of the distributions $\mm_{i}^{D}$:\footnote{Alternatively, $\mm_{\vec{\pp}_{0}}^{D}$ can be defined directly
on the cylinder sets of $\mathcal{F}$ as follows: 
\begin{equation}
\mm_{\vec{\pp}_{0}}^{D}\left(\left\{ \omega\in\Omega~|~X_{t_{1}}\left(\omega\right)\in C_{1},X_{t_{2}}\left(\omega\right)\in C_{2},...,X_{t_{n}}\left(\omega\right)\in C_{n}\right\} \right)=\underset{i\in\bigcap_{k=1}^{N}D^{-1}\left(C_{k}\right)}{\sum}p_{0i}
\end{equation}
where $t_{1},t_{2},\dots,t_{n}\in T$, $C_{1},C_{2},...,C_{n}\subseteq C$,
$n\in\mathbb{N}$, and $D^{-1}\left(C_{k}\right)$ denotes the preimage
of $C_{k}$ under the map $D\left(t_{k},\cdot\right)$.} 
\begin{equation}
\mm_{\vec{\pp}_{0}}^{D}\doteq\sum_{i=1}^{N}\pp_{0i}\mm_{i}^{D}\label{eq:composite-measure1}
\end{equation}
From \eqref{eq:det-joint-prob} and \eqref{eq:composite-measure1},
a straightforward calculation shows that for the stochastic process
thus defined, we have 
\begin{align}
\vec{\mm}_{\vec{\pp}_{0}}^{D}\left(0\right) & =\vec{\pp}_{0}\label{eq:det-joint-prob0}\\
\left(\mM_{\vec{\pp}_{0}}^{D}\left(t\right)\right)_{ij} & =\mm_{\vec{\pp}_{0}}^{D}\left(E_{i}\left(t\right)|E_{j}\left(0\right)\right)=\begin{cases}
1 & \textrm{if }D(t,j)=i\\
0 & \textrm{if }D(t,j)\neq i
\end{cases}\label{eq:det-cond-prob}
\end{align}
Crucially, the entries in \eqref{eq:det-cond-prob} are independent
of $\vec{\pp}_{0}$.

The mapping $\bM^{D}:\vec{\pp}_{0}\mapsto\mm_{\vec{\pp}_{0}}^{D}$
defines a stochastic process family that gives rise to a probability
dynamics $\bP^{D}$, via the equation of implementation: 
\begin{equation}
\bP^{D}\left(t,\vec{\pp}_{0}\right)\doteq\vec{\mm}_{\vec{\pp}_{0}}^{D}\left(t\right)\label{eq:det-implement}
\end{equation}
for all $t\in T$ and $\vec{\pp}_{0}\in\mathcal{S}_{N}$. The probability
vector trajectory $t\mapsto\bP^{D}\left(t,\vec{\pp}_{0}\right)$ describes
how the relative frequencies of configurations evolve over time in
the ensemble, as each system copy follows its deterministic trajectory
from a given initial configuration. Importantly, since the entries
of $\mM_{\vec{\pp}_{0}}^{D}\left(t\right)$, whenever defined, are
independent of $\vec{\pp}_{0}$, $\bM^{D}$ is a transition-constant
stochastic process family, in the sense of Definition~\ref{def:transition-const}.
Hence, one can introduce the common transition matrix $\mM^{D}\left(t\right)$
as in \eqref{eq:common-M}, and, by Proposition~\ref{prop:transition-const-implies-linearity},
we obtain: 
\begin{prop}
\textup{\label{prop:det-linear}$\bP^{D}$} is linear, with $\mP^{D}\left(t\right)=\mM^{D}\left(t\right)$
for all $t\in T$. 
\end{prop}

The decomposability of $\bP^{D}$ and the Markovianity of $\bM^{D}$
can be characterized directly in terms of the underlying deterministic
dynamical system $D$. To this end, we begin with a definition.
\begin{defn}
\label{def:D-decomp}Let $\mathrm{Ran}D_{t}\subseteq C$ denote the
range of map $D_{t}$. We say that the family of maps $\left\{ D_{t\leftarrow t'}:\mathrm{Ran}D_{t'}\rightarrow C\right\} _{t,t'\in T,\,t'\leq t}$
\emph{decomposes} the deterministic dynamical system $D$ iff 
\begin{align}
D_{t} & =D_{t\leftarrow t'}\circ D_{t'}\label{eq:decomp-1}
\end{align}
for all  $t,t'\in T,\,t'\leq t$. We call $D$ \emph{decomposable}
iff such a family of maps exists. 
\end{defn}

\noindent Just like the decomposability of a probability dynamics, the decomposability
of a deterministic dynamical system captures the idea that the configuration
at time $t'$ encodes the ``dynamical state'' of the system at that
moment, in the sense that it determines all future configurations.
Since the configuration at any given time uniquely determines the
future trajectory, trajectories of a decomposable deterministic dynamics
do not intersect in configuration space. With this notion in hand,
we have the following result:

\begin{restatable}{prop}{detdecompmarkov}
\label{prop:det-decomp-markov}(a) $\bP^{D}$ is decomposable if and
only if $D$ is decomposable.

(b) \textup{$\bM^{D}$} is Markovian if and only if $D$ is decomposable. 
\end{restatable}

\noindent \emph{Sketch of proof.} The key observation is that $D$ is decomposable iff deterministic trajectories never merge and then split: one can then define intermediate-time maps $D_{t\leftarrow t'}$ on $\mathrm{Ran}\,D_{t'}$, which induce $0$--$1$ stochastic matrices and hence decomposing maps for $\bP^D$, while failure of this property is detected via Proposition~\ref{prop:decomp-characterization}. When $D$ is decomposable, the relevant conditionals are always in $\{0,1\}$ so further conditioning on earlier events cannot change them (entailing Markovianity), whereas if two initial states coincide at $t'$ but diverge at $t$ a suitable initial mixture yields a direct violation of the Markov condition. (Formal proof: Appendix \ref{proof:det-decomp-markov}.)

By Proposition~\ref{prop:det-decomp-markov}, $\bP^{D}$ is decomposable
if and only if $\bM^{D}$ is Markovian. Moreover, the proof of part~(a)
shows that, whenever $\bP^{D}$ is decomposable, there exists a family
$\left\{ \mP^{D}\left(t\leftarrow t'\right)\right\} _{t,t'\in T,\,t'\leq t}$
of stochastic matrices that satisfies the divisibility condition (cf. \eqref{eq:division-matrix-proof} and \eqref{eq:composition-matrix-proof}).
Thus: 
\begin{prop}
\label{prop:det-divis}For any deterministic dynamical system $D$,
$\bP^{D}$ is decomposable if and only if it is divisible.
\end{prop}

What Proposition~\ref{prop:det-linear}, \ref{prop:det-decomp-markov}
and \ref{prop:det-divis} together establish is that all of the
assumptions (a)--(d) listed at the beginning of the section are satisfied
in the context of a probability dynamics in which probabilities represent
frequencies of configurations in a statistical mixture of deterministic
trajectories. To mark this special case, we will refer to $\bP^{D}$
as \emph{deterministic statistical dynamics}.

\subsection{System-Ancilla Statistical Dynamics}

A deterministic statistical dynamics represents a very special class
of probability dynamics in which the entries of the transition
matrices  $\mP\left(t\right)$ take only the values 0 and 1 (see \eqref{eq:det-cond-prob}). One may obtain more general stochastic transition matrices by embedding a target system into a larger, but still deterministic, system whose additional degrees of freedom are taken to be ``unobserved'' or ``unknown'' (cf.~\citet{spekkens2019}, p.~7). To describe
this scenario within our framework, suppose that a target system
is coupled with an ancillary (``environment''). The set of possible
configurations of the ancilla is assumed to be $\Lambda=\left\{ 1,2,...,M\right\} $.
Suppose that the system-ancilla composite is subjected to a deterministic
dynamics, in such a way that the configuration of the target system
at $t\neq0$ is uniquely determined by the joint configuration of
the target system plus the ancilla at $t=0$. This is described by
a map 
\begin{equation}
\SA:T\times C\times\Lambda\rightarrow C
\end{equation}
with the property that for any $\alpha\in\Lambda$, $\SA(0,\cdot,\alpha)$
is the identity function on $C$.

We now consider an ensemble of copies of the system-ancilla composite,
each prepared in some initial configuration and evolving in time according
to $\SA$. The distribution of initial configurations across
the ensemble is specified by a probability vector $\vec{\pi}\in\mathcal{S}_{NM}$,
representing a probability distribution on $C\times\Lambda$.

For a given initial configuration $i\in C$ and $\alpha\in\Lambda$
of the system-ancilla composite, the evolution of the target system
is uniquely determined. This defines a Dirac measure on the set of
trajectories of the target system, concentrated at the trajectory corresponding
to the initial configuration pair $\left(i,\alpha\right)$: 
\begin{equation}
\mm_{i,\alpha}^{\SA}\left(E_{i_{1}}\left(t_{1}\right)\wedge...\wedge E_{i_{n}}\left(t_{n}\right)\right)=\begin{cases}
1 & \textrm{if }\SA\left(t_{1},i,\alpha\right)=i_{1},...,\SA\left(t_{n},i,\alpha\right)=i_{n}\\
0 & \textrm{otherwise}
\end{cases}\label{eq:comp-joint-prob-rel2}
\end{equation}
for all $n\in\mathbb{N}$, $i_{1},...,i_{n}\in C$, and $t_{1},...,t_{n}\in T$.
The frequencies of the target system's trajectories in the ensemble
can then be written as a statistical mixture of these $\mm_{i,\alpha}^{\SA}$:
\begin{equation}
\mm_{\vec{\pi}}^{\SA}\doteq\sum_{i=1}^{N}\sum_{\alpha=1}^{M}\pi_{i,\alpha}\mm_{i,\alpha}^{\SA}\label{eq:composite-measure2}
\end{equation}
where $\pi_{i,\alpha}$ is the entry of $\vec{\pi}$ giving the probability
that the target system is in configuration $i\in C$ and the ancilla
is in configuration $\alpha\in\Lambda$. For this stochastic process
$\mm_{\vec{\pi}}^{\SA}$, equations \eqref{eq:comp-joint-prob-rel2}-\eqref{eq:composite-measure2}
yield 
\begin{align}
\vec{\mm}_{\vec{\pi}}^{\SA}\left(0\right) & =\vec{\pp}_{0}\label{eq:comp-p0-1}\\
\left(\mM_{\vec{\pi}}^{\SA}\left(t\right)\right)_{ij} & =\mm_{\vec{\pi}}^{\SA}\left(E_{i}\left(t\right)|E_{j}\left(0\right)\right)=\frac{\underset{\underset{\SA\left(t,j,\alpha\right)=i}{\alpha}}{\sum}\pi_{j,\alpha}}{p_{0j}}\label{eq:comp-cond-prob-1}
\end{align}
Here $\vec{\pp}_{0}$ is the marginal distribution of the target system
obtained from $\vec{\pi}$ via 
\begin{equation}
\pp_{0i}\doteq\sum_{\alpha=1}^{M}\pi_{i,\alpha}\label{eq:marginal}
\end{equation}
for all $i\in C$. Note that, although the overall system is governed
by a deterministic dynamics $\SA$, the matrix entries in
\eqref{eq:comp-cond-prob-1} are no longer restricted to $\left\{ 0,1\right\} $;
and in general they depend on the initial distribution $\vec{\pi}$,
specifically through its marginal $\vec{\pp}_{0}$.

Suppose now that the initial distribution of the system-ancilla composite
is not fixed but may vary. Instead of a single distribution $\vec{\pi}\in\mathcal{S}_{NM}$,
consider a family 
\begin{equation}\left\{ \vec{\pi}_{\vec{\pp}_{0}}\in\mathcal{S}_{NM}\right\} {}_{\vec{\pp}_{0}\in\mathcal{S}_{N}}
\end{equation}
The subscript $\vec{\pp}_{0}$ indicates that the marginal distribution
of the target system obtained from $\vec{\pi}_{\vec{\pp}_{0}}$, in
the sense of \eqref{eq:marginal}, is $\vec{\pp}_{0}$. That is, the
family specifies a joint initial distribution for the composite for
each possible initial distribution of the target.

In this setting, the mapping\footnote{If the family of initial distributions were not chosen to be of the
form $\left\{ \vec{\pi}_{\vec{\pp}_{0}}\in\mathcal{S}_{NM}\right\} {}_{\vec{\pp}_{0}\in\mathcal{S}_{N}}$,
there would be no guarantee that $\bM^{\SA}$ yields a well-defined
map, that is a unique assignment determined by the target system’s
initial distribution $\vec{\pp}_{0}$. This is because in general
different initial distributions $\vec{\pi}$ of the system-ancilla
composite can share the same marginal $\vec{\pp}_{0}$ on the system.
A related point is made by \citet{spekkens2019}, Sec.~III/E, but
with the diagnosis that the initial distribution of the system-ancilla
composite must be chosen so that the system and ancilla are uncorrelated.
We discuss this special case below.} $\bM^{\SA}:\vec{\pp}_{0}\mapsto\mm_{\vec{\pp}_{0}}^{\SA}\doteq\mm_{\vec{\pi}_{\vec{\pp}_{0}}}^{\SA}$
defines a stochastic process family. This family, in turn, induces
a probability dynamics $\bP^{\SA}$ through the equation of
implementation, 
\begin{equation}
\bP^{\SA}\left(t,\vec{\pp}_{0}\right)\doteq\vec{\mm}_{\vec{\pp}_{0}}^{\SA}\left(t\right)\label{eq:det-implement-12}
\end{equation}
for all $t\in T$ and for all $\vec{\pp}_{0}\in\mathcal{S}_{N}$.
$\bP^{\SA}$ will be referred to as \emph{system-ancilla statistical
dynamics}.

In contrast to $\bM^{D}$ and $\bP^{D}$, which are constrained by Propositions~\ref{prop:det-linear}, \ref{prop:det-decomp-markov}, and \ref{prop:det-divis}, $\bM^{\SA}$ and $\bP^{\SA}$ are far more general. 
This is made precise by the following simple result: 

\begin{restatable}{prop}{thereexistsenvironment}
\label{prop:thereexistsenvironment} Let $\bM: \vec{\pp}_0 \mapsto \mm_{\vec{\pp}_0}$ be an arbitrary stochastic
process family with $C=\left\{ 1,...,N\right\} $ and $T=\left\{ 0,1,...,\tau\right\} $
finite. Then there exists an ancillary system $\Lambda=\left\{ 1,...,M\right\} $,
a deterministic system-ancilla dynamics $\SA:T\times C\times\Lambda\rightarrow C$,
and a family of probability distributions $\left\{ \vec{\pi}_{\vec{\pp}_{0}}\in\mathcal{S}_{NM}\right\} {}_{\vec{\pp}_{0}\in\mathcal{S}_{N}}$
such that $\bM^{\SA}=\bM$. 
\end{restatable}

\noindent {\em Sketch of proof.} A large enough ancilla can label a complete system trajectory, allowing a probability distribution over the joint system-ancilla degrees of freedom to reproduce any stochastic process. (Formal proof: Appendix \ref{proof:thereexistsenvironment}.)

Proposition~\ref{prop:thereexistsenvironment} shows that any stochastic
process family $\bM$ (at least when its time index set is finite)
can, in principle, be realized as a system-ancilla process family
$\bM^{\SA}$. Since, by Proposition~\ref{prop:implementsalways},
every probability dynamics is implemented by a stochastic process
family, it follows that any probability dynamics $\bP$ (again, with
a finite time index set) can be represented as a system-ancilla statistical
dynamics, in the sense that $\bP^{\SA}=\bP$. As we will now
see, restrictions arise only when constraints are imposed on the initial
probability distributions of the system-ancilla composite.

Suppose that the initial configurations of the system and the ancilla
are \emph{uncorrelated} and that the probability distribution of the
ancilla can be kept \emph{fixed}. That is, assume there exists a fixed
probability distribution $\vec{\lambda}_{0}\in\mathcal{S}_{M}$ such
that 
\begin{equation}
(\pi_{\vec{\pp}_{0}})_{i,\alpha}=\pp_{0i}\lambda_{0\alpha}\label{eq:environmentsystemindependence}
\end{equation}
for all $i\in C$, $\alpha\in\Lambda$, and for all $\vec{\pp}_{0}\in\mathcal{S}_{N}$.
Substituting into \eqref{eq:composite-measure2} gives 
\begin{equation}
\mm_{\vec{\pp}_{0}}^{\SA}=\sum_{i=1}^{N}\sum_{\alpha=1}^{M}\pp_{0i}\lambda_{0\alpha}\mm_{i,\alpha}^{\SA}\label{eq:comp-measure2}
\end{equation}
and substituting into \eqref{eq:comp-cond-prob-1} yields 
\begin{equation}
\left(\mM_{\vec{\pp}_{0}}^{\SA}\left(t\right)\right)_{ij}=\mm_{\vec{\pp}_{0}}^{\SA}\left(E_{i}\left(t\right)|E_{j}\left(0\right)\right)=\underset{\underset{\SA\left(t,j,\alpha\right)=i}{\alpha}}{\sum}\lambda_{0\alpha}\label{eq:comp-cond-prob2}
\end{equation}
Crucially, \eqref{eq:comp-cond-prob2} depends only on $\vec{\lambda}_{0}$
and is independent of $\vec{\pp}_{0}$. Hence, $\bM^{\SA}$
is transition-constant, and one can introduce the common matrix $\mM^{\SA}\left(t\right)$
as in \eqref{eq:common-M}. By Proposition~\ref{prop:transition-const-implies-linearity},
it follows that: 
\begin{prop}
\label{prop:comp-linear}Suppose that the initial probability distributions
of the system and the ancilla are independent in the sense of equation~\eqref{eq:environmentsystemindependence}.
Then \textup{$\bP^{\SA}$} is linear and $\mP^{\SA}\left(t\right)=\mM^{\SA}\left(t\right)$,
for all $t\in T$. 
\end{prop}

Thus, a \emph{system-ancilla} statistical dynamics with independent
system and ancilla generalizes the concept of \emph{deterministic}
statistical dynamics: it yields a linear probability dynamics implemented
by a transition-constant stochastic process family, in which the transition
matrices are no longer restricted to entries 0 and 1.

\subsection{Stochastic Statistical Dynamics}

Another natural way to generalize single-system deterministic statistical dynamics is to allow the underlying dynamics itself to be stochastic rather than deterministic. In analogy with map $D$, such a dynamical system can be characterized
by a function 
\begin{equation}
S:T\times C\rightarrow\mathcal{S}_{N}
\end{equation}
with $S\left(0,\cdot\right):C\rightarrow\mathcal{S}_{N},i\mapsto\vec{e}_{i}$.
Here the $j$th entry $S_{j}\left(t,i\right)$ of $S\left(t,i\right)$
is interpreted as the probability that the system is in configuration
$j$ at time $t$, given that it started in configuration $i$. However,
if this probability is to be understood as a two-time conditional
probability defined by Bayes' rule---that is, as a ``transition
probability'' of type 1 in the terminology of Section~\ref{sec:Transition}---then
the map $S$ itself is not the appropriate mathematical object. As
discussed in Section~\ref{sec:twodescriptions}, this is because
the elements of the range of $S$, vectors in $\mathcal{S}_{N}$,
encode only one-time probabilities, whereas two-time conditional probabilities
require reference to two-time joint probabilities. Accordingly, the
stochastic system is more aptly represented by a collection of stochastic
processes. Specifically, consider $N$ stochastic processes, $\mm_{i}^{S}$
($i=1,...,N$), with
\begin{align}
\vec{\mm}_{i}^{S}\left(t\right) & =S\left(t,i\right)\label{eq:Si-process}
\end{align}
for all $t\in T$ and $i=1,...,N$. (By Proposition~\ref{prop:Markovian-implementation}, such processes always exist.) Since $S\left(0,i\right)=\vec{e}_{i}$,
each stochastic process $\mm_{i}^{S}$ describes the distribution
of trajectories in configuration space for a system initialized in
configuration $i$ with probability one. For each such process, the
two-time conditional probabilities satisfy 
\begin{equation}
\mm_{i}^{S}\left(E_{j}\left(t\right)|E_{i}\left(0\right)\right)=\mm_{i}^{S}\left(E_{j}\left(t\right)\right)=S_{j}\left(t,i\right)
\end{equation}
for all $t\in T$ and $j=1,...,N$. Accordingly, the transition
matrix $\mM_{i}^{S}\left(t\right)$ associated with $\mm_{i}^{S}$
has its $i$th column equal to the probability vector $S\left(t,i\right)$;
while all other entries of it are undefined, as they would correspond
to conditional probabilities with zero-probability conditioning events
(since $\mm_{i}^{S}\left(E_{j}\left(0\right)\right)=0$ for $j\neq i$).
Note that the Dirac measure $\mm_{i}^{D}$ defined in \eqref{eq:det-joint-prob}
is a special case of such a stochastic process $\mm_{i}^{S}$: for
every $t\in T$, the vector $\vec{\mm}_{i}^{S}\left(t\right)$, and
hence the well-defined column of $\mM_{i}^{S}\left(t\right)$, is
one of the standard basis vectors. Unlike $\mm_{i}^{D}$, which is
uniquely determined by the map $D$, a general stochastic process
$\mm_{i}^{S}$ is not determined by the map $S$ alone (again, since
$S$ encodes only one-time probabilities, but no multi-time probabilities).
In what follows, we therefore take the stochastic dynamical system
to be described directly by the collection of processes $\mm_{1}^{S},...,\mm_{N}^{S}$
rather than by the map $S$ itself. To indicate the underlying stochastic
dynamics, we will henceforth use the symbol $S$ not for a map but
for this collection of stochastic processes; thus $S=\left(\mm_{1}^{S},...,\mm_{N}^{S}\right)$.
In this setting, the only requirement imposed on these processes is
that
\begin{equation}
\mm_{i}^{S}\left(E_{j}\left(0\right)\right)=\begin{cases}
1 & \textrm{if }j=i\\
0 & \textrm{if }j\neq i
\end{cases}\label{eq:feature-mu-i-S}
\end{equation}
for all $i,j=1,...,N$.

Now consider a statistical ensemble of copies of the system, each
copy being initialized in a definite configuration and evolving in
time according to the stochastic dynamics $S$. More precisely, suppose
that 1)~$\pp_{0i}$ specifies the relative frequency of copies in
the ensemble that start in configuration $i$; and 2)~within each
subensemble of systems that start from a given configuration $i$,
the frequency distribution of trajectories over time is given by $\mm_{i}^{S}$.
Then the statistical mixture
\begin{equation}
\mm_{\vec{\pp}_{0}}^{S}\doteq\sum_{i=1}^{N}\pp_{0i}\mm_{i}^{S}\label{eq:mixture-stochastic}
\end{equation}
with $\vec{\pp}_{0}=\left(\pp_{01},...,\pp_{0N}\right)\in\mathcal{S}_{N}$,
describes the distribution of trajectories for the ensemble as a whole.
Note that $\mm_{\vec{\pp}_{0}}^{D}$, as given by \eqref{eq:composite-measure1},
is a special case of this mixture, obtained by setting $\mm_{i}^{S}=\mm_{i}^{D}$
for all $i=1,...,N$.

From \eqref{eq:feature-mu-i-S}, a straightforward calculation shows
that for the stochastic process $\mm_{\vec{\pp}_{0}}^{S}$ we have
\begin{align}
\vec{\mm}_{\vec{\pp}_{0}}^{S}\left(0\right) & =\vec{\pp}_{0}\\
\left(\mM_{\vec{\pp}_{0}}^{S}\left(t\right)\right)_{ij} & =\left(\mM_{j}^{S}\left(t\right)\right)_{ij}\label{eq:matrix-M-p0-S}
\end{align}
The entries in \eqref{eq:matrix-M-p0-S} can now take values other
than 0 and 1, but, crucially, remain independent of $\vec{\pp}_{0}$.
Hence, the mapping $\bM^{S}:\vec{\pp}_{0}\mapsto\mm_{\vec{\pp}_{0}}^{S}$
yields a stochastic process family that is transition-constant, with
a common matrix $\mM^{S}\left(t\right)$ specified as in \eqref{eq:common-M}.
We then introduce probability dynamics $\bP^{S}$ via the equation
of implementation: 
\begin{equation}
\bP^{S}\left(t,\vec{\pp}_{0}\right)\doteq\vec{\mm}_{\vec{\pp}_{0}}^{S}\left(t\right)\label{eq:det-implement-1}
\end{equation}
for all $t\in T$ and $\vec{\pp}_{0}\in\mathcal{S}_{N}$. The probability
vector trajectory $t\mapsto\bP^{S}\left(t,\vec{\pp}_{0}\right)$ describes
how the relative frequencies of configurations evolve over time in
the ensemble, as each system copy is driven through configuration
space by the stochastic dynamics $S$. Since $\bM^{S}$ is transition-constant,
Proposition~\ref{prop:transition-const-implies-linearity} implies
the following.
\begin{prop}
\textup{\label{prop:stochastic-linear}$\bP^{S}$} is linear, with
$\mP^{S}\left(t\right)=\mM^{S}\left(t\right)$ for all $t\in T$.
\end{prop}

The probabilities appearing in a probability dynamics $\bP^{S}$ originate
from two conceptually distinct sources. The first source is the stochastic
nature of the underlying dynamics governing the evolution of individual
system copies; this aspect is encoded in the collection of stochastic
processes $\mm_{i}^{S}$. The second source is the statistical character
of the ensemble itself: we consider a collection of such copies whose
initial configurations are distributed according to $\vec{\pp}_{0}$.
Importantly, the fact that the resulting probability dynamics is linear
has nothing to do with the first source---indeed, it holds regardless
of whether the underlying dynamics is deterministic or stochastic---but
everything to do with the second source. Accordingly, we will refer to $\bP^{S}$ (of which $\bP^{D}$
is a special case) as \emph{stochastic statistical dynamics}.

\subsection{Relating System-Ancilla and Stochastic Statistical Dynamics}\label{subsec:relating}

As we have seen, \emph{stochastic} statistical dynamics generalizes the concept of \emph{deterministic} statistical dynamics in the same way as \emph{system-ancilla} statistical dynamics with independent system and ancilla does: it yields a linear probability dynamics implemented by a transition-constant stochastic process family, in which the transition matrices are no longer restricted to entries 0 and 1.

In fact, there is a close correspondence between stochastic statistical
dynamics and system-ancilla statistical dynamics with independent
system and ancilla. To see this, consider a deterministic system-ancilla
dynamics $\SA$, and an initial probability distribution $\vec{\lambda}_{0}\in\mathcal{S}_{M}$
over the ancilla configurations. These data induce a stochastic dynamical
system $S$ on the system configurations by
\begin{equation}
\mm_{i}^{S} = \sum_{\alpha=1}^{M}\lambda_{0\alpha}\mm_{i,\alpha}^{\SA}\label{eq:S-from-D}
\end{equation}
Since $\SA(0,\cdot,\alpha)$ is the identity function, it
follows that
\begin{equation}
\mm_{i,\alpha}^{\SA}\left(E_{j}\left(0\right)\right)=\begin{cases}
1 & \textrm{if }i=j\\
0 & \textrm{if }i\neq j
\end{cases}
\end{equation}
and hence that each $\mm_{i}^{S}$ satisfies the defining requirement
\eqref{eq:feature-mu-i-S}. As introduced in \eqref{eq:S-from-D},
$\mm_{i}^{S}$ represents a statistical mixture of system trajectories
starting from configuration $i$, where the mixing is taken over the
``unknown'' initial configuration of the ancilla. With this identification
of $\mm_{i}^{S}$, and under the assumption that the system and ancilla
are initially independent, we obtain, for all $\vec{\pp}_{0}\in\mathcal{S}_{N}$,
\begin{equation}
\mm_{\vec{\pp}_{0}}^{\SA}=\sum_{i=1}^{N}\sum_{\alpha=1}^{M}\pp_{0i}\lambda_{0\alpha}\mm_{i,\alpha}^{\SA}=\sum_{i=1}^{N}\pp_{0i}\mm_{i}^{S}=\mm_{\vec{\pp}_{0}}^{S}\label{eq:D-S-correspondence}
\end{equation}
Here we used \eqref{eq:S-from-D}, the independence condition \eqref{eq:environmentsystemindependence},
and definitional equations \eqref{eq:composite-measure2} and \eqref{eq:mixture-stochastic}.

Thus, a deterministic evolution of the system-ancilla composite manifests
itself, at the level of the system alone, as a stochastic law of evolution.
One may then ask, conversely, whether any apparent stochasticity in
the evolution of a target system can, in principle, be accounted for
by deterministic interactions with ``unknown'' degrees of freedom of an ancilla with initially independent, fixed distribution. The answer is affirmative (at least for finite time index
sets):

\begin{restatable}{prop}{thereexistsenvironmentS}
\label{prop:thereexistsenvironment-S} Let $S$ be an arbitrary stochastic
dynamical system with $C=\left\{ 1,...,N\right\} $ and $T=\left\{ 0,1,...,\tau\right\} $
finite. Then there exists an ancillary system $\Lambda=\left\{ 1,...,M\right\} $,
a deterministic system-ancilla dynamics $\SA:T\times C\times\Lambda\rightarrow C$,
and a probability distribution $\vec{\lambda}_{0}\in\mathcal{S}_{M}$,
such that

(i) $\mm_{i}^{S}=\sum_{\alpha=1}^{M}\lambda_{0\alpha}\mm_{i,\alpha}^{\SA}$,
for all $i\in C$;

(ii) with the family of uncorrelated distributions $\left\{ \vec{\pi}_{\vec{\pp}_{0}}\right\} {}_{\vec{\pp}_{0}\in\mathcal{S}_{N}}$
defined by \eqref{eq:environmentsystemindependence}, $\bM^{\SA}=\bM^{S}$.
\end{restatable}

\noindent \emph{Sketch of proof.} A stochastic dynamical system is composed of $N$ stochastic processes, and we can let a large enough ancilla be composed of $N$ blocks, each of which encodes a full configuration trajectory conditional on the system starting at $i$, and let $\SA(t,i,\alpha)$ read off the $t$-th entry of block $i$.  Choose the ancilla distribution $\vec\lambda_0$ so that block $i$ is distributed according to $\mm_i^S$, so averaging over the ancilla reproduces each $\mm_i^S$, and with $(\pi_{\vec{\pp}_0})_{i,\alpha}=\pp_{0i}\lambda_{0\alpha}$ this yields $\bM^{\SA}=\bM^S$. (Formal proof: Appendix \ref{proof:thereexistsenvironment-S}.)

This result shows that stochastic statistical
dynamics and system-ancilla statistical dynamics with independent
system and ancilla, viewed as linear probability dynamics, have the same degree of generality. In light of this equivalence, it will be convenient to refer to these probability dynamics collectively as \emph{statistical dynamics}. A natural question is whether any linear probability dynamics can, in principle, be realized as a statistical dynamics. The answer is affirmative:

\begin{restatable}{prop}{thereexistsS}
\label{prop:thereexists-S}Let $\bP$ be an arbitrary linear probability
dynamics. Then there exists a stochastic dynamical system $S$ such
that $\bP^{S}=\bP$.
\end{restatable}

\noindent \emph{Sketch of proof.} For each pure initial condition $\vec e_i$, use Proposition~\ref{prop:Markovian-implementation} to pick a stochastic process $\mm_i$ whose one-time marginals reproduce the trajectory $t\mapsto \bP(t,\vec e_i)$, and define the stochastic dynamics $S$ by setting $\mm_i^S:=\mm_i$. For a general initial distribution $\vec{\pp}_0$, take the mixture $\mm_{\vec{\pp}_0}^S:=\sum_i p_{0i}\mm_i$, so $\vec\mm_{\vec{\pp}_0}^S(t)=\sum_i p_{0i}\bP(t,\vec e_i)=\bP(t,\vec{\pp}_0)$ by linearity, hence $\bP^S=\bP$. (Formal proof: Appendix \ref{proof:thereexists-S}.)

Proposition~\ref{prop:thereexistsenvironment-S} implies that any
statistical dynamics $\bP^{S}$ (at least when its time index set
is finite) is implemented by a stochastic process family of type $\bM^{\SA}$
satisfying the independence condition \eqref{eq:environmentsystemindependence}.  Together with Proposition~\ref{prop:thereexists-S},  this entails that any linear probability
dynamics $\bP$ (at least when its time index set is finite) can be
realized as a system-ancilla statistical dynamics satisfying the independence
condition. This claim is made precise in the following proposition.
\begin{prop}
\label{prop:thereexistsenvironment-lin} Let $\bP$ be a linear probability
dynamics with $T=\left\{ 0,1,...,\tau\right\} $ finite. Then there
exists an ancillary system $\Lambda=\left\{ 1,...,M\right\} $, a
deterministic system-ancilla dynamics $\SA:T\times C\times\Lambda\rightarrow C$,
and a probability distribution $\vec{\lambda}_{0}\in\mathcal{S}_{M}$,
such that, with the family of distributions $\left\{ \vec{\pi}_{\vec{\pp}_{0}}\right\} {}_{\vec{\pp}_{0}\in\mathcal{S}_{N}}$
defined by \eqref{eq:environmentsystemindependence}, we have $\bP^{\SA}=\bP$.
\end{prop}

Thus, statistical dynamics correspond to the most general class of linear probability dynamics.
This generality comes at a price: while these dynamics satisfy conditions
(a)--(b) listed at the beginning of the section, they---unlike a
single-system deterministic statistical dynamics $\bP^{D}$---fail
to satisfy condition (c), the equivalence of decomposability and divisibility.
The reason is that they can realize linear probability dynamics of
all kinds, including those that are decomposable but not divisible. 

Moreover, by Proposition~\ref{prop:Markovian-implementation}, every
probability dynamics admits a Markovian implementation, and every
non-degenerate probability dynamics admits a non-Markovian implementation
as well. \emph{A fortiori}, every decomposable linear probability
dynamics that is non-degenerate therefore has both types of implementations.
Since no assumptions are made about the Markovianity of the stochastic
processes $\mm_{i}^{S}$ from which $\bM^{S}$ (and, by equivalence,
$\bM^{\SA}$) arises, one should expect on the basis of Proposition~\ref{prop:Markovian-implementation}
that a decomposable linear dynamics will generally admit both Markovian and non-Markovian
implementations \emph{of type $\bM^{S}$ and $\bM^{\SA}$};
and hence that the equivalence between decomposability and Markovianity
formulated in point (d) also breaks down in this more general setting.
The following example illustrates that the equivalences
stated in points (c) and (d) indeed fail for statistical dynamics more general than single-system deterministic ones.

\begin{example}
Consider the probability dynamics introduced in Example~\ref{exa:counterexample},
with $N=2,T=\left\{ 0,1,2\right\} $, and 
\begin{gather}
\mP\left(1\right)=\left(\begin{array}{cc}
1 & \frac{1}{2}\\
0 & \frac{1}{2}
\end{array}\right),\,\,\,\,\,\mP\left(2\right)=\left(\begin{array}{cc}
\frac{1}{2} & 1\\
\frac{1}{2} & 0
\end{array}\right)
\end{gather}
As shown there, $\bP$ is decomposable but not divisible. We now construct
a system-ancilla dynamics $\SA$, together with an initial
probability distribution $\vec{\lambda}_{0}\in\mathcal{S}_{M}$, and
the corresponding stochastic dynamics $S$, such that $\bP^{\SA}=\bP^{S}=\bP$.

Let the configuration space of the ancilla be $\Lambda=\left\{ 1,2\right\} $,
and define dynamics $\SA$ by stipulating the values of $\SA\left(t,i,\alpha\right)$
as follows:
\begin{equation}
\begin{array}{l@{\qquad\qquad}l}
\SA\left(1,1,1\right)=1 & \SA\left(1,1,2\right)=1\\[6pt]
\SA\left(1,2,1\right)=1 & \SA\left(1,2,2\right)=2\\[6pt]
\SA\left(2,1,1\right)=1 & \SA\left(2,1,2\right)=2\\[6pt]
\SA\left(2,2,1\right)=1 & \SA\left(2,2,2\right)=1
\end{array}
\end{equation}
and choose the initial ancilla distribution to be uniform, $\lambda_{01}=\lambda_{02}=\frac{1}{2}$.
The induced stochastic dynamics $S$ can then be described by specifying
the probabilities assigned by $\mm_{1}^{S}$ and $\mm_{2}^{S}$ to
the three-time events $E_{i}\left(0\right)\wedge E_{j}\left(1\right)\wedge E_{k}\left(2\right)$
(trajectories), in accordance with \eqref{eq:S-from-D}. There are
eight such events; the nonzero probabilities are:
\begin{equation}
\begin{array}{l@{\qquad\qquad}l}
\mm_{1}^{S}\left(E_{1}\left(0\right)\wedge E_{1}\left(1\right)\wedge E_{1}\left(2\right)\right)=\frac{1}{2} & \mm_{2}^{S}\left(E_{2}\left(0\right)\wedge E_{1}\left(1\right)\wedge E_{1}\left(2\right)\right)=\frac{1}{2}\\[6pt]
\mm_{1}^{S}\left(E_{1}\left(0\right)\wedge E_{1}\left(1\right)\wedge E_{2}\left(2\right)\right)=\frac{1}{2} & \mm_{2}^{S}\left(E_{2}\left(0\right)\wedge E_{2}\left(1\right)\wedge E_{1}\left(2\right)\right)=\frac{1}{2}\\[6pt]
\end{array}\label{eq:3-time-joints}
\end{equation}

Using either \eqref{eq:matrix-M-p0-S} or \eqref{eq:comp-cond-prob2},
one verifies directly that
\begin{equation}
\mM^{\SA}\left(t\right)=\mM^{S}\left(t\right)=\mP\left(t\right)
\end{equation}
for all $t\in T$. As a consequence, $\mP^{\SA}\left(t\right)=\mP^{S}\left(t\right)$
is a statistical dynamics that is decomposable but not divisible.

Moreover, using \eqref{eq:mixture-stochastic} and \eqref{eq:3-time-joints},
a straightforward calculation shows that 
\begin{gather}
\mm_{\vec{\pp}_{0}}^{S}\left(E_{1}\left(2\right)|E_{1}\left(1\right)\wedge E_{1}\left(0\right)\right)=\frac{1}{2}\\
\mm_{\vec{\pp}_{0}}^{S}\left(E_{1}\left(2\right)|E_{1}\left(1\right)\wedge E_{2}\left(0\right)\right)=1
\end{gather}
for all initial distributions $\vec{\pp}_{0}\in\mathcal{S}_{N}$,
$\pp_{01},\pp_{02}>0$. This means that for every such $\vec{\pp}_{0},$
the process $\mm_{\vec{\pp}_{0}}^{S}$ (and equivalently $\mm_{\vec{\pp}_{0}}^{\SA}$)
is non-Markovian, despite the fact that the associated statistical
dynamics $\mP^{S}\left(t\right)=\mP^{\SA}\left(t\right)$
is decomposable.
\end{example}

\subsection{Statistical Dynamics and Statistical Mixing}
\label{sec:mixing}

It must be emphasized that the linearity of statistical dynamics is
not merely a consequence of general probabilistic constraints imposed
by Kolmogorov’s axioms, as standard presentations of
linearity cited in Section~\ref{sec:Linearity} would like to have it. Rather, it  depends
crucially on the physical interpretation of what probability vectors
represent in the given context: the instantaneous frequency distributions
of configurations in an ensemble of systems, each evolving independently
according to a deterministic or stochastic law of evolution. Conversely,
the fact that statistical mixing of deterministic or stochastic processes
leads to linear probability evolution provides no grounds for assuming
that an arbitrary probability dynamics, \emph{with no connection to
statistical mixing}, must be linear. To illustrate this distinction,
it is instructive to consider the following example.

\begin{example}
\label{exa:statistical-linearity}Consider a coin toss, with configuration
space $C\doteq\{H,T\}$. A probability vector of the corresponding
type, say $\vec{\pp}=\left(1/2,1/2\right)$, admits two distinct interpretations.
The vector $\vec{\pp}$ may represent the probabilistic state of
\begin{itemize}
\item[(a)] an individual coin that is fair, or
\item[(b)] an ensemble of coins, half of which always land heads and half of
which always land tails.
\end{itemize}
Formally, in both cases $\vec{\pp}$ defines the same probability distribution
on $C$, which can be expressed as the convex combination 
\begin{equation}
\left(\frac{1}{2},\frac{1}{2}\right)=\frac{1}{2}\left(1,0\right)+\frac{1}{2}\left(0,1\right)\label{eq:equal-mixture}
\end{equation}
However, the \emph{meaning} of this distribution differs in the two
cases. In case (b), it is a \emph{statistical mixture} corresponding
to the expansion \eqref{eq:equal-mixture}, in the sense that it describes
an ensemble that is an equal mixture of maximally biased ``heads''
and ``tails'' coins. In case (a), by contrast, although the expansion
\eqref{eq:equal-mixture} is mathematically valid, it does not describe a mixture of underlying physical systems.

Correspondingly, there are two different probability dynamics governing
the evolution of the probability vector for these two systems. Let
$T=\left\{ 0,1,2\right\} $ and consider the following cases.
\begin{itemize}
\item[(A)]  \emph{Individual coin.} The bias of the coin evolves according to
the probability dynamics $\bP$ defined in the Introduction:
\begin{eqnarray}
\bP_{0}\left(r,1-r\right) & \doteq & \left(r,1-r\right)\label{eq_probdyn1-1}\\
\bP_{1}\left(r,1-r\right) & \doteq & \left(f\left(r\right),1-f\left(r\right)\right)\\
\bP_{2}\left(r,1-r\right) & \doteq & \left(r,1-r\right)\label{eq_probdyn3-1}
\end{eqnarray}
where $r\in\left[0,1\right]$ and $f\left(r\right)=r^{2}$. Clearly,
$\bP$ is nonlinear, since
\begin{align}
\bP_{1}\left(\frac{1}{2},\frac{1}{2}\right) & =\left(\frac{1}{4},\frac{3}{4}\right)\neq\left(\frac{1}{2},\frac{1}{2}\right)=\frac{1}{2}\bP_{1}\left(1,0\right)+\frac{1}{2}\bP_{1}\left(0,1\right)
\end{align}
\item[(B)]  \emph{Ensemble of maximally biased coins.} Each coin in the ensemble
described in (b) retains its initial maximal bias. That is, a coin
that lands heads on the first toss will land heads on every subsequent
toss, and similarly for tails. This behavior is described by the deterministic
dynamical system $D$ given by
\begin{align}
D\left(t,H\right) & \doteq H\\
D\left(t,T\right) & \doteq T
\end{align}
for all $t\in T$. This induces a stochastic process family $\vec{\pp}_{0}\mapsto\mm_{\vec{\pp}_{0}}^{D}$
and a corresponding probability dynamics $\bP^{D}$, as defined in
\eqref{eq:composite-measure1} and \eqref{eq:det-implement}, respectively.
The vector $\bP_{t}^{D}\left(1/2,1/2\right)$ then describes how the
ratio of maximally biased ``heads'' and ``tails'' coins evolves
(in this case, remains unchanged) within the ensemble.
\end{itemize}
Since $\bP_{t}\left(1,0\right)=\left(1,0\right)$ and $\bP_{t}\left(0,1\right)=\left(0,1\right)$
for all $t\in T$, each individual coin in the ensemble evolves exactly
as prescribed by the probability dynamics of the individual coin in
case (A). As a consequence,
\begin{align}
\bP_{t}^{D}\left(1,0\right) & =\bP_{t}\left(1,0\right)\\
\bP_{t}^{D}\left(0,1\right) & =\bP_{t}\left(0,1\right)
\end{align}
for all $t\in T$. By contrast,
\begin{equation}
\bP_{1}^{D}\left(\frac{1}{2},\frac{1}{2}\right)=\frac{1}{2}\bP_{1}^{D}\left(1,0\right)+\frac{1}{2}\bP_{1}^{D}\left(0,1\right)=\frac{1}{2}\bP_{1}\left(1,0\right)+\frac{1}{2}\bP_{1}\left(0,1\right)\neq\bP_{1}\left(\frac{1}{2},\frac{1}{2}\right)
\end{equation}
That is, the temporal evolution of the statistical mixture described
by probability vector $\vec{\pp}=\left(1/2,1/2\right)$ does not coincide
with the temporal evolution of an individual coin described by the
same probability vector. While the evolution of the statistical mixture
is convex-combination preserving, this fact does not imply that a
fundamentally different quantity, the bias of a single coin, should
evolve linearly. The point is that the evolution of the probability
vector describing the instantaneous bias of an individual coin does
\emph{not} constitute a statistical dynamics, that is, it does \emph{not}
represent the evolution of frequencies of instantaneous configurations
in an ensemble of independently evolving systems. Consequently, there
is no reason for this evolution to be linear.
\end{example}

This analysis can be easily generalized to situations in which the
statistical dynamics is stochastic rather than deterministic. Even
in this simple form, however, the distinction highlighted in the example
has an important moral. There has long been a question in the foundations
of quantum mechanics whether quantum probabilities admit an interpretation
of type~(a) or type~(b). That is, whether quantum probabilities
arise from sampling of an ensemble of systems with determinate properties---analogous
to an ensemble of ``heads'' and ``tails'' coins---or whether there are
no such determinate properties, so that the ontology of quantum systems
is instead more akin to that of an individual fair coin. The prevailing
view is that quantum mechanics does not admit a statistical interpretation
of type~(b), as indicated by no-go results such as Bell's theorem (\citet{bell1964})
and the Kochen--Specker theorem  (\citet{kochen1968}), and that quantum probabilities
should therefore be understood along the lines of type~(a).\footnote{Of course, the general characterization of interpretations of type
(a) and (b), as well as the precise implications of theorems such
as those of Bell and Kochen--Specker for these interpretations, are
subtle issues that have been widely discussed and debated in the literature
(see e.g. \citet{szabo2008}). For a specific approach that preserves the
possibility of a statistical interpretation of quantum probabilities
even in the light of these ``no-go'' results, see, for example, the prism models of Arthur
Fine and L\'aszl\'o Szab\'o (\citet{fine1982, fine1986, szabo2008}).} Accepting this common view---or, more weakly, simply allowing
for the possibility of a type~(a) interpretation---what Example~\ref{exa:statistical-linearity}
shows is that the linearity of quantum probability evolution \emph{cannot}
be inferred from the linearity of the evolution of statistical mixtures,
since the latter presupposes a type~(b) interpretation of probabilities.
This recognition is particularly important in light of some discussions
in the literature, including the widely discussed \emph{general probabilistic
theories} approaches, which aim to establish the necessity of
linear probability evolution in quantum mechanics precisely on the
basis of the linearity of what we term statistical dynamics (see
e.g. \citet{muller2021}, pp.~20--21; \citet{diosi2025}, pp.~1--2;
for a detailed criticism see \citet{szabo2023}, Appendices 1 and 2).
In the final section, we discuss what kind of probability
evolution quantum mechanics can be said to deliver.

\section{Quantum Dynamics}
\label{sec:quantum-dynamics} 

Attempts to characterize quantum mechanics in terms of the
language of classical stochastic processes date back many decades (\cite{schrodinger1931, fenyes1952, nelson1966, grabert1979, hardy1992, garbaczewski1995, garbaczewski1996, gillespie2001}). Recently, \cite{barandes2025}  has advanced a novel perspective on the correspondence between quantum and classical formalisms for probabilistic laws of evolution, claiming to derive the former from the latter. Another recent strand of literature related to such a correspondence, largely focused on open quantum systems, develops a probabilistic characterization of quantum dynamics by explicitly drawing a parallel with classical stochastic dynamics (\cite{aaronson2005, vacchini2011, vacchini2012, smirne2013, rivas2014, gullo2014, breuer2016, li2018, li2019, spekkens2019}). In light of our preceding discussion,
we close with a few remarks on this topic. To make these remarks more
concrete, we formulate them using the example of a qubit.

\subsection{Quantum Probability Evolution}

Consider a qubit with basis states $\left|1\right\rangle \doteq\left(1,0\right),\left|2\right\rangle \doteq\left(0,1\right)\in\mathbb{C}^{2}$.
Measurement in this basis yields outcome set $C=\left\{ 1,2\right\} $,
which will be regarded as the configuration space of the system. For
any quantum state $\left|\psi\right\rangle \in\mathbb{C}^{2}$, Born's
rule gives the probability vector
\begin{equation}
\vec{\pp}=\left(\left|\left\langle 1|\psi\right\rangle \right|^{2},\left|\left\langle 2|\psi\right\rangle \right|^{2}\right)\in\mathcal{S}_{2}\label{eq:prob-vector-qubit}
\end{equation}
which specifies the probabilities of finding the qubit in each configuration. 

Let $T$ be a set of time indices under consideration. As the quantum
state $\left|\psi\right\rangle $ evolves in time, equation~\eqref{eq:prob-vector-qubit}
determines a corresponding \emph{probability vector trajectory} 
\begin{equation}
t \;\; \mapsto \;\; 
\Bigl(
\, |\langle 1|\psi(t)\rangle|^2,\; |\langle 2|\psi(t)\rangle|^2 \,
\Bigr)
\label{eq:prob-vector-trajectory-qubit}
\end{equation}
where $\left|\psi\left(t\right)\right\rangle $ denotes the quantum state
at time $t\in T$. Between time
$0$ and $t$, the evolution may be governed purely by Schrödinger's equation, in
which case $\left|\psi\left(t\right)\right\rangle =\mathcal{U}\left(t\right)\left|\psi\left(0\right)\right\rangle $,
with
 $\mathcal{U}(t): \mathbb{C}^2 \to \mathbb{C}^2
$ the unitary time-evolution operator. At
this level of generality, however, we also allow for the possibility
that, in case ``measurement is performed on the system'' between time
$0$ and $t$, the quantum state undergoes collapse and $\left|\psi\left(t\right)\right\rangle \neq \mathcal{U}\left(t\right)\left|\psi\left(0\right)\right\rangle $.
In either case, the probability vector trajectory given by \eqref{eq:prob-vector-trajectory-qubit}
is a well-defined concept.

One has to be careful, however, when attempting to interpret the quantum
formalism in terms of the richer notions of (a)~a probability dynamics
and (b)~a stochastic process (family). In both cases, one encounters
difficulties: 
\begin{itemize}
\item[{(a)}] While the quantum state uniquely determines the probability vector
in \eqref{eq:prob-vector-qubit}, the converse is not true: the probability
vector does not uniquely determine the quantum state. Therefore, even in a collapse-free evolution in which an initial
quantum state uniquely determines its temporal trajectory via Schrödinger's
equation, an initial probability vector generally does not uniquely
determine a probability vector trajectory. Consequently, Schrödinger’s
equation together with Born’s rule do not, by themselves, specify
a probability dynamics as defined in Definition~\ref{def:prob-dyn}.
(In the presence of ``measurements,'' the situation is even more complicated,
since the evolution of the quantum state itself is no longer uniquely
determined by the initial state.)
\item [{(b)}] As discussed in Section~\ref{sec:twodescriptions}, the
notion of a stochastic process requires specifying a probability measure
over the set of possible trajectories through configuration space.
This raises both conceptual and technical difficulties. In a stochastic
process, a trajectory is typically understood as a sequence
of objective states the system occupies independently of measurement
(cf.~the trajectory of a particle in Brownian motion). In quantum
mechanics, however, the existence of such ``hidden variables'' is contentious
and depends on the particular version of the theory one adopts. The
related technical problem is that, to define a stochastic process
$\mm$, one has to be able to specify its finite-dimensional distributions
$\mm\left(E_{i_{1}}\left(t_{1}\right)\wedge E_{i_{2}}\left(t_{2}\right)\wedge...\wedge E_{i_{n}}\left(t_{n}\right)\right)$.
In our example, these distributions correspond to the joint probabilities
of finding the qubit in the two configurations at various times. The
functions assigning these joint probabilities must satisfy certain
probability-theoretic requirements---the so-called Kolmogorov consistency
conditions---in order for an underlying stochastic process to exist
(see proof of Proposition~\ref{prop:linearimplement} and \ref{prop:Markovian-implementation}
in Appendix~\ref{appendix:proofs}). However, constructing such finite-dimensional
distributions \textit{on the basis of quantum mechanics}---and doing so in
a way that satisfies the Kolmogorov consistency conditions---is a
non-trivial problem, well known in the literature (see e.g. \cite{grabert1979, anastopoulos2006, strasberg2019, lonigro2024}).\footnote{Note that the non-triviality of this problem is consistent with Proposition~\ref{prop:Markovian-implementation},
which states that any probability vector trajectory, including \eqref{eq:prob-vector-trajectory-qubit},
can be implemented by a stochastic process. This is because the finite-dimensional
distributions of the implementing stochastic processes whose existence
is guaranteed by Proposition~\ref{prop:Markovian-implementation}
are not constrained by, nor connected to, quantum mechanics.} For example, the multi-time distributions obtained by
repeatedly probing a quantum system through a sequence of projective
measurements (of a single observable) generally fail to satisfy the
consistency conditions (\cite{strasberg2019}, p.~2). It is worth
emphasizing that this difficulty already arises for the simplest case
of two-time joint probabilities $\mm\left(E_{i_{1}}\left(t_{1}\right)\wedge E_{i_{2}}\left(t_{2}\right)\right)$,
and thus for the associated two-time conditional probabilities $\mm\left(E_{i_{2}}\left(t_{2}\right)|E_{i_{1}}\left(t_{1}\right)\right)$.
In light of this problem, \emph{two-time conditional probabilities
appear to lack a clear meaning in standard quantum mechanics}---at
least not one that can be captured by a classical stochastic process. 
\end{itemize}
In what follows, we restrict attention to purely unitary quantum mechanical processes with no measurement collapse. For such processes, what
\emph{is} well defined on the basis of the standard formalism is the
concept of a probability vector trajectory associated with a given
initial quantum state $\left|\psi\left(0\right)\right\rangle =\left|\psi\right\rangle $.
This trajectory is given by equation~\eqref{eq:prob-vector-trajectory-qubit},
where the state evolves as $\left|\psi\left(t\right)\right\rangle =\mathcal{U}\left(t\right)\left|\psi\right\rangle $.
We will use the notation $\vec{\pp}_{\left|\psi\right\rangle }\left(t\right)$
for this concept, and refer to it as \emph{quantum probability evolution}.
We stress again that quantum probability evolution thus understood incorporates only
processes in which the quantum state evolves linearly, as dictated
by Schrödinger’s equation.

In Section~\ref{sec:Linearity}, we defined the notion of linearity
for a probability dynamics. This can be straightforwardly generalized
to the case of quantum probability evolution. In our present setting,
the definition is as follows.
\begin{defn}
\label{def:quantum-probability-evolution-linear}Quantum probability
evolution is said to be\emph{ linear (convex-combination preserving)}
iff for any $\lambda\in\left[0,1\right]$ and any quantum states $\left|\psi\right\rangle ,\left|\varphi\right\rangle ,\left|\chi\right\rangle \in\mathbb{C}^{2}$
\begin{equation}
\vec{\pp}_{\left|\psi\right\rangle }\left(0\right)=\lambda\vec{\pp}_{\left|\varphi\right\rangle }\left(0\right)+\left(1-\lambda\right)\vec{\pp}_{\left|\chi\right\rangle }\left(0\right)
\end{equation}
implies
\begin{equation}
\vec{\pp}_{\left|\psi\right\rangle }\left(t\right)=\lambda\vec{\pp}_{\left|\varphi\right\rangle }\left(t\right)+\left(1-\lambda\right)\vec{\pp}_{\left|\chi\right\rangle }\left(t\right)
\end{equation}
for all $t\in T$.
\end{defn}

\noindent Indeed, this definition generalizes Definition~\ref{def:linear-prob-dyn}
for probability dynamics in the following sense. Suppose we consider
a subset $Q\subset\mathbb{C}^{2}$ in which quantum states are in
one-to-one correspondence with probability vectors: that is, for each
$\vec{\pp}_{0}\in\mathcal{S}_{2}$ there exists a unique $\left|\vec{\pp}_{0}\right\rangle \in Q$
such that $\vec{\pp}_{0}=\left(\left|\left\langle 1|\vec{\pp}_{0}\right\rangle \right|^{2},\left|\left\langle 2|\vec{\pp}_{0}\right\rangle \right|^{2}\right)$.
Relative to this set $Q$ of states, 
\begin{equation}
\vec{\pp}_{0}\;\; \mapsto \;\;\left(\left|\left\langle 1|\mathcal{U}\left(t\right)|\vec{\pp}_{0}\right\rangle \right|^{2},\left|\left\langle 2|\mathcal{U}\left(t\right)|\vec{\pp}_{0}\right\rangle \right|^{2}\right)\label{eq:QM-prob-dyn}
\end{equation}
defines a probability dynamics, which is linear if and only if the
implication in Definition~\ref{def:quantum-probability-evolution-linear}
holds for all $\left|\psi\right\rangle ,\left|\varphi\right\rangle ,\left|\chi\right\rangle \in Q$.
In what follows, however, we will not assume the existence of such
a privileged set $Q$.

We now demonstrate that quantum probability evolution fails to be
linear. The argument follows standard reasoning from textbook quantum
theory, but it is worth spelling out explicitly in the present context.

Suppose that the initial quantum state of the qubit is the superposition
\begin{equation}
\left|\psi\right\rangle =\frac{1}{\sqrt{2}}\left|1\right\rangle +\frac{1}{\sqrt{2}}\left|2\right\rangle \label{eq:superposition}
\end{equation}
By Born's rule \eqref{eq:prob-vector-trajectory-qubit}, the initial
probability vector is
\begin{equation}
\vec{\pp}_{\left|\psi\right\rangle }\left(0\right)=\left(\frac{1}{2},\frac{1}{2}\right)=\frac{1}{2}\vec{\pp}_{\left|1\right\rangle }\left(0\right)+\frac{1}{2}\vec{\pp}_{\left|2\right\rangle }\left(0\right)\label{eq:convex-combination-initial}
\end{equation}
That is, initially, the probability vector for the superposition is
the convex combination of the probability vectors for the superposed
states. We will now follow the probability vector trajectory $\vec{\pp}_{\left|\psi\right\rangle }\left(t\right)$
for $t>0$. Since the time evolution of quantum states is linear,
we have
\begin{align}
\mathcal{U}\left(t\right)\left|\psi\right\rangle  & =\frac{1}{\sqrt{2}}\mathcal{U}\left(t\right)\left|1\right\rangle +\frac{1}{\sqrt{2}}\mathcal{U}\left(t\right)\left|2\right\rangle 
\end{align}
Then, the $i$th entry of the corresponding probability vector is
\begin{gather}
\left(\vec{\pp}_{\left|\psi\right\rangle }\left(t\right)\right)_{i}=\left|\left\langle i|\mathcal{U}\left(t\right)|\psi\right\rangle \right|^{2}=\left|\frac{1}{\sqrt{2}}\left\langle i|\mathcal{U}\left(t\right)|1\right\rangle +\frac{1}{\sqrt{2}}\left\langle i|\mathcal{U}\left(t\right)|2\right\rangle \right|^{2}\\
=\frac{1}{2}\left|\left\langle i|\mathcal{U}\left(t\right)|1\right\rangle \right|^{2}+\frac{1}{2}\left|\left\langle i|\mathcal{U}\left(t\right)|2\right\rangle \right|^{2}+\frac{1}{2}\left\langle i|\mathcal{U}\left(t\right)|1\right\rangle \overline{\left\langle i|\mathcal{U}\left(t\right)|2\right\rangle }+\frac{1}{2}\overline{\left\langle i|\mathcal{U}\left(t\right)|1\right\rangle }\left\langle i|\mathcal{U}\left(t\right)|2\right\rangle \label{eq:cross-terms-0}\\
=\frac{1}{2}\left(\vec{\pp}_{\left|1\right\rangle }\left(t\right)\right)_{i}+\frac{1}{2}\left(\vec{\pp}_{\left|2\right\rangle }\left(t\right)\right)_{i}+\textrm{cross terms}\label{eq:cross-terms}
\end{gather}
where the overline denotes complex conjugation. For $t>0$, in contrast
to $t=0$, the states $\mathcal{U}\left(t\right)\left|1\right\rangle $ and
$\mathcal{U}\left(t\right)\left|2\right\rangle $ typically have \emph{overlapping
support} in the given basis, in the sense that for some $i$, both
$\left\langle i|\mathcal{U}\left(t\right)|1\right\rangle $ and $\left\langle i|\mathcal{U}\left(t\right)|2\right\rangle $
are non-vanishing.\footnote{The expression ``overlapping support'' is borrowed from the usual notion
of wavefunctions with overlapping support: it means that in a given
measurement basis (for wavefunctions, the coordinate basis), both
states assign nonzero amplitude to the same basis element. Note that
$\mathcal{U}\left(t\right)\left|1\right\rangle $ and $\mathcal{U}\left(t\right)\left|2\right\rangle $
can have overlapping support in this sense without ceasing to be orthogonal
in $\mathbb{C}^{2}$. Their orthogonality is guaranteed by that of
$\left|1\right\rangle $ and $\left|2\right\rangle $, together with
the fact that the unitary $\mathcal{U}\left(t\right)$ preserves inner products.
Conversely, however, if two states have non-overlapping support in
a given basis, they are necessarily orthogonal.} Consequently, the cross terms in \eqref{eq:cross-terms} are generally
non-zero. What this implies is that even though
\begin{equation}
\vec{\pp}_{\left|\psi\right\rangle }\left(0\right)=\frac{1}{2}\vec{\pp}_{\left|1\right\rangle }\left(0\right)+\frac{1}{2}\vec{\pp}_{\left|2\right\rangle }\left(0\right)
\end{equation}
 in general we have
\begin{equation}
\vec{\pp}_{\left|\psi\right\rangle }\left(t\right)\neq\frac{1}{2}\vec{\pp}_{\left|1\right\rangle }\left(t\right)+\frac{1}{2}\vec{\pp}_{\left|2\right\rangle }\left(t\right)\label{eq:break-down-linearity}
\end{equation}
for $t>0$. The evolution of probabilities does \emph{not} preserve
convex combinations; that is, quantum probability evolution is \emph{not}
linear. 

The nonlinear character of quantum probability evolution, as is illustrated
in the example, crucially depends on the fact that while the evolution
of quantum states is linear, the probabilities of configurations,
as given by Born's rule, are nonlinear functions of the quantum state.
At this point, one might wonder how the derived nonlinearity is consistent
with the density matrix formalism, where \emph{both} time evolution
and Born’s rule appear as \emph{linear} maps. Indeed, in that framework
the two mappings are
\begin{align}
\varrho_{0} &  \;\;\; \mapsto \;\;\;\varrho\left(t\right)\doteq \mathcal{U}\left(t\right)\varrho_{0}\mathcal{U}\left(t\right)^{\dagger}\label{eq:time-evolution-density-matrix}\\
\varrho &  \;\;\; \mapsto \;\;\; p_{i}\doteq tr\left(\left|i\right\rangle \left\langle i\right|\varrho\right)\label{eq:Born-rule-density-matrix}
\end{align}
for $i=1,2$, with $\varrho_{0}$ and $\varrho$ denoting $2\times2$
density matrices, $\dagger$ the adjoint, and $\left|i\right\rangle \left\langle i\right|$
the orthogonal projector onto $\left|i\right\rangle $. Both maps
\eqref{eq:time-evolution-density-matrix}-\eqref{eq:Born-rule-density-matrix}
are linear, hence so is their composition. Consequently, the probabilities
of configurations at any given time are linear functions of the initial
density matrix. Crucially, however, this does not imply that probabilities
at later times are linear functions of the initial \emph{probabilities}
(cf.~Fig.~\ref{fig:diagram}).
\begin{figure}
\begin{centering}
\begin{tikzcd}[column sep=3cm, row sep=2cm]
\varrho_{0} 
  \arrow[r, "\mathcal{U}(t)"{above}, "\text{linear}"{below}] 
  \arrow[d, "\mathit{diag}"'{right}, "\text{linear}"{rotate=270, anchor=center, pos=0.5, xshift=3pt, yshift=-5pt}] 
& 
\varrho(t) 
  \arrow[d, "\mathit{diag}"{right}, "\text{linear}"{rotate=270, anchor=center, pos=0.5, xshift=3pt, yshift=-5pt}] \\
\vec{\pp}_{0} 
  \arrow[r, dashed, "\text{generally nonlinear}"{below}] 
& 
\vec{\pp}(t)
\end{tikzcd}
\par\end{centering}
\caption{\emph{Even though both time evolution and taking the diagonal are
linear operations on density matrices, the time evolution of the resulting
probability vectors is generally nonlinear}}\label{fig:diagram}
\end{figure}
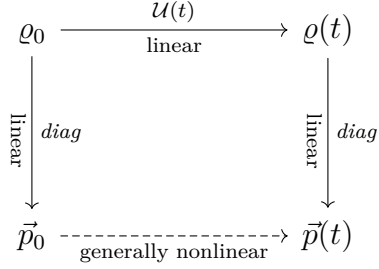
 For that to hold, it would also have to be the case that whenever
the initial probability vectors combine convexly, the corresponding
density matrices from which they are derived under Born’s rule \eqref{eq:Born-rule-density-matrix}
do so as well. But this is not the case. For example, although the superposition
state $\left|\psi\right\rangle $ defined in \eqref{eq:superposition}
is a linear combination of $\left|1\right\rangle $ and $\left|2\right\rangle $,
its density matrix $\left|\psi\right\rangle \left\langle \psi\right|$
is \emph{not} a convex combination of $\left|1\right\rangle \left\langle 1\right|$
and $\left|2\right\rangle \left\langle 2\right|$:
\begin{equation}
\left|1\right\rangle \left\langle 1\right|=\left(\begin{array}{cc}
1 & 0\\
0 & 0
\end{array}\right),\,\,\,\,\,\,\,\,\,\,\left|2\right\rangle \left\langle 2\right|=\left(\begin{array}{cc}
0 & 0\\
0 & 1
\end{array}\right)\label{eq:density-matrices}
\end{equation}
but
\begin{equation}
\left|\psi\right\rangle \left\langle \psi\right|=\left(\begin{array}{cc}
\frac{1}{2} & \frac{1}{2}\\
\frac{1}{2} & \frac{1}{2}
\end{array}\right)\label{eq:density-matrix-superposition}
\end{equation}
Yet the corresponding probability vectors given by Born’s rule \emph{do}
combine convexly, as shown in \eqref{eq:convex-combination-initial}.

Note that in the special case where the initial density matrix is
the ``mixed state'' 
\begin{equation}
\frac{1}{2}\left|1\right\rangle \left\langle 1\right|+\frac{1}{2}\left|2\right\rangle \left\langle 2\right|=\left(\begin{array}{cc}
\frac{1}{2} & 0\\
0 & \frac{1}{2}
\end{array}\right)\label{eq:mixed-state}
\end{equation}
the additional requirement is obviously satisfied: the density matrix
is itself the appropriate convex combination. This ``mixed state,''
however, has nothing to do with the superposition state $\left|\psi\right\rangle $
whose temporal evolution gives rise to the breakdown of convex-combination
preservation in the evolution of Born probabilities.

In light of these remarks, it is worth having a look at the approach to quantum mechanics advanced by \cite{barandes2025}. Here is how the paper lays out what is termed the stochastic--quantum correspondence. 
\begin{quote}
The linear marginalization relationship 
\begin{equation}
\pp_{i}\left(t\right)=\sum_{j=1}^{N}\Gamma_{ij}\left(t\leftarrow0\right)\pp_{j}\left(0\right)\label{eq:fake-lin2}
\end{equation}
between the system's standalone probabilities $\pp_{j}\left(0\right)$
at the initial time $0$ and the standalone probabilities $\pp_{i}\left(t\right)$
at the target time $t$ can now be recast as
\begin{equation}
p_{i}\left(t\right)=tr\left(\left|i\right\rangle \left\langle i\right|\varrho\left(t\right)\right)\label{eq:Born-rule-Barandes}
\end{equation}
Here
\begin{equation}
\varrho\left(t\right)\equiv\Theta\left(t\leftarrow0\right)\left[\sum_{j=1}^{N}\pp_{j}\left(0\right)\left|j\right\rangle \left\langle j\right|\right]\Theta\left(t\leftarrow0\right)^{\dagger}=\Theta\left(t\right)\textrm{diag}\left(...,\pp_{j}\left(0\right),...\right)\Theta\left(t\right)^{\dagger}\label{eq:ro-Barandes}
\end{equation}
{[}...{]} Crucially, notice how the linearity of the marginalization
relationship \eqref{eq:fake-lin2} is ultimately responsible for the
linearity of the relationship between the matrix $\varrho\left(t\right)$
and its value $\varrho\left(0\right)$ at the initial time $0$. (\cite{barandes2025},
p.~7, with minor adjustments to notation)
\end{quote}
Following the approach discussed in Section~\ref{sec:Linearity}, \cite{barandes2025} starts out by supposing a linear probability dynamics $\bP$. Then, after introducing a matrix family $\left\{ \Theta\left(t\right)\right\} _{t\in T}$
for which $\mP_{ij}\left(t\right)=\left|\Theta_{ij}\left(t\right)\right|^{2}$,
the equation $\vec{\pp}\left(t\right)=\mP\left(t\right)\vec{\pp}\left(0\right)$
is rewritten in the form \eqref{eq:Born-rule-Barandes}. The resulting
expressions resemble those of quantum mechanics \eqref{eq:time-evolution-density-matrix}-\eqref{eq:Born-rule-density-matrix}.
However, what is recovered can only have a limited connection to quantum
mechanics. As we have seen above, 1) quantum mechanics does \emph{not}
give rise to a well-defined probability dynamics, since the initial
probability vector $\vec{\pp}\left(0\right)$ does not uniquely determine
the later one $\vec{\pp}\left(t\right)$, and 2) quantum probability
evolution, described by $\vec{\pp}_{\left|\psi\right\rangle }\left(t\right)$,
is in general \emph{not} linear---even in case of a purely unitary process.
Therefore, assuming that the probabilities the theory starts out from are the same probabilities that standard quantum theory talks about, the premises on which this reconstruction rests appear to be in tension with the probabilistic structure of quantum mechanics, the very theory that the reconstruction is meant to recover.

It is instructive to highlight the discrepancy between standard quantum mechanics and the formalism developed in \cite{barandes2025}. As noted below
Definition~\ref{def:quantum-probability-evolution-linear}, one
only obtains a well-defined probability dynamics from quantum mechanics
if the set of possible initial quantum states are restricted to a
subset in one-to-one correspondence with probability vectors---in the
sense of equation \eqref{eq:QM-prob-dyn}. And indeed, in the reconstruction of \cite{barandes2025}, the possible initial states (here, density matrices)
are restricted precisely in this way. By \eqref{eq:ro-Barandes},
the initial density matrix
\begin{equation}
\varrho\left(0\right)\equiv\sum_{j=1}^{N}\pp_{j}\left(0\right)\left|j\right\rangle \left\langle j\right|=\textrm{diag}\left(...,\pp_{j}\left(0\right),...\right)
\end{equation}
is by definition diagonal, with diagonal entries given by the components
of the initial probability vector. The diagonal density matrices are
in one-to-one correspondence with probability vectors, and so such a restriction
indeed ensures the existence of a well-defined probability dynamics---as
is assumed in \cite{barandes2025} from the outset.

Crucially, the diagonal density matrices are exactly those that can
be written as convex combinations of the projectors $\left|j\right\rangle \left\langle j\right|$,
just like the ``mixed state'' \eqref{eq:mixed-state}.\footnote{Moreover, some of the arguments in \cite{barandes2025} tacitly assume not merely that
the initial density matrix is a convex combination of projectors $\left|j\right\rangle \left\langle j\right|$,
but that it is in fact one of the projectors $\left|j\right\rangle \left\langle j\right|$
itself---that is, the system must begin in one of the basis states
$\left|j\right\rangle $. This follows from the way ``pure states''
$\left|\psi\right\rangle $ are defined in his formalism: by this
definition, the version of Born’s rule expressed in terms of ``pure
states'' (\eqref{eq:prob-vector-trajectory-qubit}) is only valid in the framework of \cite{barandes2025} when the initial pure state coincides with a
basis state $\left|j\right\rangle $ (cf.~\cite{barandes2025}, p.~7).
Consequently, all arguments in \cite{barandes2025} that rely on this
formulation of Born’s rule (such as the discussion of ``division
events'' and ``decoherence'' in sections~3.7-3.8) hold only under
this restriction. Note that this restriction means that the initial
probability vectors are confined to the basis vectors $\vec{e}_{j}$,
i.e. the vertices of $\mathcal{S}_{N}$, in which case the very distinction
between linear and nonlinear probability evolution loses its meaning.
 } As we have seen, for such initial density matrices the probability
vector given by Born's rule \eqref{eq:Born-rule-Barandes} \emph{does}
evolve linearly. But this is a very special case, not representative
of quantum mechanics in general. Indeed, the density matrix \eqref{eq:density-matrix-superposition}
associated with superposition \eqref{eq:superposition}---the very
case that gives rise to nonlinear probability evolution---cannot
be represented in this way. What this formalism thus leaves out is one of the most fundamental features of quantum mechanics:
the possibility of superposition as an initial state.

\subsection{Interference and Divisibility}

In addition to the nonlinearity introduced by $\left|\cdot\right|^{2}$
in Born's rule, there is another crucial factor behind the nonlinear
character of quantum probability evolution, as illustrated in equation
\eqref{eq:cross-terms-0}: the fact that quantum states with initially
non-overlapping support generally evolve into states with overlapping
support. Recall that this means that for $t>0$, both $\left\langle i|\mathcal{U}\left(t\right)|1\right\rangle $
and $\left\langle i|\mathcal{U}\left(t\right)|2\right\rangle $ are typically
non-vanishing, even though this is not the case for $t=0$. This phenomenon
is a hallmark of the appearance of quantum interference; and in this
sense, interference can be seen as a key factor behind the nonlinear
character of quantum probability evolution.

This reading stands at odds with the interpretation of quantum
interference advanced in \cite{barandes2025}. There, it is argued
that ``interference is a direct consequence of the stochastic dynamics
not generally being divisible. More precisely, interference is nothing
more than a generic discrepancy between the \emph{actual} indivisible
stochastic dynamics and a \emph{heuristic-approximate} divisible stochastic
dynamics'' (\cite{barandes2025}, p.~12, emphasis in original).
In Appendix~\ref{appendix:Barandes-interference}, we explain in detail why this analysis is, in our view, mathematically problematic. There is, however, a more fundamental difficulty with this diagnosis.

Let us set aside the fact that, since in quantum mechanics $\vec{\pp}\left(t\right)$
is not determined by $\vec{\pp}\left(0\right)$, we have no reason
to expect that $\vec{\pp}\left(t\right)$ is determined by $\vec{\pp}\left(t'\right)$,
for any $t'<t$; and so decomposability and divisibility in the sense
of Definitions~\ref{def:decomposability} and \ref{def:divisib}
are not applicable to quantum mechanics. Instead, in analogy with
$\vec{\pp}_{\left|\psi\right\rangle }\left(t\right)$, it is natural
to introduce the weaker notion $\vec{\pp}_{\left|\psi\right\rangle ,t'}\left(t\right)$,
defined as the probability vector trajectory initiated from a given
quantum state $\left|\psi\right\rangle $ at time $t'\in T$ (rather
than at time 0). Now for the unitary time-evolution of the quantum state it holds that for all $t,t'\in T$, $t'<t$ there exists a unitary operator $\mathcal{U}\left(t\leftarrow t'\right)$ such that
\begin{equation}
\mathcal{U}\left(t\right)=\mathcal{U}\left(t\leftarrow t'\right)\,\mathcal{U}\left(t'\right)\label{eq:decomp-unitary}
\end{equation}
As a consequence,
\begin{equation}\vec{\pp}_{\left|\psi\right\rangle ,t'}\left(t\right)=\left(\left|\left\langle 1|\mathcal{U}\left(t\leftarrow t'\right)|\psi\right\rangle \right|^{2},\left|\left\langle 2|\mathcal{U}\left(t\leftarrow t'\right)|\psi\right\rangle \right|^{2}\right)
\end{equation}
is always well-defined, and quantum probability evolution can always
be decomposed into corresponding intermediate steps, in the sense
that 
\begin{equation}
\begin{aligned}
\vec{\pp}_{|\psi\rangle}(t)
&=
\left(
  \left|\langle 1|\mathcal{U}(t)|\psi\rangle\right|^{2},
  \left|\langle 2|\mathcal{U}(t)|\psi\rangle\right|^{2}
\right) \\
&=
\left(
  \left|\langle 1|\mathcal{U}(t \leftarrow t')\,\mathcal{U}(t')|\psi\rangle\right|^{2},
  \left|\langle 2|\mathcal{U}(t \leftarrow t')\,\mathcal{U}(t')|\psi\rangle\right|^{2}
\right)
=
\vec{\pp}_{\mathcal{U}(t')|\psi\rangle ,t'}(t)
\end{aligned}
\label{eq:decomp-QM}
\end{equation}
for all $t,t'\in T$, $t'<t$. Equation~\eqref{eq:decomp-QM}
says that the evolution of the probability vector $\vec{\pp}_{\left|\psi\right\rangle }\left(0\right)$
to $\vec{\pp}_{\left|\psi\right\rangle }\left(t\right)$ can be decomposed
into two steps: first, from $0$ to $t'$, the quantum state $\left|\psi\right\rangle $
evolves to $\mathcal{U}\left(t'\right)\left|\psi\right\rangle $; then, from
$t'$ to $t$, the probability vector $\vec{\pp}_{\mathcal{U}\left(t'\right)\left|\psi\right\rangle ,t'}\left(t'\right)$
evolves to $\vec{\pp}_{\mathcal{U}\left(t'\right)\left|\psi\right\rangle ,t'}\left(t\right)$.
In simpler terms, the quantum state at time $t'$, and not only at time
$0$, determines future probability vectors. In this sense, quantum
probability evolution is \emph{decomposable}.

Recall from Section~\ref{sec:Decomposability} that, in case of a well-defined
probability dynamics, what divisibility adds to the notion of decomposability
is the requirement that the decomposing maps $\bP_{t\leftarrow t'}$
be represented by stochastic matrices. Crucially, as we have emphasized,
divisibility is only seen as synonymous with decomposability when
linearity is, incorrectly, taken to be a necessary feature of
probability dynamics, so that every probability evolution is automatically
associated with stochastic matrices. For \emph{nonlinear} probability
dynamics, however, divisibility---again, understood as the requirement
that the decomposing maps be represented by stochastic matrices, and
thus be \emph{linear}---is a concept that simply does not apply.
Translating this observation into the language of quantum probabilities:
since quantum probability evolution is generally nonlinear, any notion
of ``divisibility'' that requires $\vec{\pp}_{\left|\psi\right\rangle ,t'}\left(t\right)$
to preserve convex combinations, in a sense analogous to Definition~\ref{def:quantum-probability-evolution-linear},
is inapplicable. The only notion of ``divisibility,'' or more appropriately
\emph{decomposability}, that makes sense for, and indeed is satisfied
by, the evolution of quantum probabilities is the one expressed by
equation \eqref{eq:decomp-QM}---a simple consequence of \eqref{eq:decomp-unitary}. Rather than an ``indivisible stochastic
dynamics,'' quantum probability evolution is therefore more aptly
characterized as a \emph{nonlinear yet decomposable} dynamics. Notice
that quantum interference, as one of the key sources of the nonlinearity of quantum probability evolution, is itself a central reason why the framework of indivisible stochastic processes in \cite{barandes2025} cannot model quantum mechanics and account for the interference phenomena it entails.

\subsection{Tomographic Completion}

In line with our general framework, the preceding discussion focused on a single
probability vector encoding the statistics of a single observable, corresponding to a particular choice of basis $\left|1\right\rangle,\left|2\right\rangle \in\mathbb{C}^{2}$. From the perspective of the Hilbert-space formalism, such a probability vector contains only partial information about the underlying quantum state. From this point of view, a natural extension of the framework is to consider a collection of probability
vectors encoding the statistics of a \emph{tomographically complete
set of observables}, which together uniquely determine the quantum state. This perspective will provide a further way to understand the nonlinearity of quantum probability evolution.

As an example, consider the three Pauli measurements corresponding
to the following three distinct bases in $\mathbb{C}^{2}$:

\begin{equation}
\begin{array}{rcl@{\hspace{2em}}rcl}
|z+\rangle & = & |1\rangle & |z-\rangle & = & |2\rangle\\
|x+\rangle & = & \tfrac{1}{\sqrt{2}}\big(|1\rangle+|2\rangle\big) & |x-\rangle & = & \tfrac{1}{\sqrt{2}}\big(|1\rangle-|2\rangle\big)\\
|y+\rangle & = & \tfrac{1}{\sqrt{2}}\big(|1\rangle+i|2\rangle\big) & |y-\rangle & = & \tfrac{1}{\sqrt{2}}\big(|1\rangle-i|2\rangle\big)
\end{array}
\end{equation}
One can then define the six-component vector 
\begin{equation}
\underline{v}\doteq\left(\begin{array}{c}
tr\left(\left|z+\right\rangle \left\langle z+\right|\varrho\right)\\
tr\left(\left|z-\right\rangle \left\langle z-\right|\varrho\right)\\
tr\left(\left|x+\right\rangle \left\langle x+\right|\varrho\right)\\
tr\left(\left|x-\right\rangle \left\langle x-\right|\varrho\right)\\
tr\left(\left|y+\right\rangle \left\langle y+\right|\varrho\right)\\
tr\left(\left|y-\right\rangle \left\langle y-\right|\varrho\right)
\end{array}\right)\label{eq:6d-vector}
\end{equation}
whose value uniquely determines the quantum state $\varrho$. Note
that $\underline{v}$ is no longer a probability vector as its entries sum
to 3 rather than to 1. On the other hand, because $\underline{v}$ uniquely
determines $\varrho$, its initial value uniquely determines its later
values via the unitary evolution of the quantum state---unlike the
probability vector $\vec{\pp}$ in \eqref{eq:prob-vector-qubit}, which consists
only of the first two entries of \eqref{eq:6d-vector} (with $\varrho=\left|\psi\right\rangle \left\langle \psi\right|$). That is, there
exists a well-defined mapping 
\begin{equation}
\mathbf{V}_{t}:\underline{v}\left(0\right)\mapsto\underline{v}\left(t\right)
\end{equation}
where $\underline{v}\left(t\right)$ is the value of vector $\underline{v}$ at
time $t$, obtained by substituting $\varrho\left(t\right)$ for $\varrho$
in \eqref{eq:6d-vector}.

One may now ask: although the upper block corresponding to the first two
entries does not give rise to a linear evolution in the sense of Definition~\ref{def:quantum-probability-evolution-linear},
is $\mathbf{V}_{t}$, the time evolution of the vector $\underline{v}$
as a whole, convex-combination preserving? The answer is affirmative.
The simple argument that follows is formulated in the concrete case
of the Pauli measurements, but it applies more generally to any tomographically
complete set of observables.

Let $\mathcal{D}$ denote the convex set of density matrices and consider
the map 
\begin{equation}
\phi:\mathcal{D}\rightarrow\mathbb{R}^{6},\;\;\varrho\mapsto\underline{v}\label{eq:map}
\end{equation}
$\phi$ preserves convex combinations, and since the three
Pauli measurements constitute a tomographically complete set of observables,
it is injective. Because the inverse of any convex-combination preserving
injection also preserves convex combinations, the inverse map $\phi^{-1}:\mathrm{Ran}\phi\rightarrow\mathcal{D}$
is convex-combination preserving as well.\footnote{Let $\underline{v}_{1}=\phi\left(\varrho_{1}\right)$ and $\underline{v}_{2}=\phi\left(\varrho_{2}\right)$.
For any $\lambda\in\left[0,1\right]$, $\lambda\underline{v}_{1}+\left(1-\lambda\right)\underline{v}_{2}\in\mathrm{Ran}\phi$,
and we have
\begin{align}
\phi^{-1}\left(\lambda\underline{v}_{1}+\left(1-\lambda\right)\underline{v}_{2}\right) & =\phi^{-1}\left(\lambda\phi\left(\varrho_{1}\right)+\left(1-\lambda\right)\phi\left(\varrho_{2}\right)\right)=\phi^{-1}\left(\phi\left(\lambda\varrho_{1}+\left(1-\lambda\right)\varrho_{2}\right)\right)\\
 & =\lambda\varrho_{1}+\left(1-\lambda\right)\varrho_{2}=\lambda\phi^{-1}\left(\underline{v}_{1}\right)+\left(1-\lambda\right)\phi^{-1}\left(\underline{v}_{2}\right)
\end{align}
} Consequently, the time evolution $\underline{v}\left(0\right)\mapsto\underline{v}\left(t\right)$,
given by 
\begin{equation}
\mathbf{V}_{t}=\phi\circ\mathcal{U}_{\varrho}\left(t\right)\circ\phi^{-1}
\end{equation}
(where $\mathcal{U}_{\varrho}\left(t\right)$ denotes the unitary
evolution acting on density matrices as defined in \eqref{eq:time-evolution-density-matrix}),
is a composition of convex-combination preserving maps and therefore
itself preserves convex combinations. 

It should be emphasized that this fact is in no contradiction with the nonlinearity
of the evolution \emph{of the upper two-entry block} of $\underline{v}$,
given by the probability vector  $\vec{\pp}$ in \eqref{eq:prob-vector-qubit}.
We illustrate this in a simple explicit
example provided by the superposition state $\left|\psi\right\rangle $
defined in \eqref{eq:superposition}. One readily verifies that
\begin{eqnarray}
\underline{v}_{\left|1\right\rangle } & \doteq & \phi\left(\left|1\right\rangle \left\langle 1\right|\right)=\left(1,0,\frac{1}{2},\frac{1}{2},\frac{1}{2},\frac{1}{2}\right)\\
\underline{v}_{\left|2\right\rangle } & \doteq & \phi\left(\left|2\right\rangle \left\langle 2\right|\right)=\left(0,1,\frac{1}{2},\frac{1}{2},\frac{1}{2},\frac{1}{2}\right)\\
\underline{v}_{\left|\psi\right\rangle } & \doteq & \phi\left(\left|\psi\right\rangle \left\langle \psi\right|\right)=\left(\frac{1}{2},\frac{1}{2},1,0,\frac{1}{2},\frac{1}{2}\right)
\end{eqnarray}
While in the upper block $\underline{v}_{\left|\psi\right\rangle }$
is an equal convex combination of $\underline{v}_{\left|1\right\rangle }$
and $\underline{v}_{\left|2\right\rangle }$, this is not the case for the
vectors as a whole:
\begin{equation}
\underline{v}_{\left|\psi\right\rangle }\neq\frac{1}{2}\underline{v}_{\left|1\right\rangle }+\frac{1}{2}\underline{v}_{\left|2\right\rangle }
\end{equation}
Therefore, in complete harmony with the fact that the evolution $\mathbf{V}_{t}$
is convex-combination preserving, we generally have
\begin{equation}
\mathbf{V}_{t}\left(\underline{v}_{\left|\psi\right\rangle }\right)\neq\frac{1}{2}\mathbf{V}_{t}\left(\underline{v}_{\left|1\right\rangle }\right)+\frac{1}{2}\mathbf{V}_{t}\left(\underline{v}_{\left|2\right\rangle }\right)
\end{equation}
In general, this inequality already holds when restricted to the upper block,
despite the fact that the upper block of $\underline{v}_{\left|\psi\right\rangle }$
is identical to that of $\frac{1}{2}\underline{v}_{\left|1\right\rangle }+\frac{1}{2}\underline{v}_{\left|2\right\rangle }$.
This explains why linearity can break down in the upper
block, as seen in equation~\eqref{eq:break-down-linearity}, even
though $\mathbf{V}_{t}$ is convex-combination preserving when considered
as a whole.

It must be stressed once again that the vector $\underline{v}$, by contrast with
$\vec{\pp}$, is \emph{not} a probability vector. Therefore, even though the time-evolution
of $\underline{v}$ is given by a convex-combination preserving map, this
does not license the conclusion that quantum mechanics can be described
in terms of what we have defined as a \emph{linear probability dynamics}.

\subsection{Concluding Remark}

The discussion of this section illustrates the subtlety of describing quantum dynamics within a classical probabilistic framework. We have seen that neither the notion of a stochastic process (or stochastic process family) nor that of a probability dynamics is adequate for this purpose. Instead, two alternative conceptual tools are available for characterizing the dynamics: $\vec{\pp}_{\left|\psi\right\rangle }\left(t\right)$ or
$\mathbf{V}_{t}$. Under the first representation, the dynamics fails to be linear and,
in general, cannot even be expressed as a map. Under the second,
the dynamics is represented by a convex-combination preserving
map; however, the object whose evolution is thereby described---the
vector $\underline{v}$ defined in equation~\eqref{eq:fake-lin2}---cannot
be associated with a probability distribution over a single configuration
space.

It must be emphasized that none of these facts about the probabilistic
structure of quantum theory can be anticipated on the basis of \emph{a priori} reasoning
about probability alone. This point is particularly important in light of
recent approaches in quantum foundations that attempt to derive the
formalism of quantum mechanics from seemingly more natural probabilistic
assumptions. In addition to \cite{barandes2025}, a widely discussed example is the framework of \emph{general probabilistic theories} (e.g. \cite{hardy2001, chiribella2011, masanes2011, muller2021}). A common assumption in many of these reconstruction programs
is that the temporal evolution of probabilities preserves convex combinations. Two main types of arguments are typically offered in support of
this assumption. The first, discussed in Section~\ref{sec:Linearity},
is formulated in terms of a single probability vector and appeals to the \emph{law of total
probability}. The second, examined in Section~\ref{sec:mixing}, can be formulated either in terms of a single probability vector or in terms of a collection of probability vectors thought to characterize the system’s probabilistic “state,” and
relies on the properties of \emph{statistical mixing}.\footnote{In the general probabilistic theories approach, an object analogous to our vector $\underline{v}$,
encoding the statistics of a set of “fiducial” measurements, appears
to be treated as the primary object of the probabilistic description. The argument from statistical mixing is then formulated in terms of this object. See e.g. \cite{hardy2001},
pp.~10--12.} As we have seen, the
first argument is flawed because it conflates the concepts of probability
dynamics and stochastic process. The second has force only under a
specific interpretation of the probabilities in question; more precisely,
it presupposes a particular underlying ontology for the system. That
ontology, however, is at odds with common interpretations of quantum
mechanics. In view of all this, even when quantum probability evolution can be represented in a convex-combination preserving form, it remains questionable whether these reconstruction programs provide an explanation for \emph{why that has to be the case}.
\vspace{0.2in}

\noindent
{\bf Acknowledgements.} This work has been supported by the Hungarian Scientific Research Fund, Advanced 152165.

\appendix

\section{Appendix \label{appendix:proofs}}

\transitionconstimplieslinearity*

\begin{proof}\label{proof:transition-const-implies-linearity}
From \eqref{eq:fake-lin-0}---valid only for $\vec{\pp}_{0}\in\mathcal{S}_{N}$
with $\pp_{01},\pp_{02},...,\pp_{0N}>0$---and the defining equation \eqref{eq:common-M}
of $\mM\left(t\right)$, we have
\begin{equation}
\bP\left(t,\vec{\pp}_{0}\right)=\mM\left(t\right)\vec{\pp}_{0}
\end{equation}
for all $t\in T$ and for all $\vec{\pp}_{0}\in\mathcal{S}_{N}$ {with $\pp_{01},\pp_{02},...,\pp_{0N}>0$.}

{Now fix an arbitrary $\vec{\pp}_{0}\in\mathcal{S}_{N}$ (allowing zero entries), fix $t\in T$,
and let $J\doteq\{j\in C\mid \pp_{0j}>0\}$. For each $i\in C$, using that
$\{E_j(0)\}_{j=1}^N$ is a partition of $\Omega$ and that $\mm_{\vec{\pp}_{0}}(E_j(0))=\pp_{0j}$, we have
\begin{align}
\mm_{\vec{\pp}_{0}}\!\left(E_i(t)\right)
&=\sum_{j=1}^N \mm_{\vec{\pp}_{0}}\!\left(E_i(t)\wedge E_j(0)\right)
=\sum_{j\in J} \mm_{\vec{\pp}_{0}}\!\left(E_i(t)\wedge E_j(0)\right) \nonumber\\
&=\sum_{j\in J} \mm_{\vec{\pp}_{0}}\!\left(E_i(t)\mid E_j(0)\right)\mm_{\vec{\pp}_{0}}\!\left(E_j(0)\right)
=\sum_{j\in J} \mm_{\vec{\pp}_{0}}\!\left(E_i(t)\mid E_j(0)\right)\pp_{0j} \label{eq:tc-proof-sumJ}
\end{align}
By transition-constancy \eqref{eq:transition-constant} and by \eqref{eq:common-M}, for each $j\in J$ we have $\mm_{\vec{\pp}_{0}}\!\left(E_i(t)\mid E_j(0)\right)=\mM_{ij}(t)$, hence
\begin{equation}
\mm_{\vec{\pp}_{0}}\!\left(E_i(t)\right)=\sum_{j\in J}\mM_{ij}(t)\pp_{0j}
=\sum_{j=1}^N \mM_{ij}(t)\pp_{0j}=\big(\mM(t)\vec{\pp}_{0}\big)_i
\end{equation}
where the second equality uses that $\pp_{0j}=0$ for $j\notin J$.
Therefore $\bP(t,\vec{\pp}_{0})=\mM(t)\vec{\pp}_{0}$ holds for all $\vec{\pp}_{0}\in\mathcal{S}_N$ and all $t\in T$.}

Since $\bP_{t}$ acts by matrix multiplication with $\mM\left(t\right)$---which
is now independent of $\vec{\pp}_{0}$---$\bP_{t}^ {}$ is convex
combination preserving and $\mP\left(t\right)=\mM\left(t\right)$,
for all $t\in T$.
\end{proof}

\linearimplement*

\begin{proof}\label{proof:linearimplement}
First, let $\vec{\pp}_{0}\in\mathcal{S}_{N}$ be fixed. We now define a suitable distribution function $F$ to which we can apply Kolmogorov's extension theorem (e.g. \cite{skorokhod1996},
p.~8). For any $t\in T$ and any $C_t\subseteq C$, let
\begin{equation}
	F_{t}(C_t) \doteq \sum_{j=1}^N  \pp_{0j} \sum_{i\in C_t} \mP_{ij}(t)	
\end{equation}
and for any $t\in T$ and $C_{0} ,C_t \subseteq C$, let
\begin{equation}
	F_{0,t}(C_{0},C_t)  \doteq \sum_{j\in C_{0}} \pp_{0j} \sum_{i\in C_t} \mP_{ij}(t)
\end{equation}
Note that since $\mP(0)=I$, $F_{0}(C_{0})=\sum_{j\in C_{0}}\pp_{0j}$.

Higher-order joint distributions are defined as follows. Assume that
$t_{1},\dots,t_{m}\in T$ are distinct and $0\not\in\{t_{1},\dots,t_{m}\}$,
with $m\geq2$. For subsets $C_{0},C_{1},\dots,C_{m}\subseteq C$,
let
\begin{equation}
F_{t_{1},\dots,t_{m}}(C_{1},\dots,C_{m})\doteq\sum_{j=1}^{N}\pp_{0j}\prod_{k=1}^{m}\left(\sum_{i\in C_{k}}\mP_{ij}(t_{k})\right)\label{eq:Fnincsnulla}
\end{equation}
and
\begin{equation}
F_{0,t_{1},\dots,t_{m}}(C_{0},C_{1},\dots,C_{m})\doteq\sum_{j\in C_{0}}\pp_{0j}\prod_{k=1}^{m}\left(\sum_{i\in C_{k}}\mP_{ij}(t_{k})\right)\label{eq:Fvannulla}
\end{equation}

We now verify the conditions of Kolmogorov's extension theorem for the function $F$ defined above:
\begin{itemize}
	\item Non-negativity and $\sigma$-additivity:
	\begin{itemize}
		\item Since $\mP(t)$ is a stochastic matrix and $\vec{\pp}_{0} \in \mathcal{S}_{N}$, hence $\mP(t) \vec{\pp}_{0} \in \mathcal{S}_{N}$. Hence, $F_{t}: C_t \mapsto \sum_{j=1}^N \sum_{i\in C_t} \mP_{ij}(t) \pp_{0j}$ is a probability measure on $C$. 
		\item $F_{0,t}$ is additive by definition. Moreover, since $\mP(t)$ is a stochastic matrix and $\vec{\pp}_{0} \in \mathcal{S}_{N}$, it follows that $F_{0,t}(C,C)=1$, so $F_{0,t}$ is a probability measure on $C^{2}$.
		\item For general $m$, \eqref{eq:Fnincsnulla} is a convex combination
(with weights $\pp_{0j}\ge0$, $\sum_{j=1}^{N}\pp_{0j}=1$) of product
measures on $C^{m}$, hence define a probability measures on $C^{m}$.
\eqref{eq:Fvannulla} is additive by definition. Moreover, since
$\mP(t)$ is a stochastic matrix and $\vec{\pp}_{0}\in\mathcal{S}_{N}$,
it follows that $F_{0,t_{1},\dots,t_{m}}(C,C,\dots,C)=1$, so $F_{0,t_{1},\dots,t_{m}}$
is a probability measure on $C^{m+1}$.
	\end{itemize}
	\item Symmetry under permutations: follows from \eqref{eq:Fnincsnulla}-\eqref{eq:Fvannulla}, since the product over $k$ is commutative; for tuples containing time $0$, the factor depending on $C_{0}$ only appears as restricting the $j$-sum, and the remaining factors are symmetric in the positive times.
	\item Marginalization: for $t, t_{1},t_{2}>0$
		\begin{eqnarray}
			F_{0,t}(C_{0},C)		& = &	\sum_{j\in C_{0}} \pp_{0j} \left( \sum_{i\in C} \mP_{ij}(t) \right) = \sum_{j\in C_{0}} \pp_{0j} = F_{0}(C_{0})	\\
			F_{0,t}(C,C_{t})			& = &	\sum_{j\in C} \pp_{0j} \left( \sum_{i\in C_{t}} \mP_{ij}(t) \right) = F_{t}(C_{t})	\\
			F_{t_{1},t_{2}}(C_{1},C) 	& = & \sum_{j=1}^N \pp_{0j} \left(\sum_{i\in C_{1}} \mP_{ij}(t_{1}) \right) \cdot 1 = F_{t_{1}}(C_{1})
		\end{eqnarray}
		and similarly for all higher-order marginals. (Let $S \subseteq T$ be a finite set of distinct times and $S'=S\setminus\{s_{\ell}\}$ for some $s_{\ell}\in S$. If $s_{\ell}\neq 0$, then summing over $C_{\ell}=C$ replaces the corresponding factor by
$\sum_{i\in C}\mP_{ij}(s_{\ell})=1$ (since the matrices are stochastic), yielding exactly the formula for $F_{S'}$ from $F_{S}$. If $s_{\ell}=0$, then summing over $C_{0}=C$ replaces $\sum_{j\in C_{0}}\pp_{0j}$ by $1$, again recovering the formula for $F_{S'}$.)
\end{itemize}

Since these three conditions hold, by Kolmogorov's extension theorem there exists a unique probability measure $\mm_{\vec{\pp}_{0}}$ on $(\Omega,{\cal F})$ whose finite-dimensional distributions are given by the $F$'s above:
\begin{equation}
	\mm_{\vec{\pp}_{0}}(E_{i_{1}}(t_{1}) \wedge \cdots \wedge E_{i_{m}}(t_{m})) = F_{t_{1},\dots,t_{m}}(\{ i_{1}\} ,\dots, \{ i_{m}\} )
\end{equation}
In particular, for every $t\in T$ and $i\in C$,
\begin{equation}
	\mm_{\vec{\pp}_{0}}( E_{i}(t)) = \sum_{j = 1}^N \pp_{0j} \mP_{ij}(t) = ( \mP(t) \vec{\pp}_{0} )_{i} = ( \bP_{t}( \vec{\pp}_{0} ))_{i}
\end{equation}
so $\vec{\mm}_{\vec{\pp}_{0}}(t)=\bP(t, \vec{\pp}_{0})$ for all $t\in T$, that is, the stochastic process $\mm_{\vec{\pp}_{0}}$ implements the solution of $\bP$ with initial condition $\vec{\pp}_{0}$.

Moreover, for all $t\in T$ and $i,j\in C$,
\begin{equation}
	\mm_{\vec{\pp}_{0}}( E_{j}(0) \wedge E_{i}(t)) = \pp_{0j} \mP_{ij}(t)
\end{equation}
hence whenever $\pp_{0j}>0$,
\begin{equation}\label{eq:joeredmeny}
	\mm_{\vec{\pp}_{0}}(E_{i}(t) |  E_{j}(0)) = \mP_{ij}(t)
\end{equation}

Now, by varying $\vec{\pp}_{0}$ one obtains the required stochastic process family $\bM: \vec{\pp}_{0} \mapsto \mm_{\vec{\pp}_{0}}$, since it follows from equation (\ref{eq:joeredmeny}) that $\bM$ is transition-constant. In particular, for positive $\vec{\pp}_{0}$ the common transition matrix \eqref{eq:common-M} satisfies $\mM(t)=\mP(t)$ for all $t\in T$.
\end{proof}

\markoviansatisfieschapmankolmogorov*

\begin{proof}\label{proof:markovian-satisfies-chapman-kolmogorov}
\noindent Let $t,t'\in T$ with $0\le t'\le t$. {Fix $k\in C$ with $\mm(E_k(0))>0$ and fix $i\in C$.
Let
\begin{equation}
J \doteq \left\{\, j\in C \;:\; \mm\big(E_j(t')\wedge E_k(0)\big)>0 \,\right\}
\end{equation}
By the law of total probability applied under the conditional measure $\mm(\cdot \mid E_k(0))$ (summing only over those $j$ for which the relevant conditionals are defined), we obtain
\begin{equation}
\mm\left(E_{i}\left(t\right)\mid E_{k}\left(0\right)\right)
=\sum_{j\in J}\mm\left(E_{i}\left(t\right)\mid E_{k}\left(0\right)\wedge E_{j}\left(t'\right)\right)\mm\left(E_{j}\left(t'\right)\mid E_{k}\left(0\right)\right)
\label{eq:law-of-total-prob-3}
\end{equation}
If the process is Markovian, then for every $j\in J$,
\begin{equation}
\mm\left(E_{i}\left(t\right)\mid E_{k}\left(0\right)\wedge E_{j}\left(t'\right)\right)
=\mm\left(E_{i}\left(t\right)\mid E_{j}\left(t'\right)\right)
\end{equation}
and therefore}
\begin{equation}
\mm\left(E_{i}\left(t\right)\mid E_{k}\left(0\right)\right)
=\sum_{j\in J}\mm\left(E_{i}\left(t\right)\mid E_{j}\left(t'\right)\right)\mm\left(E_{j}\left(t'\right)\mid E_{k}\left(0\right)\right)
\label{eq:law-total-prob-markovian}
\end{equation}
This is precisely the $ik$ entry of the Chapman--Kolmogorov equation \eqref{eq:transition-3} (whenever the conditional probabilities involved are well-defined).
\end{proof}

\decompmapslinear*

\begin{proof}\label{proof:decomp-maps-linear}
Let $t,t'\in T,\,t'\leq t$. Suppose that $\vec{\pp},\vec{q}\in\mathrm{Ran}\bP_{t'}$.
Then there are $\vec{\pp}_{0},\vec{q}_{0}\in\mathcal{S}_{N}$ such
that $\vec{\pp}=\bP_{t'}\left(\vec{\pp}_{0}\right)$ and $\vec{q}=\bP_{t'}\left(\vec{q}_{0}\right)$.
Let $\lambda\in\left[0,1\right]$ be arbitrary. Since $\bP_{t'}$
is convex-combination preserving, $\lambda\vec{\pp}+\left(1-\lambda\right)\vec{q}$
is also in $\mathrm{Ran}\bP_{t'}$. For this vector, we have: 
\begin{gather}
\bP_{t\leftarrow t'}\left(\lambda\vec{\pp}+\left(1-\lambda\right)\vec{q}\right)=\bP_{t\leftarrow t'}\left(\lambda\bP_{t'}\left(\vec{\pp}_{0}\right)+\left(1-\lambda\right)\bP_{t'}\left(\vec{q}_{0}\right)\right)\label{eq:convex-comb-pres-1-1}\\
=\bP_{t\leftarrow t'}\left(\bP_{t'}\left(\lambda\vec{\pp}_{0}+\left(1-\lambda\right)\vec{q}_{0}\right)\right)=\bP_{t}\left(\lambda\vec{\pp}_{0}+\left(1-\lambda\right)\vec{q}_{0}\right)\\
=\lambda\bP_{t}\left(\vec{\pp}_{0}\right)+\left(1-\lambda\right)\bP_{t}\left(\vec{q}_{0}\right)=\lambda\bP_{t\leftarrow t'}\left(\bP_{t'}\left(\vec{\pp}_{0}\right)\right)+\left(1-\lambda\right)\bP_{t\leftarrow t'}\left(\bP_{t'}\left(\vec{q}_{0}\right)\right)\\
=\lambda\bP_{t\leftarrow t'}\left(\vec{\pp}\right)+\left(1-\lambda\right)\bP_{t\leftarrow t'}\left(\vec{q}\right)\label{eq:convex-comb-pres-4-1}
\end{gather}
Here we have used the assumption that $\bP_{t'}$ and $\bP_{t}$ are
convex-combination preserving, as well as the definition of decomposability
(equation~\eqref{eq:decomp}). Comparing \eqref{eq:convex-comb-pres-1-1}
and \eqref{eq:convex-comb-pres-4-1} confirms that $\bP_{t\leftarrow t'}$
preserves the convex combination of vectors in its domain, $\mathrm{Ran}\bP_{t'}$.
\end{proof}

\linearextension*

\begin{proof}\label{proof:linear-extension}
Let $t,t'\in T,\,t'\leq t$. Suppose that the decomposing map $\bP_{t\leftarrow t'}$ can be extended to a convex-combination preserving map $\bar{\bP}_{t\leftarrow t'}:\mathcal{S}_{N}\rightarrow\mathcal{S}_{N}$.
Since any convex-combination preserving map on $\mathcal{S}_{N}$
uniquely extends to a linear map on $\mathbb{R}^{N}$ (see footnote~\ref{fn:extension}),
there exists a matrix $\mP\left(t\leftarrow t'\right)$ such that
\begin{equation}
\bP_{t\leftarrow t'}\left(\vec{\pp}\right)=\mP\left(t\leftarrow t'\right)\vec{\pp}\label{eq:matrix-repr}
\end{equation}
for all $\vec{\pp}\in\mathrm{Ran}\bP_{t'}$. Because $\bar{\bP}_{t\leftarrow t'}$
maps probability vectors to probability vectors, $\mP\left(t\leftarrow t'\right)$
must be stochastic. Then, for any $\vec{\pp}_{0}\in\mathcal{S}_{N}$,
\begin{equation}
\mP\left(t\right)\vec{\pp}_{0}=\bP_{t}\left(\vec{\pp}_{0}\right)=\bP_{t\leftarrow t'}\left(\bP_{t'}\left(\vec{\pp}_{0}\right)\right)=\mP\left(t\leftarrow t'\right)\mP\left(t'\right)\vec{\pp}_{0}\label{eq:divisib-derivation}
\end{equation}
where the last equality follows from \eqref{eq:matrix-repr}. Since
this holds for all $\vec{\pp}_{0}\in\mathcal{S}_{N}$, and $\mathcal{S}_{N}$
contains the standard basis vectors of $\mathbb{R}^{N}$, we obtain
the matrix equality
\begin{equation}
\mP\left(t\right)=\mP\left(t\leftarrow t'\right)\mP\left(t'\right)
\end{equation}
As this equation holds for all $t,t'\in T,\,t'\leq t$, it follows
that $\bP$ is divisible.

Conversely, suppose that $\bP$ is divisible. Then for any $t,t'\in T,\,t'\leq t$,
the action of stochastic matrix $\mP\left(t\leftarrow t'\right)$
defines a convex-combination preserving map $\bar{\bP}_{t\leftarrow t'}:\mathcal{S}_{N}\rightarrow\mathcal{S}_{N}$.
Its restriction $\bP_{t\leftarrow t'}$ to $\mathrm{Ran}\bP_{t'}$
satisfies
\begin{equation}
\bP_{t}\left(\vec{\pp}_{0}\right)=\mP\left(t\right)\vec{\pp}_{0}=\mP\left(t\leftarrow t'\right)\mP\left(t'\right)\vec{\pp}_{0}=\bP_{t\leftarrow t'}\left(\bP_{t'}\left(\vec{\pp}_{0}\right)\right)\label{eq:decomp-derivation}
\end{equation}
for all $\vec{\pp}_{0}\in\mathcal{S}_{N}$, where the last equality
holds because $\bP_{t\leftarrow t'}$ coincides with the action of
$\mP\left(t\leftarrow t'\right)$ on $\mathrm{Ran}\bP_{t'}$. Thus,
the family $\left\{ \bP_{t\leftarrow t'}\right\} _{t,t'\in T,\,t'\leq t}$
decomposes $\bP$. By Proposition~\ref{prop:decomp-characterization},
this family is unique, and therefore the maps $\bar{\bP}_{t\leftarrow t'}$
provide the required extensions.
\end{proof}

\Markovianimplementation*

\begin{proof}\label{proof:Markovian-implementation}
(a) Let $\vec{\pp}(t) : T \rightarrow \mathcal{S}_N$ be a probability vector trajectory. Define, for any $m \in \mathbb{N}$, for any $t_1 <  \dots < t_m \in T$, and for any subsets $C_1, \dots, C_m \subseteq C$, the function
\begin{equation}
	F_{t_1, \dots, t_m}(C_1,\dots,C_m) \doteq \prod_{k=1}^m \left( \sum_{i \in C_k} \pp_i(t_k) \right)
\end{equation}
where $\pp_i(t)$ is the $i$th entry of $\vec{\pp}(t)$. Given $F_{t_1, \dots, t_m}$ for $t_1 <  \dots < t_m \in T$, for any $\pi_1, \dots, \pi_m$ permutation of numbers $1, \dots, m$, further define
\begin{equation}
	F_{t_{\pi_1}, \dots, t_{\pi_m}}(C_1,\dots,C_m) \doteq F_{t_1, \dots, t_m}(C_{\pi^{-1}_1},\dots,C_{\pi^{-1}_m})
\end{equation}
We need to verify the following three conditions to be able to apply Kolmogorov's extension theorem to the so-defined function $F_{t_{\pi_1}, \dots, t_{\pi_m}}(C_1,\dots,C_m)$:
\begin{itemize}
	\item Non-negativity and $\sigma$-additivity: For fixed $t_1, \dots, t_m$ and all $k = 1,\dots,m$, the map $C \mapsto \sum_{i \in C} \pp_i(t_k)$ is a measure on $C$. Therefore, their product is a measure on $C^m$.
	\item Symmetry under permutations: guaranteed by definition.
	\item Marginalization:
	\begin{equation}
		F_{t_1,\dots,t_m}(C_1,\dots,C_{m-1}, C) = \prod_{k=1}^{m-1} \left( \sum_{i\in C_k} \pp_i(t_k) \right) \cdot \sum_{i\in C} \pp_i(t_m) = F_{t_1,\dots,t_{m-1}}(C_1,\dots,C_{m-1}) 
	\end{equation}
	since $\sum_{i \in C} \pp_i (t_m) =1$.
\end{itemize}
Since the three conditions are met, by Kolmogorov's extension theorem there exists a stochastic process $(\Omega, {\cal F}, X, \mm)$ whose finite dimensional distributions are given by:
\begin{equation}
	\mm \left( E_{i_1}(t_1) \wedge \dots \wedge E_{i_m} (t_m) \right) = F_{t_1, \dots, t_m} (\{ i_1 \}, \dots, \{ i_m \}) = \prod_{k=1}^m \pp_{i_k}(t_k)
\end{equation}
Since $\mm(E_i(t)) = F_{t} (\{ i \}) = \pp_i(t)$, this stochastic process implements $\vec{\pp}(t)$. Finally, this implementation is Markovian since
\begin{eqnarray}
	\mm \left( E_{i_{m+1}}(t_{m+1}) ~|~ E_{i_m}(t_m) \wedge \dots \wedge E_{i_1}(t_1) \right) = \frac{ {\displaystyle\prod_{k=1}^{m+1} \pp_{i_k}(t_k)}}{{\displaystyle\prod_{k=1}^{m} \pp_{i_k}(t_k)}} \\
	 = \frac{ \pp_{i_{m+1}} (t_{m+1}) \pp_{i_{m}}(t_{m})}{\pp_{i_{m}}(t_{m})} = \mm \left( E_{i_{m+1}}(t_{m+1}) ~|~ E_{i_m}(t_m) \right)
\end{eqnarray}

(b) First, we show that if a solution of a probability dynamics is non-degenerate then it has a non-Markovian implementation.  Let $\vec{\pp}(t) : T \rightarrow \mathcal{S}_N$ be a non-degenerate probability vector trajectory. Since $\vec{\pp}(t)$ is non-degenerate, there exists $t_{k_1} < t_{k_2} < t_{k_3} \in T$ and $i^1_{k_1}, i^2_{k_1}, i^1_{k_3}, i^2_{k_3} \in C$ such that $0< \pp_{i^1_{k_1}}(t_{k_1}), \pp_{i^2_{k_1}}(t_{k_1}) <1$ and $0< \pp_{i^1_{k_3}}(t_{k_3}),  \pp_{i^2_{k_3}}(t_{k_3}) <1$.

(b1) Furthermore, let us assume that there also exist $i^1_{k_2}, i^2_{k_2} \in C$ such that $0< \pp_{i^1_{k_2}}(t_{k_2}), \pp_{i^2_{k_2}}(t_{k_2}) <1$. In this case, for $j,l,n \in C$ define $K_{jln}$ as
\begin{equation}\label{Kjln}
K_{jln}  = 
\begin{cases}
	\pp_j(t_{k_1}) \pp_l(t_{k_2}) \pp_n(t_{k_3})				&	\text{if} ~ j \not\in \{ i^1_{k_1}, i^2_{k_1} \} ~\text{or}~ l \not\in \{ i^1_{k_2}, i^2_{k_2} \} ~\text{or}~ n \not\in \{ i^1_{k_3}, i^2_{k_3} \}	\\
	\pp_j(t_{k_1}) \pp_l(t_{k_2}) \pp_n(t_{k_3})	 + \epsilon		&	\text{if} ~ j = i^a_{k_1}, l = i^b_{k_2}, n = i^c_{k_3} ~\text{and}~ a+b+c ~ \text{is odd}	\\
	\pp_j(t_{k_1}) \pp_l(t_{k_2}) \pp_n(t_{k_3})	 - \epsilon		&	\text{if} ~ j = i^a_{k_1}, l = i^b_{k_2}, n = i^c_{k_3} ~\text{and}~ a+b+c ~ \text{is even}

\end{cases}
\end{equation}
where $a,b,c$ takes values from $\{ 1, 2 \}$ and $0<\epsilon <1$ is chosen so that $0 \leq K_{jln} \leq 1$ for all $j,l,n \in C$ (due to the assumptions, such an $\epsilon$ exists).

Note that 
\begin{eqnarray*}
	& & \sum_{j \in C} K_{jln}  = \\
	& & \left( \sum_{j \not\in \{ i^1_{k_1}, i^2_{k_1}\}} \pp_j(t_{k_1}) \pp_l(t_{k_2}) \pp_n(t_{k_3}) \right) + \pp_{i^1_{k_1}}(t_{k_1}) \pp_l(t_{k_2}) \pp_n(t_{k_3}) + \pp_{i^2_{k_1}}(t_{k_1}) \pp_l(t_{k_2}) \pp_n(t_{k_3}) + \epsilon - \epsilon \\ 
	& & = \pp_l(t_{k_2}) \pp_n(t_{k_3})
\end{eqnarray*}
similarly, $\sum_{l \in C} K_{jln} = \pp_j(t_{k_1}) \pp_n(t_{k_3})$ and $\sum_{n \in C} K_{jln} = \pp_j(t_{k_1}) \pp_l(t_{k_2})$. It follows that $\sum_{j, l, n \in C} K_{jln} = 1$. Thus, $K_{jln}$ yields a probability distribution on $C^3$ which marginalizes to product distributions on $C^2$, e.g.
\begin{equation}
\sum_{j \in C, ~l \in C_1, ~n \in C_2} K_{jln} = \sum_{l \in C_1, ~n \in C_2} \pp_l(t_{k_2}) \pp_n(t_{k_3}) = \left( \sum_{l \in C_1} \pp_l(t_{k_2}) \right) \cdot \left( \sum_{n \in C_2} \pp_n(t_{k_3}) \right)
\end{equation}

Analogously to the proof in (a) we now define a function $F_{t_1, \dots, t_m}(C_1,\dots,C_m)$ to which we then apply the Kolmogorov extension theorem. Let $m \in \mathbb{N}$, $t_1 < \dots < t_m \in T$, and let $C_1, \dots, C_m \subseteq C$. If $t_{k_1}, t_{k_2}, t_{k_3} \in \{ t_1, \dots, t_m \}$ then define $F_{t_1, \dots, t_m}(C_1,\dots,C_m)$ as
\begin{equation}\label{eq:Fujdef}
	F_{t_1, \dots, t_m}(C_1,\dots,C_m) \doteq \left( \prod_{k \in \{1,\dots,m\} \setminus  \{ k_1, k_2, k_3 \}} \left( \sum_{i \in C_k} \pp_i(t_k) \right) \right) \cdot \sum_{j \in C_{k_1}, ~l \in C_{k_2}, ~n \in C_{k_3}} K_{jln} 
\end{equation}
otherwise, define $F_{t_1, \dots, t_m}(C_1,\dots,C_m)$ as
\begin{equation}\label{eq:Feredetidef}
	F_{t_1, \dots, t_m}(C_1,\dots,C_m) \doteq \prod_{k=1}^m \left( \sum_{i \in C_k} \pp_i(t_k) \right)
\end{equation}
finally, given a so-defined $F_{t_1, \dots, t_m}$ for $t_1 <  \dots < t_m \in T$, for any $\pi_1, \dots, \pi_m$ permutation of numbers $1, \dots, m$, define
\begin{equation}
	F_{t_{\pi_1}, \dots, t_{\pi_m}}(C_1,\dots,C_m) \doteq F_{t_1, \dots, t_m}(C_{\pi^{-1}_1},\dots,C_{\pi^{-1}_m})
\end{equation}

We only need to verify the three conditions to be able to apply Kolmogorov's extension theorem when $F_{t_1, \dots, t_m}(C_1,\dots,C_m)$ takes the form \eqref{eq:Fujdef} (namely, when $t_{k_1}, t_{k_2}, t_{k_3} \in \{ t_1, \dots, t_m \}$; the case when $F_{t_1, \dots, t_m}(C_1,\dots,C_m)$ takes the form \eqref{eq:Feredetidef} has already been covered by the proof of part (a)).
\begin{itemize}
	\item Non-negativity and $\sigma$-additivity: for fixed $t_1, \dots, t_m$, the map $$(C_{k_1}, C_{k_2}, C_{k_3}) \mapsto {\displaystyle\sum_{j \in C_{k_1}, ~l \in C_{k_2}, ~n \in C_{k_3}}} K_{jln}$$ is a probability measure on $C^3$; the map $$(\dots , C_k, \dots) \mapsto {\displaystyle\prod_{k \in \{1,\dots,m\} \setminus  \{ k_1, k_2, k_3 \}}} \left( \sum_{i \in C_k} \pp_i(t_k) \right)$$ is a probability measure on $C^{m-3}$; hence their product is a probability measure on $C^m$.
	\item Symmetry under permutation: guaranteed by definition.
	\item Marginalization: for $k \in \{1,\dots,m\} \setminus  \{ k_1, k_2, k_3 \}$ the proof is analogous to the proof in (a). For $k \in \{ k_1, k_2, k_3 \}$ marginalization follows from the fact that $K_{jln}$ marginalizes to product distributions on $C^2$.
\end{itemize}

Since the three conditions are met, by Kolmogorov's extension theorem there exists a stochastic process $(\Omega, {\cal F}, X, \mm)$ whose finite dimensional distributions are given by
\begin{equation}
	\mm \left( E_{i_1}(t_1) \wedge \dots \wedge E_{i_m} (t_m) \right) = F_{t_1, \dots, t_m} (\{ i_1 \}, \dots, \{ i_m \})
\end{equation}
Since $\mm(E_i(t)) = F_{t} (\{ i \}) = \pp_i(t)$, this stochastic process implements $\vec{\pp}(t)$. Finally, this implementation is non-Markovian since
\begin{eqnarray}
	\mm \left( E_{i^1_{k_2}}(t_{k_2}) \right) & = & \pp_{i^1_{k_2}}(t_{k_2})	\\
	\mm \left( E_{i^1_{k_1}}(t_{k_1}) \wedge E_{i^1_{k_2}}(t_{k_2}) \right) & = & \pp_{i^1_{k_1}}(t_{k_1}) \pp_{i^1_{k_2}}(t_{k_2})	\\
	\mm \left( E_{i^1_{k_2}}(t_{k_2}) \wedge E_{i^1_{k_3}}(t_{k_3}) \right) & = & \pp_{i^1_{k_2}}(t_{k_2}) \pp_{i^1_{k_3}}(t_{k_3})	\\
	\mm \left( E_{i^1_{k_1}}(t_{k_1}) \wedge E_{i^1_{k_2}}(t_{k_2}) \wedge E_{i^1_{k_3}}(t_{k_3}) \right) & = & \pp_{i^1_{k_1}}(t_{k_1}) \pp_{i^1_{k_2}}(t_{k_2}) \pp_{i^1_{k_3}}(t_{k_3}) + \epsilon
\end{eqnarray}
and hence 
\begin{eqnarray}
	&  & \mm \left( E_{i^1_{k_3}}(t_{k_3}) ~|~ E_{i^1_{k_2}}(t_{k_2}) \wedge E_{i^1_{k_1}}(t_{k_1}) \right) = \frac{\pp_{i^1_{k_1}}(t_{k_1}) \pp_{i^1_{k_2}}(t_{k_2}) \pp_{i^1_{k_3}}(t_{k_3}) + \epsilon}{\pp_{i^1_{k_1}}(t_{k_1}) \pp_{i^1_{k_2}}(t_{k_2})} \neq	\\
	& & \frac{\pp_{i^1_{k_2}}(t_{k_2}) \pp_{i^1_{k_3}}(t_{k_3})}{\pp_{i^1_{k_2}}(t_{k_2})} = \mm \left( E_{i^1_{k_3}}(t_{k_3}) ~|~ E_{i^1_{k_2}}(t_{k_2}) \right) 
\end{eqnarray}

(b2) In case there does not exist $i^1_{k_2}, i^2_{k_2} \in C$ such that $0< \pp_{i^1_{k_2}}(t_{k_2}),  \pp_{i^2_{k_2}}(t_{k_2}) <1$, then let $i^1_{k_2}$ be the index for which $\pp_{i^1_{k_2}}(t_{k_2}) = 1$ (hence $\pp_{l}(t_{k_2})=0$ for all $l\neq i^1_{k_2}$). Define $K_{jln}$, instead of (\ref{Kjln}), as
\begin{equation}\label{Kjln-degenerate}
K_{jln}  = 
\begin{cases}
	0				&	\text{if} ~ l \neq i^1_{k_2}	\\
	\pp_j(t_{k_1}) \pp_n(t_{k_3})				&	\text{if} ~ l = i^1_{k_2} ~\text{and}~ \bigl( j \not\in \{ i^1_{k_1}, i^2_{k_1} \} ~\text{or}~ n \not\in \{ i^1_{k_3}, i^2_{k_3} \}\bigr)	\\
	\pp_j(t_{k_1}) \pp_n(t_{k_3})	 + \epsilon		&	\text{if} ~ j = i^a_{k_1}, l = i^1_{k_2}, n = i^c_{k_3} ~\text{and}~ a+1+c ~ \text{is odd}	\\
	\pp_j(t_{k_1}) \pp_n(t_{k_3})	 - \epsilon		&	\text{if} ~ j = i^a_{k_1}, l = i^1_{k_2}, n = i^c_{k_3} ~\text{and}~ a+1+c ~ \text{is even}
\end{cases}
\end{equation}
where $a,c \in\{1,2\}$ and $0<\epsilon<1$ is chosen so that $0\leq K_{jln}\leq 1$ for all $j,l,n\in C$ (such an $\epsilon$ exists because $0<\pp_{i^1_{k_1}}(t_{k_1}),\pp_{i^2_{k_1}}(t_{k_1})<1$ and $0<\pp_{i^1_{k_3}}(t_{k_3}),\pp_{i^2_{k_3}}(t_{k_3})<1$).

Then $K_{jln}$ is a probability distribution on $C^{3}$ with the following marginals:
\begin{itemize}
\item for the pairs involving $t_{k_2}$ one has
\begin{equation}
\sum_{n\in C}K_{jln}=\pp_j(t_{k_1})\pp_l(t_{k_2})
\qquad\text{and}\qquad
\sum_{j\in C}K_{jln}=\pp_l(t_{k_2})\pp_n(t_{k_3})
\end{equation}
since $\pp_l(t_{k_2})=0$ for $l\neq i^1_{k_2}$ and the $\pm\epsilon$ perturbations cancel in the relevant sums;
\item for the $(t_{k_1},t_{k_3})$-marginal define
\begin{equation}
J_{jn}\doteq \sum_{l\in C}K_{jln}=K_{j\,i^1_{k_2}\,n}
\end{equation}
which satisfies $\sum_{n\in C}J_{jn}=\pp_j(t_{k_1})$ and $\sum_{j\in C}J_{jn}=\pp_n(t_{k_3})$, but is \emph{not} the product distribution (because of the $\pm\epsilon$ terms).
\end{itemize}

Now repeat the above construction of the finite-dimensional distributions, with the following (minimal) modification of the definition of $F_{t_1,\dots,t_m}$: keep \eqref{eq:Fujdef} unchanged when $t_{k_1},t_{k_2},t_{k_3}\in\{t_1,\dots,t_m\}$, but when $t_{k_1},t_{k_3}\in\{t_1,\dots,t_m\}$ and $t_{k_2}\notin\{t_1,\dots,t_m\}$ (letting $k_1,k_3$ denote the positions of $t_{k_1},t_{k_3}$ in the ordered list $t_1<\cdots<t_m$) define instead
\begin{equation}\label{Fdegeneratepair}
	F_{t_1, \dots, t_m}(C_1,\dots,C_m) \doteq 
	\left( \prod_{k \in \{1,\dots,m\} \setminus  \{ k_1, k_3 \}} \left( \sum_{i \in C_k} \pp_i(t_k) \right) \right) \cdot \sum_{j \in C_{k_1}, ~n \in C_{k_3}} J_{jn}
\end{equation}
In all remaining cases (i.e., if neither $\{t_{k_1},t_{k_2},t_{k_3}\}\subseteq\{t_1,\dots,t_m\}$ nor $\{t_{k_1},t_{k_3}\}\subseteq\{t_1,\dots,t_m\}$) keep the product definition \eqref{eq:Feredetidef}.

With this modification, Kolmogorov consistency is verified similarly as before: the only additional marginalization to check is that marginalizing \eqref{eq:Fujdef} over $t_{k_2}$ yields (\ref{Fdegeneratepair}), which holds because $\sum_{l\in C}K_{jln}=J_{jn}$.

Finally, the resulting implementation is non-Markovian: since $\mm(E_{i^1_{k_2}}(t_{k_2}))=\pp_{i^1_{k_2}}(t_{k_2})=1$ and
\begin{eqnarray}
\mm\!\left(E_{i^1_{k_1}}(t_{k_1})\wedge E_{i^1_{k_2}}(t_{k_2})\right)&=&\pp_{i^1_{k_1}}(t_{k_1})\\
\mm\!\left(E_{i^1_{k_1}}(t_{k_1})\wedge E_{i^1_{k_2}}(t_{k_2})\wedge E_{i^1_{k_3}}(t_{k_3})\right)
&=&K_{i^1_{k_1}\,i^1_{k_2}\,i^1_{k_3}}
=\pp_{i^1_{k_1}}(t_{k_1})\pp_{i^1_{k_3}}(t_{k_3})+\epsilon
\end{eqnarray}
we obtain (whenever the conditional probabilities are well-defined)
\begin{align}
 \mm\!\left(E_{i^1_{k_3}}(t_{k_3})~\middle|~E_{i^1_{k_2}}(t_{k_2})\wedge E_{i^1_{k_1}}(t_{k_1})\right)
&=\frac{\pp_{i^1_{k_1}}(t_{k_1})\pp_{i^1_{k_3}}(t_{k_3})+\epsilon}{\pp_{i^1_{k_1}}(t_{k_1})} 
\neq \\  \pp_{i^1_{k_3}}(t_{k_3})
&=\mm\!\left(E_{i^1_{k_3}}(t_{k_3})~\middle|~E_{i^1_{k_2}}(t_{k_2})\right)
\end{align}
so the process is not Markovian.

\bigskip

Second, we show that a degenerate solution cannot have a non-Markovian implementation. There are three cases in which a solution $\vec{\pp}(t) = \bP(t, \vec{\pp}_0)$ of probability dynamics $\bP$ can be degenerate: either for all time $t\in T$ there is an $i_t \in C$ such that $\pp_{i_t}(t) = 1$; or there only exists one time $t\in T$ for which there is an $i \in C$ such that $0 < \pp_i(t) < 1$; or, there only exists two times $t_m < t_{m+1} \in T$ for which there are $i_m, i_{m+1} \in C$ such that $0 < \pp_{i_m}(t_m), \pp_{i_{m+1}}(t_{m+1}) < 1$, but there does not exist $t \in T$ such that $t_m < t < t_{m+1}$. One can easily verify for these three cases that the Markov condition (whenever it is meaningful) is always trivially satisfied. For instance, $\mm \left( E_{i_{m+1}}(t_{m+1}) ~|~ E_{i_m}(t_m) \wedge E_{i_{m-1}}(t_{m-1}) \wedge \dots \wedge E_{i_1}(t_1) \right)$ in the third case is only well-defined when $\mm(E_{i_{m-1}}(t_{m-1})) = \dots = \mm(E_{i_1}(t_1)) = 1$, but then $$\mm \left( E_{i_{m+1}}(t_{m+1}) ~|~ E_{i_m}(t_m) \wedge E_{i_{m-1}}(t_{m-1}) \wedge \dots \wedge E_{i_1}(t_1) \right) = \mm \left( E_{i_{m+1}}(t_{m+1}) ~|~ E_{i_m}(t_m) \right)$$ 

\end{proof}

\detdecompmarkov*

\begin{proof}\label{proof:det-decomp-markov}
(a) Suppose $D$ is decomposable. Then there is a family of maps $\left\{ D_{t\leftarrow t'}:\mathrm{Ran}D_{t'}\rightarrow C\right\} _{t,t'\in T,\,t'\leq t}$
decomposing $D$. For each decomposing map $D_{t\leftarrow t'}$,
introduce an extension $\bar{D}_{t\leftarrow t'}:C\rightarrow C$
to the whole of $C$, satisfying $\left.\bar{D}_{t\leftarrow t'}\right|_{\mathrm{Ran}D_{t'}}=D_{t\leftarrow t'}$;
the values of $\bar{D}_{t\leftarrow t'}$ on $C\setminus\mathrm{Ran}D_{t'}$
are chosen arbitrarily. We then define, for all $t,t'\in T,\,t'\leq t$,
the stochastic matrices: 
\begin{equation}
\left(\mP^{D}\left(t\leftarrow t'\right)\right)_{ij}\doteq\begin{cases}
1 & \textrm{if }\bar{D}_{t\leftarrow t'}\left(j\right)=i\\
0 & \textrm{if }\bar{D}_{t\leftarrow t'}\left(j\right)\neq i
\end{cases}\label{eq:division-matrix-proof}
\end{equation}
Clearly, 
\begin{equation}
\mP^{D}\left(t\right)=\mP^{D}\left(t\leftarrow t'\right)\mP^{D}\left(t'\right)\label{eq:composition-matrix-proof}
\end{equation}
for all $t,t'\in T,\,t'\leq t$. Hence, the family of convex combination
preserving maps 
\begin{equation}
\left\{ \left.\bP_{t\leftarrow t'}^{D}\right|_{\mathrm{Ran}\bP_{t'}}:\mathrm{Ran}\bP_{t'}\rightarrow\mathcal{S}_{N}\right\} _{t,t'\in T,\,t'\leq t}
\end{equation}
 determined by the stochastic matrices $\mP^{D}\left(t\leftarrow t'\right)$
decomposes probability dynamics $\bP^{D}$.

Conversely, suppose $D$ is not decomposable. This means that there
exist $t,t'\in T,\,t'<t$, and $i,j\in C$ such that $D\left(t',i\right)=D\left(t',j\right)$
but $D\left(t,i\right)\neq D\left(t,j\right)$. This entails that
$\bP^{D}\left(t',\vec{e}_{i}\right)=\bP^{D}\left(t',\vec{e}_{j}\right)$
but $\bP^{D}\left(t,\vec{e}_{i}\right)\neq\bP^{D}\left(t,\vec{e}_{j}\right)$;
in which case, due to Proposition~\ref{prop:decomp-characterization}~(i),
$\bP^{D}$ is not decomposable.

(b) Suppose $D$ is decomposable, with decomposing family of maps
\begin{equation}
\left\{ D_{t\leftarrow t'}:\mathrm{Ran}D_{t'}\rightarrow C\right\} _{t,t'\in T,\,t'\leq t}
\end{equation}
First we verify that 
\begin{equation}
\mm_{\vec{\pp}_{0}}^{D}\left(E_{i}\left(t\right)|E_{j}\left(t'\right)\right)=\begin{cases}
1 & \textrm{if }D_{t\leftarrow t'}\left(j\right)=i\\
0 & \textrm{if }D_{t\leftarrow t'}\left(j\right)\neq i
\end{cases}\label{eq:det-cond-prob3}
\end{equation}
for all $t,t'\in T,\,t'\leq t$, $i,j\in C$, with $\mm_{\vec{\pp}_{0}}^{D}\left(E_{j}\left(t'\right)\right)>0$,
and all $\vec{\pp}_{0}\in\mathcal{S}_{N}$. 
{Indeed, by Bayes' rule and the law of total probability we have 
\begin{equation}
\mm_{\vec{\pp}_{0}}^{D}\left(E_{i}\left(t\right)|E_{j}\left(t'\right)\right)
=\frac{\sum_{k=1}^{N}\mm_{\vec{\pp}_{0}}^{D}\left(E_{i}\left(t\right)\wedge E_{j}\left(t'\right)\wedge E_{k}\left(0\right)\right)}{\sum_{k=1}^{N}\mm_{\vec{\pp}_{0}}^{D}\left(E_{j}\left(t'\right)\wedge E_{k}\left(0\right)\right)}
\label{eq:det-cond-prob2}
\end{equation}
}

Using \eqref{eq:det-joint-prob} and the existence of maps $D_{t\leftarrow t'}$,
the terms in the numerator can be expressed as 
\begin{equation}
\mm_{\vec{\pp}_{0}}^{D}\left(E_{i}\left(t\right)\wedge E_{j}\left(t'\right)\wedge E_{k}\left(0\right)\right)=\begin{cases}
\mm_{\vec{\pp}_{0}}^{D}\left(E_{k}\left(0\right)\right) & \textrm{if }D\left(t',k\right)=j\textrm{ and }D_{t\leftarrow t'}\left(j\right)=i\\
0 & \textrm{otherwise}
\end{cases}
\end{equation}
{By \eqref{eq:det-joint-prob} and \eqref{eq:det-joint-prob0}, the terms in the denominator can be
written as 
\begin{equation}
\mm_{\vec{\pp}_{0}}^{D}\left(E_{j}\left(t'\right)\wedge E_{k}\left(0\right)\right)=\begin{cases}
\mm_{\vec{\pp}_{0}}^{D}\left(E_{k}\left(0\right)\right) & \textrm{if }D\left(t',k\right)=j\\
0 & \textrm{otherwise}
\end{cases}
\end{equation}
}
We have two cases: 
\begin{itemize}
\item if $D_{t\leftarrow t'}\left(j\right)=i$, then the numerator and denominator
of \eqref{eq:det-cond-prob2} coincide, hence $\mm_{\vec{\pp}_{0}}^{D}\left(E_{i}\left(t\right)|E_{j}\left(t'\right)\right)=1$, 
\item if $D_{t\leftarrow t'}\left(j\right)\neq i$, then the numerator of
\eqref{eq:det-cond-prob2} vanishes and $\mm_{\vec{\pp}_{0}}^{D}\left(E_{i}\left(t\right)|E_{j}\left(t'\right)\right)=0$. 
\end{itemize}
This verifies \eqref{eq:det-cond-prob3}.

Now, it is an elementary fact from probability theory that for any
probability measure $\mu$: if $\mm\left(A|B\right)\in\left\{ 1,0\right\} $,
then for any event $C$ (with $\mm\left(B\wedge C\right)>0$), it
holds that $\mm\left(A|B\wedge C\right)=\mm\left(A|B\right)$. Applying
this fact yields 
\begin{equation}
\mm_{\vec{\pp}_{0}}^{D}\left(E_{i_{m+1}}\left(t_{m+1}\right)|E_{i_{1}}\left(t_{1}\right)\wedge...\wedge E_{i_{m}}\left(t_{m}\right)\right)=\mm_{\vec{\pp}_{0}}^{D}\left(E_{i_{m+1}}\left(t_{m+1}\right)|E_{i_{m}}\left(t_{m}\right)\right)\label{eq:markov-1-1}
\end{equation}
for all $m\in\mathbb{N}$, $i_{1},...,i_{m},i_{m+1}\in C$, $t_{1}\leq t_{2}\leq...\leq t_{m}\leq t_{m+1}\in T$,
and $\vec{\pp}_{0}\in\mathcal{S}_{N}$ (whenever both sides are well-defined).
Hence the stochastic process family $\bM^{D}$ is Markovian.

Conversely, suppose $D$ is not decomposable. Then there exist $t,t'\in T,\,t'<t$,
and $i,j,k,l\in C$ such that $D\left(t',i\right)=D\left(t',j\right)=k$
and $D\left(t,i\right)=l\neq D\left(t,j\right)$. Choose $\vec{\pp}_{0}\in\mathcal{S}_{N}$
with $\pp_{0i}=\pp_{0j}=1/2$. Using equations \eqref{eq:det-joint-prob0}-\eqref{eq:det-cond-prob}
and \eqref{eq:det-cond-prob2}, one finds
\begin{align}
\mm_{\vec{\pp}_{0}}^{D}\left(E_{l}\left(t\right)|E_{k}\left(t'\right)\wedge E_{i}\left(0\right)\right) & =1\\
\mm_{\vec{\pp}_{0}}^{D}\left(E_{l}\left(t\right)|E_{k}\left(t'\right)\right) & =\frac{1}{2}
\end{align}
Thus the stochastic process family $\bM^{D}$ is not Markovian. 
\end{proof}

\thereexistsenvironment*

\begin{proof}\label{proof:thereexistsenvironment}
Let $M=N^{\tau}$ and identify the elements of $\Lambda=\left\{ 1,...,N^{\tau}\right\} $
with those of $C^{\tau}$, ordered arbitrarily. In this setting, a
configuration pair $\left(i,\alpha\right)\in C\times\Lambda$ corresponds
to a temporal trajectory of the target system through $C$, i.e. to
an $\omega\in\Omega$, where $i$ specifies the initial configuration
and $\alpha$ encodes the remainder of the trajectory.

With the identification $\omega=\left(i,\alpha\right)$, define the
deterministic system-ancilla dynamics by 
\begin{equation}
\SA\left(t,i,\alpha\right)\doteq X_{t}\left(\omega\right)
\end{equation}
for all $t\in T$, $i\in C$, $\alpha\in\Lambda$, and define the initial probability distributions by 
\begin{equation}
(\pi_{\vec{\pp}_{0}})_{i,\alpha}\doteq\mm_{\vec{\pp}_{0}}\left(\left\{ \omega\right\} \right)\label{eq:i-alpha-omega}
\end{equation}
for all $i\in C$, $\alpha\in\Lambda$ and $\vec{\pp}_{0}\in\mathcal{S}_{N}$.
Here, $(\pi_{\vec{\pp}_{0}})_{i,\alpha}$ denotes the entry of $\vec{\pi}_{\vec{\pp}_{0}}$
representing the probability that the target system is in configuration
$i\in C$ and the ancilla is in configuration $\alpha\in\Lambda$.
Then, from \eqref{eq:comp-joint-prob-rel2}, \eqref{eq:composite-measure2}, and \eqref{eq:i-alpha-omega} we have 
\begin{equation}
\mm_{\vec{\pp}_{0}}^{\SA}=\sum_{i=1}^{N}\sum_{\alpha=1}^{M}(\pi_{\vec{\pp}_{0}})_{i,\alpha}\mm_{i,\alpha}^{\SA}=\underset{\omega\in\Omega}{\sum}\mm_{\vec{\pp}_{0}}\left(\left\{ \omega\right\} \right)\delta_{\omega}=\mm_{\vec{\pp}_{0}}
\end{equation}
for all $\vec{\pp}_{0}\in\mathcal{S}_{N}$, where $\delta_{\omega}$
is the Dirac measure on $\Omega$ concentrated at $\omega$. Hence,
$\bM^{\SA}\left(\vec{\pp}_{0}\right)=\bM\left(\vec{\pp}_{0}\right)$
for all $\vec{\pp}_{0}\in\mathcal{S}_{N}$, and therefore $\bM^{\SA}=\bM$. 
\end{proof}

\thereexistsS*

\begin{proof}\label{proof:thereexists-S}
Fix $i\in C$ and consider the probability vector trajectory $t\mapsto\bP\left(t,\vec{e}_{i}\right)$.
By Proposition~\ref{prop:Markovian-implementation}, there exists
a stochastic process $\mm_{i}$ that implements this probability vector
trajectory. Since $\bP\left(0,\cdot\right)$ is the identity map on
$\mathcal{S}_{N}$, we have
\begin{equation}
\vec{\mm}_{i}\left(0\right)=\bP\left(0,\vec{e}_{i}\right)=\vec{e}_{i}
\end{equation}
Hence $\mm_{i}$ satisfies \eqref{eq:feature-mu-i-S}, the defining
condition of process $\mm_{i}^{S}$. We therefore define $\mm_{i}^{S}\doteq\mm_{i}$
for all $i=1,...,N$. 

Now let $\vec{\pp}_{0}\in\mathcal{S}_{N}$. Based on the definition
of the statistical mixture \eqref{eq:mixture-stochastic} and the
linearity of $\bP$, we obtain
\begin{equation}
\bP^{S}\left(t,\vec{\pp}_{0}\right)=\vec{\mm}_{\vec{\pp}_{0}}^{S}\left(t\right)=\sum_{i=1}^{N}\pp_{0i}\vec{\mm}_{i}\left(t\right)=\sum_{i=1}^{N}\pp_{0i}\bP\left(t,\vec{e}_{i}\right)=\bP\left(t,\vec{\pp}_{0}\right)
\end{equation}
for all $t\in T$. Hence $\bP^{S}=\bP$, as claimed.
\end{proof}

\thereexistsenvironmentS*

\begin{proof}\label{proof:thereexistsenvironment-S}
We adopt similar notation to the one introduced in the proof of Proposition~\ref{prop:thereexistsenvironment}.
Let $\Omega$ be the set of trajectories $\omega:T\to C$ with $T=\{0,1,\dots,\tau\}$. Set $M \doteq (N^\tau)^N$ and identify the elements of $\Lambda=\{1,\dots,M\}$ with those of $(C^\tau)^N$,
ordered arbitrarily. Thus each $\alpha\in\Lambda$ can be written as
\begin{equation}
\alpha = \big(\alpha^{(1)},\dots,\alpha^{(N)}\big),
\qquad \alpha^{(r)}\in C^\tau \ \ (r=1,\dots,N)
\end{equation}
where $\alpha^{(r)}$ encodes a sequence of configurations of length $\tau$ for times $1,\dots,\tau$.

For each $i\in C$ and $\alpha\in\Lambda$, define the trajectory $\omega_{i,\alpha}\in\Omega$ by
\begin{equation}
\omega_{i,\alpha}(0)\doteq i,
\qquad
\omega_{i,\alpha}(t)\doteq \alpha^{(i)}(t)\ \ \text{for }t=1,\dots,\tau
\end{equation}
(where we view $\alpha^{(i)}$ as a function on $\{1,\dots,\tau\}$ in the natural way).
Define the deterministic system-ancilla dynamics by
\begin{equation}
\SA(t,i,\alpha)\doteq X_t(\omega_{i,\alpha})=\omega_{i,\alpha}(t)=\alpha^{(i)}(t)\label{eq:envSbizonyitas-SA}
\end{equation}
In particular, $\SA(0,\cdot,\alpha)$ is the identity on $C$ for every $\alpha\in\Lambda$.

Next, define a probability distribution $\vec{\lambda}_0\in\mathcal{S}_M$ as follows.
For each $r\in C$ and each $\beta\in C^\tau$, let $\omega_{r,\beta}\in\Omega$ denote the
trajectory with $\omega_{r,\beta}(0)=r$ and $(\omega_{r,\beta}(1),\dots,\omega_{r,\beta}(\tau))=\beta$,
and set
\begin{equation}
q_r(\beta)\doteq \mm_r^S(\{\omega_{r,\beta}\})\label{eq:envSbizonyitas-qr}
\end{equation}
(These $q_r$ satisfy 
\begin{equation}
\sum_{\beta\in C^\tau} q_r(\beta)=1\label{eq:envSbizonyitas-qr-sum}
\end{equation}
since $\mm_r^S$ is a probability measure
and is supported on trajectories starting at $r$ at time $0$.)
Now define, for $\alpha=(\alpha^{(1)},\dots,\alpha^{(N)})\in\Lambda$,
\begin{equation}
\lambda_{0\alpha}\doteq \prod_{r=1}^N q_r(\alpha^{(r)})\label{eq:envSbizonyitas-lambda}
\end{equation}
Then $\sum_{\alpha\in\Lambda}\lambda_{0\alpha}=1$, hence $\vec{\lambda}_0\in\mathcal{S}_M$.

We now verify (i): $\mm_i^S=\sum_{\alpha=1}^M \lambda_{0\alpha}\mm_{i,\alpha}^{\SA}$ for all $i\in C$. Fix $i\in C$ and $\beta\in C^\tau$. For $\omega_{i,\beta}$ we then have
\begin{align}
\sum_{\alpha\in\Lambda}\lambda_{0\alpha}\mm_{i,\alpha}^{\SA}(\{\omega_{i,\beta}\})
&=\sum_{\alpha\in\Lambda}\lambda_{0\alpha}\,\delta_{\omega_{i,\alpha}}(\{\omega_{i,\beta}\}) 
=\sum_{\substack{\alpha\in\Lambda:\\ \omega_{i,\alpha}=\omega_{i,\beta}}}\lambda_{0\alpha}
=\sum_{\substack{\alpha\in\Lambda:\\ \alpha^{(i)}=\beta}}\lambda_{0\alpha} \\
&=\sum_{\substack{\alpha\in\Lambda:\\ \alpha^{(i)}=\beta}}
\prod_{r=1}^N q_r(\alpha^{(r)}) \\
&= q_i(\beta)\!\!\sum_{\alpha^{(1)}\in C^\tau}\cdots\!\!\sum_{\alpha^{(i-1)}\in C^\tau}
\sum_{\alpha^{(i+1)}\in C^\tau}\cdots\!\!\sum_{\alpha^{(N)}\in C^\tau}
\prod_{r\neq i} q_r(\alpha^{(r)}) \\
&= q_i(\beta)\prod_{r\neq i}\left(\sum_{\gamma\in C^\tau} q_r(\gamma)\right) = q_i(\beta)\cdot 1 =\mm_i^S(\{\omega_{i,\beta}\})
\end{align}
where the first equality uses \eqref{eq:comp-joint-prob-rel2} together with \eqref{eq:envSbizonyitas-SA} ($\mm_{i,\alpha}^{\SA}$ is the Dirac measure concentrated at $\omega_{i,\alpha}$); the third equality uses that $\omega_{i,\alpha}=\omega_{i,\beta}$ is equivalent to $\alpha^{(i)}=\beta$ by the definition of $\omega_{i,\alpha}$; the fourth equality uses \eqref{eq:envSbizonyitas-lambda}; the seventh equality uses \eqref{eq:envSbizonyitas-qr-sum}; and the last equality uses \eqref{eq:envSbizonyitas-qr}. Hence $\sum_{\alpha}\lambda_{0\alpha}\mm_{i,\alpha}^{\SA}=\mm_i^S$ for all $i \in C$, proving (i).

For (ii), define the family $\{\vec{\pi}_{\vec{\pp}_0}\}_{\vec{\pp}_0\in\mathcal{S}_N}$ by the independence condition
\begin{equation}
(\pi_{\vec{\pp}_0})_{i,\alpha}=\pp_{0i}\lambda_{0\alpha}
\qquad
\end{equation}
for all $i\in C,\ \alpha\in\Lambda$. Then, using \eqref{eq:composite-measure2} and \eqref{eq:comp-measure2},
\begin{align}
\mm_{\vec{\pp}_0}^{\SA}
&=\sum_{i=1}^N\sum_{\alpha=1}^M \pp_{0i}\lambda_{0\alpha}\mm_{i,\alpha}^{\SA}
=\sum_{i=1}^N \pp_{0i}\left(\sum_{\alpha=1}^M \lambda_{0\alpha}\mm_{i,\alpha}^{\SA}\right)
=\sum_{i=1}^N \pp_{0i}\mm_i^S
=\mm_{\vec{\pp}_0}^S
\end{align}
Therefore $\bM^{\SA}=\bM^S$.
\end{proof}

\section{Appendix}\label{appendix:innerimplementation}

As throughout Sections \ref{sec:twodescriptions}--\ref{sec:quantum-dynamics}, let $C = \{1, \dots, N\}$ be a set of configurations, $0 \in T\subseteq\mathbb{R}$ a set of time indices, $\Omega$ the set of all $T\rightarrow C$ functions, $\mathcal{F}$ the cylinder $\sigma$-algebra of $\Omega$, ${\cal M}$ the set of probability measures (the set of canonical stochastic processes) $\mm$ over $(\Omega, \mathcal{F})$, $\bP:T\times\mathcal{S}_{N}\rightarrow\mathcal{S}_{N}$ a probability dynamics, and $\bM: \mathcal{S}_{N}\rightarrow{\cal M}, ~\vec{\pp}_{0}\mapsto\mm_{\vec{\pp}_{0}}$ a (canonical) stochastic process family. Suppose that $\bM$ implements $\bP$. Starting from the stochastic process \emph{family} $\bM$ we now define a \emph{single}, non-canonical stochastic process $\tilde{\bM}$ which, in a precise sense, contains the same information about probability dynamics $\bP$ as does $\bM$.

Thus, let $\tilde{\Omega} \doteq  \Omega \times \mathcal{S}_{N}$, $\tilde{\mathcal{F}} \doteq \mathcal{F} \otimes {\cal B}(\mathcal{S}_{N})$, let $\lambda$ be the normalized Lebesgue measure on $(\mathcal{S}_{N}, {\cal B}(\mathcal{S}_{N}))$, and let $\tilde{\mm}_{\vec{\pp}}$ be the probability measures on $(\Omega, \mathcal{F})$ defined as $\tilde{\mm}_{\vec{\pp}} \doteq  \mm_{\vec{\pp}}$ for each $\vec{\pp} \in \mathcal{S}_{N}$. For any $\tilde{A} \in \tilde{\mathcal{F}}$ and for any $\vec{\pp} \in \mathcal{S}_{N}$ define $A_{\vec{\pp}}\in \mathcal{F}$ as
\begin{equation}
	A_{\vec{\pp}}  \doteq \{ \tau \in \Omega ~|~ (\tau, \vec{\pp}) \in \tilde{A} \}
\end{equation}
and define the probability measure $\tilde{\mm}$ on $(\tilde{\Omega}, \tilde{\mathcal{F}})$ as
\begin{equation}\label{eq:probmeasureinner}
	\tilde{\mm}(\tilde{A}) = \int_{\mathcal{S}_{N}} \tilde{\mm}_{\vec{\pp}} (A_{\vec{\pp}}) ~ d\lambda(\vec{\pp})
\end{equation}
Finally, let $\tilde{X} : T \times \tilde{\Omega} \rightarrow C$ be the random variable defined as
\begin{equation}
	\tilde{X}(t, (\tau, \vec{\pp})) \doteq \tau(t)
\end{equation}
We arrived at a single, non-canonical stochastic process $\tilde{\bM} = (\tilde{\Omega}, \tilde{\mathcal{F}}, \tilde{X}, \tilde{\mm})$ which we may call as the {\em uniform inner representation of the (canonical) stochastic process family $\bM$}.

Define the sets 
\begin{equation}
	\tilde{E}_i(t) \doteq \{ \tau \in \Omega  ~|~ \tau(t) = i \}
\end{equation}
and the vector
\begin{equation}
	\vec{\tilde{\mm}}_{\vec{\pp}} (t) \doteq ( \tilde{\mm}_{\vec{\pp}}(\tilde{E}_1(t)), \dots, \tilde{\mm}_{\vec{\pp}}(\tilde{E}_N(t)))
\end{equation}
for all $t \in T$.
We say that stochastic process $(\tilde{\Omega}, \tilde{\mathcal{F}}, \tilde{X}, \tilde{\mm})$ \emph{inner-implements} probability dynamics $\bP$ iff, for all $t \in T$ and for all $\vec{\pp}_0 \in \mathcal{S}_{N}$, we have
\begin{equation}
	\vec{\tilde{\mm}}_{\vec{\pp}_0}(t)  = \bP(t, \vec{\pp}_0)
\end{equation}

From the previous construction it is easy to see the following:

\begin{prop}\label{prop:innerimplements}
Let $\bP$ be a probability dynamics, $\bM$ a (canonical) stochastic process family, and let the stochastic process $\tilde{\bM}$ be the uniform inner representation of $\bM$. Then $\bM$ implements $\bP$ if and only if $\tilde{\bM}$ inner-implements $\bP$.
\end{prop}

Combined with Proposition \ref{prop:implementsalways}, which asserts that a probability dynamics can always be implemented by a stochastic process family, Proposition \ref{prop:innerimplements} shows a sense in which a probability dynamics can also be reproduced by a single stochastic process on an enlarged sample space. In some cases such a single stochastic process may provide a natural probabilistic description of a physical system. However, it is important to distinguish this construction from a system--ancilla implementation in which an ancilla is prepared independently in a fixed state and then marginalised: in the inner representation, the relevant system statistics are recovered \emph{conditional on} the value of an additional parameter, rather than by averaging it out.

For example, the uniform inner representation of the stochastic process family implementing the probability dynamics of the coin toss example of the Introduction can be described as follows. The coin is first ``shaken'', which we model as the selection of an initial parameter $\vec{\pp}\in\mathcal{S}_{2}$ according to the uniform (Lebesgue) measure on $\mathcal{S}_{2}$. This parameter $\vec{\pp}$ labels the corresponding member of the implementing family and thereby determines a conditional probability measure $\mm_{\vec{\pp}}$ on configuration trajectories $C^{3}$ via equations (\ref{eq:mur1})--(\ref{eq:mur8}). The so-prepared coin is then tossed three times, changing the position of its inner weight according to (\ref{eq_probdyn1})--(\ref{eq_probdyn3}). The sample space of the uniform inner representation is $C^{3}\times\mathcal{S}_{2}$ and its probability measure $\mm$ on $C^{3}\times\mathcal{S}_{2}$ is generated from the conditional probability measures $\mm_{\vec{\pp}}$ on $C^{3}$ for $\vec{\pp}\in\mathcal{S}_{2}$ via equation (\ref{eq:probmeasureinner}). Crucially, the probability dynamics of interest is obtained from this single measure by conditioning on the selected value of $\vec{\pp}$ (i.e.\ by working with $\mm_{\vec{\pp}}$), not by taking the unconditional marginal on $C^{3}$, which would correspond to averaging over $\vec{\pp}$ and in general would not reproduce the original (possibly nonlinear) dependence on the initial parameter.

It is important to emphasize that, while the state space of each stochastic process of a (canonical) stochastic process family agrees with the set of configurations $C$ of the probability dynamics, and while their sample spaces consist of the configuration trajectories $\Omega$, the sample space of the inner representation is inflated to $\Omega \times \mathcal{S}_N$. This inflation of the sample space allows the inner representation to contain the same set of probability measures on $(\Omega, \mathcal{F})$ (as conditional probabilities parametrized with $\vec{\pp} \in \mathcal{S}_N$) which a (canonical) stochastic process family explicitly expresses. In general, Proposition \ref{prop:innerimplements} is due to the mathematical concept of a stochastic process being flexible enough to represent any sort of information. However, Proposition \ref{prop:innerimplements} does not change the key conceptual point that a {\em set} of probability measures for the configuration trajectories is needed to implement or inner-implement a probability dynamics, and a single probability measure for the configuration trajectories is not sufficient to do so. 

In order to put emphasis on this conceptual point, and despite that in some cases it may be natural to probabilistically describe a physical system with a single stochastic process whose sample space differs from the possible configuration trajectories or whose state space differs from the possible configurations, the article focuses on canonical stochastic processes and (canonical) stochastic process families.

\section{Appendix }\label{appendix:Barandes-interference}

To recover the formalism of quantum mechanics, \cite{barandes2025} considers
what we would describe as a linear and indivisible probability dynamics
$\bP$, with the special property that the matrices $\mP\left(t\right)$
are \emph{unistochastic}. This means that there exists a family $\left\{ U\left(t\right)\right\} _{t\in T}$
of unitary matrices such that
\begin{equation}
\mP\left(t\right)=\overline{U\left(t\right)}\odot U\left(t\right)
\end{equation}
for all $t\in T.$ Here the overline denotes complex conjugation and
$\odot$ is the Schur-Hadamard product, defined for arbitrary $N\times N$
matrices $X$ and $Y$ as entry-wise multiplication: $\left(X\odot Y\right){}_{ij}\equiv X_{ij}Y_{ij}$.
In the interpretation proposed in \cite{barandes2025}, $U\left(t\right)$ plays the role of the unitary time-evolution operator familiar from quantum theory---here
emerging from the description of a generic ``indivisible stochastic
process.'' The author writes:
\begin{quote}
There is an alternative---and far-reaching---way to understand the
generic indivisibility of a time-dependent transition matrix $\mP\left(t\right)$.
To this end, suppose that $\mP\left(t\right)$ happens to be unistochastic,
with unitary time-evolution operator $U\left(t\right)$. Then, for
any two times $t$ and $t'$, one can define a \emph{relative} time-evolution
operator
\begin{equation}
U\left(t\leftarrow t'\right)\equiv U(t)U^{-1}\left(t'\right)
\end{equation}
which is guaranteed to be unitary and which yields the composition
law
\begin{equation}
U(t)=U\left(t\leftarrow t'\right)U\left(t'\right)
\end{equation}
Note that this composition law does not extend to the transition matrix
$\mP\left(t\right)$ due to cross terms.

With
\begin{equation}
\mP_{kj}\left(t'\right)\equiv\left|U_{kj}\left(t'\right)\right|^{2}
\end{equation}
defined as usual, and defining
\begin{equation}
\mP_{ik}\left(t\leftarrow t'\right)\equiv\left|U_{ik}\left(t\leftarrow t'\right)\right|^{2}
\end{equation}
{[}...{]} one sees that the discrepancy between the true transition
matrix $\mP\left(t\right)$ and its would-be division $\mP\left(t\leftarrow t'\right)\mP\left(t'\right)$
is given by
\begin{equation}
\mP_{ij}\left(t\right)-\left[\mP\left(t\leftarrow t'\right)\mP\left(t'\right)\right]_{ij}=\underset{k\neq l}{\sum}\overline{U_{ik}\left(t\leftarrow t'\right)\left(U\left(t'\right)\vec{e}_{j}\right)_{k}}U_{il}\left(t\leftarrow t'\right)\left(U\left(t'\right)\vec{e}_{j}\right)_{l}\label{eq:discrepancy}
\end{equation}

{[}...{]} Remarkably, the right-hand side of \eqref{eq:discrepancy}
gives the general mathematical formula for quantum interference, despite
the absence of manifestly quantum-theoretic assumptions. One sees
from this analysis that interference is a direct consequence of the
stochastic dynamics not generally being divisible. More precisely,
interference is nothing more than a generic discrepancy between the
\emph{actual} indivisible stochastic dynamics and a \emph{heuristic-approximate}
divisible stochastic dynamics. (\cite{barandes2025}, p.~12, emphasis
in original, with notation adjusted to present purposes) 
\end{quote}
Aside from the methodological concern about whether one can identify
``the general mathematical formula for quantum interference'' solely
on the basis of a formal similarity between expressions with ``cross
terms,'' there is a concrete mathematical problem with this interpretation.
The problem is that the ``discrepancy'' between $\mP\left(t\right)$
and $\mP\left(t\leftarrow t'\right)\mP\left(t'\right)$, as described
above, is unrelated to indivisibility. Notice that the comparison is effectively between $\mP\left(t\right)$ and $\left(\overline{U\left(t\leftarrow t'\right)}\odot U\left(t\leftarrow t'\right)\right)\mP\left(t'\right)$. And these two matrices do not generally coincide even when the probability
dynamics \emph{is} divisible. The reason is that $\mP\left(t\leftarrow t'\right)$,
defined via the condition that the probability dynamics is divisible,
does not in general equal $\overline{U\left(t\leftarrow t'\right)}\odot U\left(t\leftarrow t'\right)$.
This is a consequence of the fact that the Schur-Hadamard product,
on the one hand, and matrix multiplication and inversion, on the other,
are not commutative operations. More precisely:

\noindent
\begin{equation}
\overline{U\left(t\right)U\left(t'\right)^{-1}}\odot U\left(t\right)U\left(t'\right)^{-1}\neq\left(\overline{U\left(t\right)}\odot U\left(t\right)\right)\left(\overline{U\left(t'\right)}\odot U\left(t'\right)\right)^{-1}\label{eq:not-commute}
\end{equation}
assuming that the inverse on the right-hand side exists. In other
words, the following diagram does not commute:\[
\begin{tikzcd}[column sep=large, row sep=large]
    {U(t),\,U(t')} 
        \arrow[rrr, "{\begin{array}{c}\text{dynamics} \\[-2pt] \text{decomposition}\end{array}}"] 
        \arrow[d, "{\begin{array}{c}\text{Schur-Hadamard} \\[-2pt] \text{square}\end{array}}"', anchor=center] 
    & & & 
    {U(t)U(t')^{-1}} 
        \arrow[d, "{\begin{array}{c}\text{Schur-Hadamard} \\[-2pt] \text{square}\end{array}}", style={draw=red}, anchor=center] \\
    {\mathbb{P}(t),\,\mathbb{P}(t')} 
        \arrow[rrr, "{\begin{array}{c}\text{dynamics} \\[-2pt] \text{decomposition}\end{array}}"'] 
    & & & 
    {\mathbb{P}(t)\,\mathbb{P}(t')^{-1}}
\end{tikzcd}
\]On the left-hand side of \eqref{eq:not-commute} we have $\overline{U\left(t\leftarrow t'\right)}\odot U\left(t\leftarrow t'\right)$;
on the right-hand side, we have $\mP\left(t\right)\mP\left(t'\right)^{-1}$.
The two expressions diverge even in cases where---and, indeed, regardless
of whether---$\mP\left(t\right)\mP\left(t'\right)^{-1}$ happens
to be a stochastic matrix, that is, when the dynamics is divisible.
The following example illustrates this point.
\begin{example}
Let $N=2,T=\left\{ 0,1,2\right\} $, and consider the linear probability
dynamics $\bP$ defined by:
\begin{gather}
\mP\left(1\right)=\left(\begin{array}{cc}
\cos^{2}\frac{\pi}{8} & \sin^{2}\frac{\pi}{8}\\
\sin^{2}\frac{\pi}{8} & \cos^{2}\frac{\pi}{8}
\end{array}\right),\,\,\,\,\,\,\,\,\,\,\mP\left(2\right)=\left(\begin{array}{cc}
\cos^{2}\frac{3\pi}{8} & \sin^{2}\frac{3\pi}{8}\\
\sin^{2}\frac{3\pi}{8} & \cos^{2}\frac{3\pi}{8}
\end{array}\right)
\end{gather}
These matrices are unistochastic, since each is the Schur--Hadamard
square of the corresponding unitary (rotation) matrix:
\begin{gather}
U\left(1\right)=\left(\begin{array}{cc}
\cos\frac{\pi}{8} & -\sin\frac{\pi}{8}\\
\sin\frac{\pi}{8} & \cos\frac{\pi}{8}
\end{array}\right),\,\,\,\,\,\,\,\,\,\,U\left(2\right)=\left(\begin{array}{cc}
\cos\frac{3\pi}{8} & -\sin\frac{3\pi}{8}\\
\sin\frac{3\pi}{8} & \cos\frac{3\pi}{8}
\end{array}\right)
\end{gather}
It is straightforward to verify that $\mP\left(1\right)$ is invertible
and
\begin{align}
\mP\left(1\right)^{-1} & =\sqrt{2}\left(\begin{array}{cc}
\cos^{2}\frac{\pi}{8} & -\sin^{2}\frac{\pi}{8}\\
-\sin^{2}\frac{\pi}{8} & \cos^{2}\frac{\pi}{8}
\end{array}\right)\\
\mP\left(2\right)\mP\left(1\right)^{-1} & =\left(\begin{array}{cc}
0 & 1\\
1 & 0
\end{array}\right)\label{eq:rotation-decomp}
\end{align}
Since $\mP\left(2\right)\mP\left(1\right)^{-1}$ is a stochastic matrix,
$\bP$ is a divisible probability dynamics, with $\mP\left(2\leftarrow1\right)\doteq\mP\left(2\right)\mP\left(1\right)^{-1}$.

On the other hand, 
\begin{equation}
U\left(2\leftarrow1\right)\doteq U(2)U^{-1}\left(1\right)=\left(\begin{array}{cc}
\cos\frac{\pi}{4} & -\sin\frac{\pi}{4}\\
\sin\frac{\pi}{4} & \cos\frac{\pi}{4}
\end{array}\right)
\end{equation}
and its Schur-Hadamard square is
\begin{equation}
\overline{U\left(2\leftarrow1\right)}\odot U\left(2\leftarrow1\right)=\left(\begin{array}{cc}
\frac{1}{2} & \frac{1}{2}\\
\frac{1}{2} & \frac{1}{2}
\end{array}\right)
\end{equation}
which is not equal to $\mP\left(2\leftarrow1\right)$, as given by
\eqref{eq:rotation-decomp}.
\end{example}

\bibliographystyle{plainnat}
\bibliography{main}
\end{document}